\documentclass{LMCS}

\def\doi{8 (2:14) 2012}
\lmcsheading%
{\doi}
{1--35}
{}
{}
{Sep.~23, 2011}
{Jun.~20, 2012}
{}

\usepackage{isabelle}
\usepackage{isabellesym}
\usepackage{amsmath}
\usepackage{amssymb}
\usepackage{tikz}
\usepackage{pgf}
\usepackage{ot1patch}
\usepackage{times}
\usepackage{boxedminipage}
\usepackage{proof}
\usepackage{setspace}
\usepackage{afterpage}
\usepackage{pdfsetup}

\allowdisplaybreaks
\isabellestyle{it}
\renewcommand{\isastyleminor}{\it}%
\renewcommand{\isastyle}{\normalsize\it}%

\DeclareRobustCommand{\flqq}{\mbox{\guillemotleft}}
\DeclareRobustCommand{\frqq}{\mbox{\guillemotright}}
\renewcommand{\isacharunderscore}{\mbox{$\_\!\_$}}
\renewcommand{\isasymbullet}{{\raisebox{-0.4mm}{\Large$\boldsymbol{\hspace{-0.5mm}\cdot\hspace{-0.5mm}}$}}}
\def\dn{\,\stackrel{\mbox{\scriptsize def}}{=}\,}
\renewcommand{\isasymequiv}{$\dn$}

\renewcommand{\isasymemptyset}{$\varnothing$}
\newcommand{\isasymnotapprox}{$\not\approx$}
\newcommand{\isasymLET}{$\mathtt{let}$}
\newcommand{\isasymAND}{$\mathtt{and}$}
\newcommand{\isasymIN}{$\mathtt{in}$}
\newcommand{\isasymEND}{$\mathtt{end}$}

\newcommand{\isasymANIL}{$\mathtt{anil}$}
\newcommand{\isasymACONS}{$\mathtt{acons}$}
\newcommand{\isasymCASE}{$\mathtt{case}$}
\newcommand{\isasymOF}{$\mathtt{of}$}
\newcommand{\isasymAL}{\makebox[0mm][l]{$^\alpha$}}
\newcommand{\isasymPRIME}{\makebox[0mm][l]{$'$}}

\newenvironment{proof-of}[1]{{\em Proof of #1:}}{}

\begin{document}

\title[Genral Bindings]{General Bindings and Alpha-Equivalence in Nominal
Isabelle\rsuper*}
\author[C.~Urban]{Christian Urban\rsuper a} 
\address{{\lsuper a}King's College London, United Kingdom}	
\email{christian.urban@kcl.ac.uk}

\author[C.~Kaliszyk]{Cezary Kaliszyk\rsuper b}
\address{{\lsuper b}University of Innsbruck, Austria}
\email{cezary.kaliszyk@uibk.ac.at}
\thanks{{\lsuper*}This is a revised and expanded version of~\cite{UrbanKaliszyk11}}


\begin{abstract} 
Nominal Isabelle is a definitional extension of the Isabelle/HOL theorem
prover. It provides a proving infrastructure for reasoning about
programming language calculi involving named bound variables (as
opposed to de-Bruijn indices). In this paper we present an extension of
Nominal Isabelle for dealing with general bindings, that means
term constructors where multiple variables are bound at once. Such general
bindings are ubiquitous in programming language research and only very
poorly supported with single binders, such as lambda-abstractions. Our
extension includes new definitions of alpha-equivalence and establishes
automatically the reasoning infrastructure for alpha-equated terms. We
also prove strong induction principles that have the usual variable
convention already built in.
\end{abstract}

\maketitle
\begin{isabellebody}%
\def\isabellecontext{Paper}%
\isadelimtheory
\endisadelimtheory
\isatagtheory
\endisatagtheory
{\isafoldtheory}%
\isadelimtheory
\endisadelimtheory
\isamarkupsection{Introduction%
}
\isamarkuptrue%
\begin{isamarkuptext}%
So far, Nominal Isabelle provided a mechanism for constructing alpha-equated
  terms, for example lambda-terms

  \[
  \isa{t\ {\isaliteral{3A}{\isacharcolon}}{\isaliteral{3A}{\isacharcolon}}{\isaliteral{3D}{\isacharequal}}\ x\ {\isaliteral{7C}{\isacharbar}}\ t\ t\ {\isaliteral{7C}{\isacharbar}}\ {\isaliteral{5C3C6C616D6264613E}{\isasymlambda}}x{\isaliteral{2E}{\isachardot}}\ t}
  \]\smallskip

  \noindent
  where free and bound variables have names.  For such alpha-equated terms,
  Nominal Isabelle derives automatically a reasoning infrastructure that has
  been used successfully in formalisations of an equivalence checking
  algorithm for LF \cite{UrbanCheneyBerghofer08}, Typed
  Scheme~\cite{TobinHochstadtFelleisen08}, several calculi for concurrency
  \cite{BengtsonParow09} and a strong normalisation result for cut-elimination
  in classical logic \cite{UrbanZhu08}. It has also been used by Pollack for
  formalisations in the locally-nameless approach to binding
  \cite{SatoPollack10}.

  However, Nominal Isabelle has fared less well in a formalisation of the
  algorithm W \cite{UrbanNipkow09}, where types and type-schemes are,
  respectively, of the form

  \begin{equation}\label{tysch}
  \begin{array}{l}
  \isa{T\ {\isaliteral{3A}{\isacharcolon}}{\isaliteral{3A}{\isacharcolon}}{\isaliteral{3D}{\isacharequal}}\ x\ {\isaliteral{7C}{\isacharbar}}\ T\ {\isaliteral{5C3C72696768746172726F773E}{\isasymrightarrow}}\ T}\hspace{15mm}
  \isa{S\ {\isaliteral{3A}{\isacharcolon}}{\isaliteral{3A}{\isacharcolon}}{\isaliteral{3D}{\isacharequal}}\ {\isaliteral{5C3C666F72616C6C3E}{\isasymforall}}{\isaliteral{7B}{\isacharbraceleft}}x\isaliteral{5C3C5E697375623E}{}\isactrlisub {\isadigit{1}}{\isaliteral{2C}{\isacharcomma}}{\isaliteral{5C3C646F74733E}{\isasymdots}}{\isaliteral{2C}{\isacharcomma}}\ x\isaliteral{5C3C5E697375623E}{}\isactrlisub n{\isaliteral{7D}{\isacharbraceright}}{\isaliteral{2E}{\isachardot}}\ T}
  \end{array}
  \end{equation}\smallskip

  \noindent
  and the \isa{{\isaliteral{5C3C666F72616C6C3E}{\isasymforall}}}-quantification binds a finite (possibly empty) set of
  type-variables.  While it is possible to implement this kind of more general
  binders by iterating single binders, like \isa{{\isaliteral{5C3C666F72616C6C3E}{\isasymforall}}x\isaliteral{5C3C5E697375623E}{}\isactrlisub {\isadigit{1}}{\isaliteral{2E}{\isachardot}}{\isaliteral{5C3C666F72616C6C3E}{\isasymforall}}x\isaliteral{5C3C5E697375623E}{}\isactrlisub {\isadigit{2}}{\isaliteral{2E}{\isachardot}}{\isaliteral{2E}{\isachardot}}{\isaliteral{2E}{\isachardot}}{\isaliteral{5C3C666F72616C6C3E}{\isasymforall}}x\isaliteral{5C3C5E697375623E}{}\isactrlisub n{\isaliteral{2E}{\isachardot}}T}, this leads to a rather clumsy
  formalisation of W. For example, the usual definition for a
  type being an instance of a type-scheme requires in the iterated version 
  the following auxiliary \emph{unbinding relation}:

  \[
  \infer{\isa{T} \hookrightarrow ([], \isa{T})}{}\qquad
  \infer{\forall \isa{x{\isaliteral{2E}{\isachardot}}S} \hookrightarrow (\isa{x}\!::\!\isa{xs}, \isa{T})}
   {\isa{S} \hookrightarrow (\isa{xs}, \isa{T})}
  \]\smallskip

  \noindent
  Its purpose is to relate a type-scheme with a list of type-variables and a type. It is used to
  address the following problem:
  Given a type-scheme, say \isa{S}, how does one get access to the bound type-variables 
  and the type-part of \isa{S}? The unbinding relation gives an answer to this problem, though 
  in general it will only provide \emph{a} list of type-variables together with \emph{a} type that are  
  ``alpha-equivalent'' to \isa{S}. This is because unbinding is a relation; it cannot be a function
  for alpha-equated type-schemes. With the unbinding relation  
  in place, we can define when a type \isa{T} is an instance of a type-scheme \isa{S} as follows:

  \[
  \isa{T\ {\isaliteral{5C3C707265633E}{\isasymprec}}\ S\ {\isaliteral{5C3C65717569763E}{\isasymequiv}}\ {\isaliteral{5C3C6578697374733E}{\isasymexists}}xs\ T{\isaliteral{27}{\isacharprime}}\ {\isaliteral{5C3C7369676D613E}{\isasymsigma}}{\isaliteral{2E}{\isachardot}}\ S\ {\isaliteral{5C3C686F6F6B72696768746172726F773E}{\isasymhookrightarrow}}\ {\isaliteral{28}{\isacharparenleft}}xs{\isaliteral{2C}{\isacharcomma}}\ T{\isaliteral{27}{\isacharprime}}{\isaliteral{29}{\isacharparenright}}\ {\isaliteral{5C3C616E643E}{\isasymand}}\ dom\ {\isaliteral{5C3C7369676D613E}{\isasymsigma}}\ {\isaliteral{3D}{\isacharequal}}\ set\ xs\ {\isaliteral{5C3C616E643E}{\isasymand}}\ {\isaliteral{5C3C7369676D613E}{\isasymsigma}}{\isaliteral{28}{\isacharparenleft}}T{\isaliteral{27}{\isacharprime}}{\isaliteral{29}{\isacharparenright}}\ {\isaliteral{3D}{\isacharequal}}\ T}
  \]\smallskip
  
  \noindent
  This means there exists a list of type-variables \isa{xs} and a type \isa{T{\isaliteral{27}{\isacharprime}}} to which
  the type-scheme \isa{S} unbinds, and there exists a substitution \isa{{\isaliteral{5C3C7369676D613E}{\isasymsigma}}} whose domain is
  \isa{xs} (seen as set) such that \isa{{\isaliteral{5C3C7369676D613E}{\isasymsigma}}{\isaliteral{28}{\isacharparenleft}}T{\isaliteral{27}{\isacharprime}}{\isaliteral{29}{\isacharparenright}}\ {\isaliteral{3D}{\isacharequal}}\ T}.
  The problem with this definition is that we cannot follow the usual proofs 
  that are by induction on the type-part of the type-scheme (since it is under
  an existential quantifier and only an alpha-variant). The implementation of 
  type-schemes using iterations of single binders 
  prevents us from directly ``unbinding'' the bound type-variables and the type-part. 
  Clearly, a more dignified approach for formalising algorithm W is desirable. 
  The purpose of this paper is to introduce general binders, which 
  allow us to represent type-schemes so that they can bind multiple variables at once
  and as a result solve this problem more straightforwardly.
  The need of iterating single binders is also one reason
  why the existing Nominal Isabelle and similar theorem provers that only provide
  mechanisms for binding single variables have so far not fared very well with
  the more advanced tasks in the POPLmark challenge \cite{challenge05},
  because also there one would like to bind multiple variables at once.

  Binding multiple variables has interesting properties that cannot be captured
  easily by iterating single binders. For example in the case of type-schemes we do not
  want to make a distinction about the order of the bound variables. Therefore
  we would like to regard in \eqref{ex1} below  the first pair of type-schemes as alpha-equivalent,
  but assuming that \isa{x}, \isa{y} and \isa{z} are distinct variables,
  the second pair should \emph{not} be alpha-equivalent:

  \begin{equation}\label{ex1}
  \isa{{\isaliteral{5C3C666F72616C6C3E}{\isasymforall}}{\isaliteral{7B}{\isacharbraceleft}}x{\isaliteral{2C}{\isacharcomma}}\ y{\isaliteral{7D}{\isacharbraceright}}{\isaliteral{2E}{\isachardot}}\ x\ {\isaliteral{5C3C72696768746172726F773E}{\isasymrightarrow}}\ y\ \ {\isaliteral{5C3C617070726F783E}{\isasymapprox}}\isaliteral{5C3C5E697375623E}{}\isactrlisub {\isaliteral{5C3C616C7068613E}{\isasymalpha}}\ \ {\isaliteral{5C3C666F72616C6C3E}{\isasymforall}}{\isaliteral{7B}{\isacharbraceleft}}x{\isaliteral{2C}{\isacharcomma}}\ y{\isaliteral{7D}{\isacharbraceright}}{\isaliteral{2E}{\isachardot}}\ y\ {\isaliteral{5C3C72696768746172726F773E}{\isasymrightarrow}}\ x}\hspace{10mm}
  \isa{{\isaliteral{5C3C666F72616C6C3E}{\isasymforall}}{\isaliteral{7B}{\isacharbraceleft}}x{\isaliteral{2C}{\isacharcomma}}\ y{\isaliteral{7D}{\isacharbraceright}}{\isaliteral{2E}{\isachardot}}\ x\ {\isaliteral{5C3C72696768746172726F773E}{\isasymrightarrow}}\ y\ \ {\isaliteral{5C3C6E6F74617070726F783E}{\isasymnotapprox}}\isaliteral{5C3C5E697375623E}{}\isactrlisub {\isaliteral{5C3C616C7068613E}{\isasymalpha}}\ \ {\isaliteral{5C3C666F72616C6C3E}{\isasymforall}}{\isaliteral{7B}{\isacharbraceleft}}z{\isaliteral{7D}{\isacharbraceright}}{\isaliteral{2E}{\isachardot}}\ z\ {\isaliteral{5C3C72696768746172726F773E}{\isasymrightarrow}}\ z}
  \end{equation}\smallskip

  \noindent
  Moreover, we like to regard type-schemes as alpha-equivalent, if they differ
  only on \emph{vacuous} binders, such as

  \begin{equation}\label{ex3}
  \isa{{\isaliteral{5C3C666F72616C6C3E}{\isasymforall}}{\isaliteral{7B}{\isacharbraceleft}}x{\isaliteral{7D}{\isacharbraceright}}{\isaliteral{2E}{\isachardot}}\ x\ {\isaliteral{5C3C72696768746172726F773E}{\isasymrightarrow}}\ y\ \ {\isaliteral{5C3C617070726F783E}{\isasymapprox}}\isaliteral{5C3C5E697375623E}{}\isactrlisub {\isaliteral{5C3C616C7068613E}{\isasymalpha}}\ \ {\isaliteral{5C3C666F72616C6C3E}{\isasymforall}}{\isaliteral{7B}{\isacharbraceleft}}x{\isaliteral{2C}{\isacharcomma}}\ z{\isaliteral{7D}{\isacharbraceright}}{\isaliteral{2E}{\isachardot}}\ x\ {\isaliteral{5C3C72696768746172726F773E}{\isasymrightarrow}}\ y}
  \end{equation}\smallskip

  \noindent
  where \isa{z} does not occur freely in the type.  In this paper we will
  give a general binding mechanism and associated notion of alpha-equivalence
  that can be used to faithfully represent this kind of binding in Nominal
  Isabelle.  The difficulty of finding the right notion for alpha-equivalence
  can be appreciated in this case by considering that the definition given for
  type-schemes by Leroy in \cite[Page 18--19]{Leroy92} is incorrect (it omits a side-condition).

  However, the notion of alpha-equivalence that is preserved by vacuous
  binders is not always wanted. For example in terms like

  \begin{equation}\label{one}
  \isa{{\isaliteral{5C3C4C45543E}{\isasymLET}}\ x\ {\isaliteral{3D}{\isacharequal}}\ {\isadigit{3}}\ {\isaliteral{5C3C414E443E}{\isasymAND}}\ y\ {\isaliteral{3D}{\isacharequal}}\ {\isadigit{2}}\ {\isaliteral{5C3C494E3E}{\isasymIN}}\ x\ {\isaliteral{2D}{\isacharminus}}\ y\ {\isaliteral{5C3C454E443E}{\isasymEND}}}
  \end{equation}\smallskip

  \noindent
  we might not care in which order the assignments \isa{x\ {\isaliteral{3D}{\isacharequal}}\ {\isadigit{3}}} and
  \mbox{\isa{y\ {\isaliteral{3D}{\isacharequal}}\ {\isadigit{2}}}} are given, but it would be often unusual (particularly
  in strict languages) to regard \eqref{one} as alpha-equivalent with

  \[
  \isa{{\isaliteral{5C3C4C45543E}{\isasymLET}}\ x\ {\isaliteral{3D}{\isacharequal}}\ {\isadigit{3}}\ {\isaliteral{5C3C414E443E}{\isasymAND}}\ y\ {\isaliteral{3D}{\isacharequal}}\ {\isadigit{2}}\ {\isaliteral{5C3C414E443E}{\isasymAND}}\ z\ {\isaliteral{3D}{\isacharequal}}\ foo\ {\isaliteral{5C3C494E3E}{\isasymIN}}\ x\ {\isaliteral{2D}{\isacharminus}}\ y\ {\isaliteral{5C3C454E443E}{\isasymEND}}}
  \]\smallskip

  \noindent
  Therefore we will also provide a separate binding mechanism for cases in
  which the order of binders does not matter, but the `cardinality' of the
  binders has to agree.

  However, we found that this is still not sufficient for dealing with
  language constructs frequently occurring in programming language
  research. For example in \isa{{\isaliteral{5C3C4C45543E}{\isasymLET}}}s containing patterns like

  \begin{equation}\label{two}
  \isa{{\isaliteral{5C3C4C45543E}{\isasymLET}}\ {\isaliteral{28}{\isacharparenleft}}x{\isaliteral{2C}{\isacharcomma}}\ y{\isaliteral{29}{\isacharparenright}}\ {\isaliteral{3D}{\isacharequal}}\ {\isaliteral{28}{\isacharparenleft}}{\isadigit{3}}{\isaliteral{2C}{\isacharcomma}}\ {\isadigit{2}}{\isaliteral{29}{\isacharparenright}}\ {\isaliteral{5C3C494E3E}{\isasymIN}}\ x\ {\isaliteral{2D}{\isacharminus}}\ y\ {\isaliteral{5C3C454E443E}{\isasymEND}}}
  \end{equation}\smallskip

  \noindent
  we want to bind all variables from the pattern inside the body of the
  $\mathtt{let}$, but we also care about the order of these variables, since
  we do not want to regard \eqref{two} as alpha-equivalent with

  \[
  \isa{{\isaliteral{5C3C4C45543E}{\isasymLET}}\ {\isaliteral{28}{\isacharparenleft}}y{\isaliteral{2C}{\isacharcomma}}\ x{\isaliteral{29}{\isacharparenright}}\ {\isaliteral{3D}{\isacharequal}}\ {\isaliteral{28}{\isacharparenleft}}{\isadigit{3}}{\isaliteral{2C}{\isacharcomma}}\ {\isadigit{2}}{\isaliteral{29}{\isacharparenright}}\ {\isaliteral{5C3C494E3E}{\isasymIN}}\ x\ {\isaliteral{2D}{\isacharminus}}\ y\ {\isaliteral{5C3C454E443E}{\isasymEND}}}
  \]\smallskip

  \noindent
  As a result, we provide three general binding mechanisms each of which binds
  multiple variables at once, and let the user choose which one is intended
  when formalising a term-calculus.

  By providing these general binding mechanisms, however, we have to work
  around a problem that has been pointed out by Pottier \cite{Pottier06} and
  Cheney \cite{Cheney05}: in \isa{{\isaliteral{5C3C4C45543E}{\isasymLET}}}-constructs of the form

  \[
  \isa{{\isaliteral{5C3C4C45543E}{\isasymLET}}\ x\isaliteral{5C3C5E697375623E}{}\isactrlisub {\isadigit{1}}\ {\isaliteral{3D}{\isacharequal}}\ t\isaliteral{5C3C5E697375623E}{}\isactrlisub {\isadigit{1}}\ {\isaliteral{5C3C414E443E}{\isasymAND}}\ {\isaliteral{5C3C646F74733E}{\isasymdots}}\ {\isaliteral{5C3C414E443E}{\isasymAND}}\ x\isaliteral{5C3C5E697375623E}{}\isactrlisub n\ {\isaliteral{3D}{\isacharequal}}\ t\isaliteral{5C3C5E697375623E}{}\isactrlisub n\ {\isaliteral{5C3C494E3E}{\isasymIN}}\ s\ {\isaliteral{5C3C454E443E}{\isasymEND}}}
  \]\smallskip

  \noindent
  we care about the information that there are as many bound variables \isa{x\isaliteral{5C3C5E697375623E}{}\isactrlisub i} as there are \isa{t\isaliteral{5C3C5E697375623E}{}\isactrlisub i}. We lose this information if
  we represent the \isa{{\isaliteral{5C3C4C45543E}{\isasymLET}}}-constructor by something like

  \[
  \isa{{\isaliteral{5C3C4C45543E}{\isasymLET}}\ {\isaliteral{28}{\isacharparenleft}}{\isaliteral{5C3C6C616D6264613E}{\isasymlambda}}x\isaliteral{5C3C5E697375623E}{}\isactrlisub {\isadigit{1}}{\isaliteral{5C3C646F74733E}{\isasymdots}}x\isaliteral{5C3C5E697375623E}{}\isactrlisub n\ {\isaliteral{2E}{\isachardot}}\ s{\isaliteral{29}{\isacharparenright}}\ \ {\isaliteral{5B}{\isacharbrackleft}}t\isaliteral{5C3C5E697375623E}{}\isactrlisub {\isadigit{1}}{\isaliteral{2C}{\isacharcomma}}{\isaliteral{5C3C646F74733E}{\isasymdots}}{\isaliteral{2C}{\isacharcomma}}t\isaliteral{5C3C5E697375623E}{}\isactrlisub n{\isaliteral{5D}{\isacharbrackright}}}
  \]\smallskip

  \noindent
  where the notation \isa{{\isaliteral{5C3C6C616D6264613E}{\isasymlambda}}{\isaliteral{5F}{\isacharunderscore}}\ {\isaliteral{2E}{\isachardot}}\ {\isaliteral{5F}{\isacharunderscore}}} indicates that the list of \isa{x\isaliteral{5C3C5E697375623E}{}\isactrlisub i} becomes bound in \isa{s}. In this representation the term
  \mbox{\isa{{\isaliteral{5C3C4C45543E}{\isasymLET}}\ {\isaliteral{28}{\isacharparenleft}}{\isaliteral{5C3C6C616D6264613E}{\isasymlambda}}x\ {\isaliteral{2E}{\isachardot}}\ s{\isaliteral{29}{\isacharparenright}}\ {\isaliteral{5B}{\isacharbrackleft}}t\isaliteral{5C3C5E697375623E}{}\isactrlisub {\isadigit{1}}{\isaliteral{2C}{\isacharcomma}}\ t\isaliteral{5C3C5E697375623E}{}\isactrlisub {\isadigit{2}}{\isaliteral{5D}{\isacharbrackright}}}} is a perfectly
  legal instance, but the lengths of the two lists do not agree. To exclude
  such terms, additional predicates about well-formed terms are needed in
  order to ensure that the two lists are of equal length. This can result in
  very messy reasoning (see for example~\cite{BengtsonParow09}). To avoid
  this, we will allow type specifications for \isa{{\isaliteral{5C3C4C45543E}{\isasymLET}}}s as follows

  \[
  \mbox{\begin{tabular}{r@ {\hspace{2mm}}r@ {\hspace{2mm}}ll}
  \isa{trm} & \isa{{\isaliteral{3A}{\isacharcolon}}{\isaliteral{3A}{\isacharcolon}}{\isaliteral{3D}{\isacharequal}}}  & \isa{{\isaliteral{5C3C646F74733E}{\isasymdots}}} \\
              & \isa{{\isaliteral{7C}{\isacharbar}}}    & \isa{{\isaliteral{5C3C4C45543E}{\isasymLET}}\ \ as{\isaliteral{3A}{\isacharcolon}}{\isaliteral{3A}{\isacharcolon}}assn\ \ s{\isaliteral{3A}{\isacharcolon}}{\isaliteral{3A}{\isacharcolon}}trm}\hspace{2mm} 
                                 \isacommand{binds} \isa{bn{\isaliteral{28}{\isacharparenleft}}as{\isaliteral{29}{\isacharparenright}}} \isacommand{in} \isa{s}\\[1mm]
  \isa{assn} & \isa{{\isaliteral{3A}{\isacharcolon}}{\isaliteral{3A}{\isacharcolon}}{\isaliteral{3D}{\isacharequal}}} & \isa{{\isaliteral{5C3C414E494C3E}{\isasymANIL}}}\\
               &  \isa{{\isaliteral{7C}{\isacharbar}}}  & \isa{{\isaliteral{5C3C41434F4E533E}{\isasymACONS}}\ \ name\ \ trm\ \ assn}
  \end{tabular}}
  \]\smallskip

  \noindent
  where \isa{assn} is an auxiliary type representing a list of assignments
  and \isa{bn} an auxiliary function identifying the variables to be bound
  by the \isa{{\isaliteral{5C3C4C45543E}{\isasymLET}}}. This function can be defined by recursion over \isa{assn} as follows

  \[
  \isa{bn{\isaliteral{28}{\isacharparenleft}}{\isaliteral{5C3C414E494C3E}{\isasymANIL}}{\isaliteral{29}{\isacharparenright}}\ {\isaliteral{3D}{\isacharequal}}}~\isa{{\isaliteral{5C3C656D7074797365743E}{\isasymemptyset}}} \hspace{10mm} 
  \isa{bn{\isaliteral{28}{\isacharparenleft}}{\isaliteral{5C3C41434F4E533E}{\isasymACONS}}\ x\ t\ as{\isaliteral{29}{\isacharparenright}}\ {\isaliteral{3D}{\isacharequal}}\ {\isaliteral{7B}{\isacharbraceleft}}x{\isaliteral{7D}{\isacharbraceright}}\ {\isaliteral{5C3C756E696F6E3E}{\isasymunion}}\ bn{\isaliteral{28}{\isacharparenleft}}as{\isaliteral{29}{\isacharparenright}}} 
  \]\smallskip

  \noindent
  The scope of the binding is indicated by labels given to the types, for
  example \isa{s{\isaliteral{3A}{\isacharcolon}}{\isaliteral{3A}{\isacharcolon}}trm}, and a binding clause, in this case
  \isacommand{binds} \isa{bn{\isaliteral{28}{\isacharparenleft}}as{\isaliteral{29}{\isacharparenright}}} \isacommand{in} \isa{s}. This binding
  clause states that all the names the function \isa{bn{\isaliteral{28}{\isacharparenleft}}as{\isaliteral{29}{\isacharparenright}}} returns
  should be bound in \isa{s}.  This style of specifying terms and bindings
  is heavily inspired by the syntax of the Ott-tool \cite{ott-jfp}. Our work
  extends Ott in several aspects: one is that we support three binding
  modes---Ott has only one, namely the one where the order of binders matters.
  Another is that our reasoning infrastructure, like strong induction principles
  and the notion of free variables, is derived from first principles within 
  the Isabelle/HOL theorem prover.

  However, we will not be able to cope with all specifications that are
  allowed by Ott. One reason is that Ott lets the user specify `empty' types
  like \mbox{\isa{t\ {\isaliteral{3A}{\isacharcolon}}{\isaliteral{3A}{\isacharcolon}}{\isaliteral{3D}{\isacharequal}}\ t\ t\ {\isaliteral{7C}{\isacharbar}}\ {\isaliteral{5C3C6C616D6264613E}{\isasymlambda}}x{\isaliteral{2E}{\isachardot}}\ t}} where no clause for variables is
  given. Arguably, such specifications make some sense in the context of Coq's
  type theory (which Ott supports), but not at all in a HOL-based environment
  where every datatype must have a non-empty set-theoretic model
  \cite{Berghofer99}.  Another reason is that we establish the reasoning
  infrastructure for alpha-\emph{equated} terms. In contrast, Ott produces a
  reasoning infrastructure in Isabelle/HOL for \emph{non}-alpha-equated, or
  `raw', terms. While our alpha-equated terms and the `raw' terms produced by
  Ott use names for bound variables, there is a key difference: working with
  alpha-equated terms means, for example, that the two type-schemes

  \[
  \isa{{\isaliteral{5C3C666F72616C6C3E}{\isasymforall}}{\isaliteral{7B}{\isacharbraceleft}}x{\isaliteral{7D}{\isacharbraceright}}{\isaliteral{2E}{\isachardot}}\ x\ {\isaliteral{5C3C72696768746172726F773E}{\isasymrightarrow}}\ y\ \ {\isaliteral{3D}{\isacharequal}}\ {\isaliteral{5C3C666F72616C6C3E}{\isasymforall}}{\isaliteral{7B}{\isacharbraceleft}}x{\isaliteral{2C}{\isacharcomma}}\ z{\isaliteral{7D}{\isacharbraceright}}{\isaliteral{2E}{\isachardot}}\ x\ {\isaliteral{5C3C72696768746172726F773E}{\isasymrightarrow}}\ y} 
  \]\smallskip
  
  \noindent
  are not just alpha-equal, but actually \emph{equal}! As a result, we can
  only support specifications that make sense on the level of alpha-equated
  terms (offending specifications, which for example bind a variable according
  to a variable bound somewhere else, are not excluded by Ott, but we have
  to).  

  Our insistence on reasoning with alpha-equated terms comes from the
  wealth of experience we gained with the older version of Nominal Isabelle:
  for non-trivial properties, reasoning with alpha-equated terms is much
  easier than reasoning with `raw' terms. The fundamental reason for this is
  that the HOL-logic underlying Nominal Isabelle allows us to replace
  `equals-by-equals'. In contrast, replacing
  `alpha-equals-by-alpha-equals' in a representation based on `raw' terms
  requires a lot of extra reasoning work.

  Although in informal settings a reasoning infrastructure for alpha-equated
  terms is nearly always taken for granted, establishing it automatically in
  Isabelle/HOL is a rather non-trivial task. For every
  specification we will need to construct type(s) containing as elements the
  alpha-equated terms. To do so, we use the standard HOL-technique of defining
  a new type by identifying a non-empty subset of an existing type.  The
  construction we perform in Isabelle/HOL can be illustrated by the following picture:

  \begin{equation}\label{picture}
  \mbox
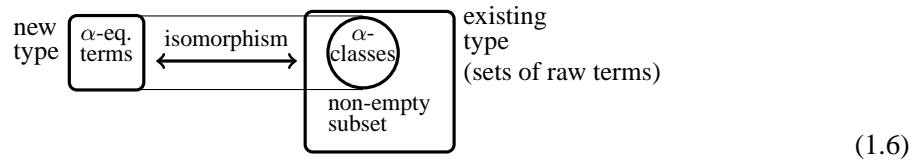
{\begin{tikzpicture}[scale=1.1]
  
  \draw[very thick] (0.7,0.4) circle (4.25mm);
  \draw[rounded corners=1mm, very thick] ( 0.0,-0.8) rectangle ( 1.8, 0.9);
  \draw[rounded corners=1mm, very thick] (-1.95,0.85) rectangle (-2.85,-0.05);
  
  \draw (-2.0, 0.845) --  (0.7,0.845);
  \draw (-2.0,-0.045)  -- (0.7,-0.045);

  \draw ( 0.7, 0.5) node {\footnotesize\begin{tabular}{@ {}c@ {}}$\alpha$-\\[-1mm]classes\end{tabular}};
  \draw (-2.4, 0.5) node {\footnotesize\begin{tabular}{@ {}c@ {}}$\alpha$-eq.\\[-1mm]terms\end{tabular}};
  \draw (1.8, 0.48) node[right=-0.1mm]
    {\small\begin{tabular}{@ {}l@ {}}existing\\[-1mm] type\\ (sets of raw terms)\end{tabular}};
  \draw (0.9, -0.35) node {\footnotesize\begin{tabular}{@ {}l@ {}}non-empty\\[-1mm]subset\end{tabular}};
  \draw (-3.25, 0.55) node {\small\begin{tabular}{@ {}l@ {}}new\\[-1mm]type\end{tabular}};
  
  \draw[<->, very thick] (-1.8, 0.3) -- (-0.1,0.3);
  \draw (-0.95, 0.3) node[above=0mm] {\footnotesize{}isomorphism};

  \end{tikzpicture}}
  \end{equation}\smallskip

  \noindent
  We take as the starting point a definition of raw terms (defined as a
  datatype in Isabelle/HOL); then identify the alpha-equivalence classes in
  the type of sets of raw terms according to our alpha-equivalence relation,
  and finally define the new type as these alpha-equivalence classes (the
  non-emptiness requirement is always satisfied whenever the raw terms are
  definable as datatype in Isabelle/HOL and our relation for alpha-equivalence
  is an equivalence relation).

  The fact that we obtain an isomorphism between the new type and the
  non-empty subset shows that the new type is a faithful representation of
  alpha-equated terms. That is not the case for example for terms using the
  locally nameless representation of binders \cite{McKinnaPollack99}: in this
  representation there are `junk' terms that need to be excluded by
  reasoning about a well-formedness predicate.

  The problem with introducing a new type in Isabelle/HOL is that in order to
  be useful, a reasoning infrastructure needs to be `lifted' from the
  underlying subset to the new type. This is usually a tricky and arduous
  task. To ease it, we re-implemented in Isabelle/HOL \cite{KaliszykUrban11}
  the quotient package described by Homeier \cite{Homeier05} for the HOL4
  system. This package allows us to lift definitions and theorems involving
  raw terms to definitions and theorems involving alpha-equated terms. For
  example if we define the free-variable function over raw lambda-terms
  as follows

  \[
  \mbox{\begin{tabular}{l@ {\hspace{1mm}}r@ {\hspace{1mm}}l}
  \isa{fv{\isaliteral{28}{\isacharparenleft}}x{\isaliteral{29}{\isacharparenright}}}     & \isa{{\isaliteral{5C3C65717569763E}{\isasymequiv}}} & \isa{{\isaliteral{7B}{\isacharbraceleft}}x{\isaliteral{7D}{\isacharbraceright}}}\\
  \isa{fv{\isaliteral{28}{\isacharparenleft}}t\isaliteral{5C3C5E697375623E}{}\isactrlisub {\isadigit{1}}\ t\isaliteral{5C3C5E697375623E}{}\isactrlisub {\isadigit{2}}{\isaliteral{29}{\isacharparenright}}} & \isa{{\isaliteral{5C3C65717569763E}{\isasymequiv}}} & \isa{fv{\isaliteral{28}{\isacharparenleft}}t\isaliteral{5C3C5E697375623E}{}\isactrlisub {\isadigit{1}}{\isaliteral{29}{\isacharparenright}}\ {\isaliteral{5C3C756E696F6E3E}{\isasymunion}}\ fv{\isaliteral{28}{\isacharparenleft}}t\isaliteral{5C3C5E697375623E}{}\isactrlisub {\isadigit{2}}{\isaliteral{29}{\isacharparenright}}}\\
  \isa{fv{\isaliteral{28}{\isacharparenleft}}{\isaliteral{5C3C6C616D6264613E}{\isasymlambda}}x{\isaliteral{2E}{\isachardot}}t{\isaliteral{29}{\isacharparenright}}}  & \isa{{\isaliteral{5C3C65717569763E}{\isasymequiv}}} & \isa{fv{\isaliteral{28}{\isacharparenleft}}t{\isaliteral{29}{\isacharparenright}}\ {\isaliteral{2D}{\isacharminus}}\ {\isaliteral{7B}{\isacharbraceleft}}x{\isaliteral{7D}{\isacharbraceright}}}
  \end{tabular}}
  \]\smallskip
  
  \noindent
  then with the help of the quotient package we can obtain a function \isa{fv\isaliteral{5C3C5E7375703E}{}\isactrlsup {\isaliteral{5C3C616C7068613E}{\isasymalpha}}}
  operating on quotients, that is alpha-equivalence classes of lambda-terms. This
  lifted function is characterised by the equations

  \[
  \mbox{\begin{tabular}{l@ {\hspace{1mm}}r@ {\hspace{1mm}}l}
  \isa{fv\isaliteral{5C3C5E7375703E}{}\isactrlsup {\isaliteral{5C3C616C7068613E}{\isasymalpha}}{\isaliteral{28}{\isacharparenleft}}x{\isaliteral{29}{\isacharparenright}}}     & \isa{{\isaliteral{3D}{\isacharequal}}} & \isa{{\isaliteral{7B}{\isacharbraceleft}}x{\isaliteral{7D}{\isacharbraceright}}}\\
  \isa{fv\isaliteral{5C3C5E7375703E}{}\isactrlsup {\isaliteral{5C3C616C7068613E}{\isasymalpha}}{\isaliteral{28}{\isacharparenleft}}t\isaliteral{5C3C5E697375623E}{}\isactrlisub {\isadigit{1}}\ t\isaliteral{5C3C5E697375623E}{}\isactrlisub {\isadigit{2}}{\isaliteral{29}{\isacharparenright}}} & \isa{{\isaliteral{3D}{\isacharequal}}} & \isa{fv\isaliteral{5C3C5E7375703E}{}\isactrlsup {\isaliteral{5C3C616C7068613E}{\isasymalpha}}{\isaliteral{28}{\isacharparenleft}}t\isaliteral{5C3C5E697375623E}{}\isactrlisub {\isadigit{1}}{\isaliteral{29}{\isacharparenright}}\ {\isaliteral{5C3C756E696F6E3E}{\isasymunion}}\ fv\isaliteral{5C3C5E7375703E}{}\isactrlsup {\isaliteral{5C3C616C7068613E}{\isasymalpha}}{\isaliteral{28}{\isacharparenleft}}t\isaliteral{5C3C5E697375623E}{}\isactrlisub {\isadigit{2}}{\isaliteral{29}{\isacharparenright}}}\\
  \isa{fv\isaliteral{5C3C5E7375703E}{}\isactrlsup {\isaliteral{5C3C616C7068613E}{\isasymalpha}}{\isaliteral{28}{\isacharparenleft}}{\isaliteral{5C3C6C616D6264613E}{\isasymlambda}}x{\isaliteral{2E}{\isachardot}}t{\isaliteral{29}{\isacharparenright}}}  & \isa{{\isaliteral{3D}{\isacharequal}}} & \isa{fv\isaliteral{5C3C5E7375703E}{}\isactrlsup {\isaliteral{5C3C616C7068613E}{\isasymalpha}}{\isaliteral{28}{\isacharparenleft}}t{\isaliteral{29}{\isacharparenright}}\ {\isaliteral{2D}{\isacharminus}}\ {\isaliteral{7B}{\isacharbraceleft}}x{\isaliteral{7D}{\isacharbraceright}}}
  \end{tabular}}
  \]\smallskip

  \noindent
  (Note that this means also the term-constructors for variables, applications
  and lambda are lifted to the quotient level.)  This construction, of course,
  only works if alpha-equivalence is indeed an equivalence relation, and the
  `raw' definitions and theorems are respectful w.r.t.~alpha-equivalence.
  For example, we will not be able to lift a bound-variable function. Although
  this function can be defined for raw terms, it does not respect
  alpha-equivalence and therefore cannot be lifted. 
  To sum up, every lifting
  of theorems to the quotient level needs proofs of some respectfulness
  properties (see \cite{Homeier05}). In the paper we show that we are able to
  automate these proofs and as a result can automatically establish a reasoning 
  infrastructure for alpha-equated terms.\smallskip

  The examples we have in mind where our reasoning infrastructure will be
  helpful include the term language of Core-Haskell (see
  Figure~\ref{corehas}). This term language involves patterns that have lists
  of type-, coercion- and term-variables, all of which are bound in \isa{{\isaliteral{5C3C434153453E}{\isasymCASE}}}-expressions. In these patterns we do not know in advance how many
  variables need to be bound. Another example is the algorithm W,
  which includes multiple binders in type-schemes.\medskip

  \noindent
  {\bf Contributions:} We provide three new definitions for when terms
  involving general binders are alpha-equivalent. These definitions are
  inspired by earlier work of Pitts \cite{Pitts04}. By means of automati\-cally-generated
  proofs, we establish a reasoning infrastructure for alpha-equated terms,
  including properties about support, freshness and equality conditions for
  alpha-equated terms. We are also able to automatically derive strong
  induction principles that have the variable convention already built in.
  For this we simplify the earlier automated proofs by using the proving tools
  from the function package~\cite{Krauss09} of Isabelle/HOL.  The method
  behind our specification of general binders is taken from the Ott-tool, but
  we introduce crucial restrictions, and also extensions, so that our
  specifications make sense for reasoning about alpha-equated terms.  The main
  improvement over Ott is that we introduce three binding modes (only one is
  present in Ott), provide formalised definitions for alpha-equivalence and
  for free variables of our terms, and also derive a reasoning infrastructure
  for our specifications from `first principles' inside a theorem prover.

  \begin{figure}[t]
  \begin{boxedminipage}{\linewidth}
  \begin{center}
  \begin{tabular}{@ {\hspace{8mm}}r@ {\hspace{2mm}}r@ {\hspace{2mm}}l}
  \multicolumn{3}{@ {}l}{Type Kinds}\\
  \isa{{\isaliteral{5C3C6B617070613E}{\isasymkappa}}} & \isa{{\isaliteral{3A}{\isacharcolon}}{\isaliteral{3A}{\isacharcolon}}{\isaliteral{3D}{\isacharequal}}} & \isa{{\isaliteral{5C3C737461723E}{\isasymstar}}\ {\isaliteral{7C}{\isacharbar}}\ {\isaliteral{5C3C6B617070613E}{\isasymkappa}}\isaliteral{5C3C5E697375623E}{}\isactrlisub {\isadigit{1}}\ {\isaliteral{5C3C72696768746172726F773E}{\isasymrightarrow}}\ {\isaliteral{5C3C6B617070613E}{\isasymkappa}}\isaliteral{5C3C5E697375623E}{}\isactrlisub {\isadigit{2}}}\smallskip\\
  \multicolumn{3}{@ {}l}{Coercion Kinds}\\
  \isa{{\isaliteral{5C3C696F74613E}{\isasymiota}}} & \isa{{\isaliteral{3A}{\isacharcolon}}{\isaliteral{3A}{\isacharcolon}}{\isaliteral{3D}{\isacharequal}}} & \isa{{\isaliteral{5C3C7369676D613E}{\isasymsigma}}\isaliteral{5C3C5E697375623E}{}\isactrlisub {\isadigit{1}}\ {\isaliteral{5C3C73696D3E}{\isasymsim}}\ {\isaliteral{5C3C7369676D613E}{\isasymsigma}}\isaliteral{5C3C5E697375623E}{}\isactrlisub {\isadigit{2}}}\smallskip\\
  \multicolumn{3}{@ {}l}{Types}\\
  \isa{{\isaliteral{5C3C7369676D613E}{\isasymsigma}}} & \isa{{\isaliteral{3A}{\isacharcolon}}{\isaliteral{3A}{\isacharcolon}}{\isaliteral{3D}{\isacharequal}}} & \isa{a\ {\isaliteral{7C}{\isacharbar}}\ T\ {\isaliteral{7C}{\isacharbar}}\ {\isaliteral{5C3C7369676D613E}{\isasymsigma}}\isaliteral{5C3C5E697375623E}{}\isactrlisub {\isadigit{1}}\ {\isaliteral{5C3C7369676D613E}{\isasymsigma}}\isaliteral{5C3C5E697375623E}{}\isactrlisub {\isadigit{2}}\ {\isaliteral{7C}{\isacharbar}}\ S\isaliteral{5C3C5E697375623E}{}\isactrlisub n}$\;\overline{\isa{{\isaliteral{5C3C7369676D613E}{\isasymsigma}}}}$\isa{\isaliteral{5C3C5E7375703E}{}\isactrlsup n} 
  \isa{{\isaliteral{7C}{\isacharbar}}\ {\isaliteral{5C3C666F72616C6C3E}{\isasymforall}}a{\isaliteral{3A}{\isacharcolon}}{\isaliteral{5C3C6B617070613E}{\isasymkappa}}{\isaliteral{2E}{\isachardot}}\ {\isaliteral{5C3C7369676D613E}{\isasymsigma}}\ {\isaliteral{7C}{\isacharbar}}\ {\isaliteral{5C3C696F74613E}{\isasymiota}}\ {\isaliteral{5C3C52696768746172726F773E}{\isasymRightarrow}}\ {\isaliteral{5C3C7369676D613E}{\isasymsigma}}}\smallskip\\
  \multicolumn{3}{@ {}l}{Coercion Types}\\
  \isa{{\isaliteral{5C3C67616D6D613E}{\isasymgamma}}} & \isa{{\isaliteral{3A}{\isacharcolon}}{\isaliteral{3A}{\isacharcolon}}{\isaliteral{3D}{\isacharequal}}} & \isa{c\ {\isaliteral{7C}{\isacharbar}}\ C\ {\isaliteral{7C}{\isacharbar}}\ {\isaliteral{5C3C67616D6D613E}{\isasymgamma}}\isaliteral{5C3C5E697375623E}{}\isactrlisub {\isadigit{1}}\ {\isaliteral{5C3C67616D6D613E}{\isasymgamma}}\isaliteral{5C3C5E697375623E}{}\isactrlisub {\isadigit{2}}\ {\isaliteral{7C}{\isacharbar}}\ S\isaliteral{5C3C5E697375623E}{}\isactrlisub n}$\;\overline{\isa{{\isaliteral{5C3C67616D6D613E}{\isasymgamma}}}}$\isa{\isaliteral{5C3C5E7375703E}{}\isactrlsup n}
  \isa{{\isaliteral{7C}{\isacharbar}}\ {\isaliteral{5C3C666F72616C6C3E}{\isasymforall}}c{\isaliteral{3A}{\isacharcolon}}{\isaliteral{5C3C696F74613E}{\isasymiota}}{\isaliteral{2E}{\isachardot}}\ {\isaliteral{5C3C67616D6D613E}{\isasymgamma}}\ {\isaliteral{7C}{\isacharbar}}\ {\isaliteral{5C3C696F74613E}{\isasymiota}}\ {\isaliteral{5C3C52696768746172726F773E}{\isasymRightarrow}}\ {\isaliteral{5C3C67616D6D613E}{\isasymgamma}}\ {\isaliteral{7C}{\isacharbar}}\ refl\ {\isaliteral{5C3C7369676D613E}{\isasymsigma}}\ {\isaliteral{7C}{\isacharbar}}\ sym\ {\isaliteral{5C3C67616D6D613E}{\isasymgamma}}\ {\isaliteral{7C}{\isacharbar}}\ {\isaliteral{5C3C67616D6D613E}{\isasymgamma}}\isaliteral{5C3C5E697375623E}{}\isactrlisub {\isadigit{1}}\ {\isaliteral{5C3C636972633E}{\isasymcirc}}\ {\isaliteral{5C3C67616D6D613E}{\isasymgamma}}\isaliteral{5C3C5E697375623E}{}\isactrlisub {\isadigit{2}}}\\
  & \isa{{\isaliteral{7C}{\isacharbar}}} & \isa{{\isaliteral{5C3C67616D6D613E}{\isasymgamma}}\ {\isaliteral{40}{\isacharat}}\ {\isaliteral{5C3C7369676D613E}{\isasymsigma}}\ {\isaliteral{7C}{\isacharbar}}\ left\ {\isaliteral{5C3C67616D6D613E}{\isasymgamma}}\ {\isaliteral{7C}{\isacharbar}}\ right\ {\isaliteral{5C3C67616D6D613E}{\isasymgamma}}\ {\isaliteral{7C}{\isacharbar}}\ {\isaliteral{5C3C67616D6D613E}{\isasymgamma}}\isaliteral{5C3C5E697375623E}{}\isactrlisub {\isadigit{1}}\ {\isaliteral{5C3C73696D3E}{\isasymsim}}\ {\isaliteral{5C3C67616D6D613E}{\isasymgamma}}\isaliteral{5C3C5E697375623E}{}\isactrlisub {\isadigit{2}}\ {\isaliteral{7C}{\isacharbar}}\ rightc\ {\isaliteral{5C3C67616D6D613E}{\isasymgamma}}\ {\isaliteral{7C}{\isacharbar}}\ leftc\ {\isaliteral{5C3C67616D6D613E}{\isasymgamma}}\ {\isaliteral{7C}{\isacharbar}}\ {\isaliteral{5C3C67616D6D613E}{\isasymgamma}}\isaliteral{5C3C5E697375623E}{}\isactrlisub {\isadigit{1}}\ {\isaliteral{5C3C747269616E676C6572696768743E}{\isasymtriangleright}}\ {\isaliteral{5C3C67616D6D613E}{\isasymgamma}}\isaliteral{5C3C5E697375623E}{}\isactrlisub {\isadigit{2}}}\smallskip\\
  \multicolumn{3}{@ {}l}{Terms}\\
  \isa{e} & \isa{{\isaliteral{3A}{\isacharcolon}}{\isaliteral{3A}{\isacharcolon}}{\isaliteral{3D}{\isacharequal}}} & \isa{x\ {\isaliteral{7C}{\isacharbar}}\ K\ {\isaliteral{7C}{\isacharbar}}\ {\isaliteral{5C3C4C616D6264613E}{\isasymLambda}}a{\isaliteral{3A}{\isacharcolon}}{\isaliteral{5C3C6B617070613E}{\isasymkappa}}{\isaliteral{2E}{\isachardot}}\ e\ {\isaliteral{7C}{\isacharbar}}\ {\isaliteral{5C3C4C616D6264613E}{\isasymLambda}}c{\isaliteral{3A}{\isacharcolon}}{\isaliteral{5C3C696F74613E}{\isasymiota}}{\isaliteral{2E}{\isachardot}}\ e\ {\isaliteral{7C}{\isacharbar}}\ e\ {\isaliteral{5C3C7369676D613E}{\isasymsigma}}\ {\isaliteral{7C}{\isacharbar}}\ e\ {\isaliteral{5C3C67616D6D613E}{\isasymgamma}}\ {\isaliteral{7C}{\isacharbar}}\ {\isaliteral{5C3C6C616D6264613E}{\isasymlambda}}x{\isaliteral{3A}{\isacharcolon}}{\isaliteral{5C3C7369676D613E}{\isasymsigma}}{\isaliteral{2E}{\isachardot}}\ e\ {\isaliteral{7C}{\isacharbar}}\ e\isaliteral{5C3C5E697375623E}{}\isactrlisub {\isadigit{1}}\ e\isaliteral{5C3C5E697375623E}{}\isactrlisub {\isadigit{2}}}\\
  & \isa{{\isaliteral{7C}{\isacharbar}}} & \isa{{\isaliteral{5C3C4C45543E}{\isasymLET}}\ x{\isaliteral{3A}{\isacharcolon}}{\isaliteral{5C3C7369676D613E}{\isasymsigma}}\ {\isaliteral{3D}{\isacharequal}}\ e\isaliteral{5C3C5E697375623E}{}\isactrlisub {\isadigit{1}}\ {\isaliteral{5C3C494E3E}{\isasymIN}}\ e\isaliteral{5C3C5E697375623E}{}\isactrlisub {\isadigit{2}}\ {\isaliteral{7C}{\isacharbar}}\ {\isaliteral{5C3C434153453E}{\isasymCASE}}\ e\isaliteral{5C3C5E697375623E}{}\isactrlisub {\isadigit{1}}\ {\isaliteral{5C3C4F463E}{\isasymOF}}}$\;\overline{\isa{p\ {\isaliteral{5C3C72696768746172726F773E}{\isasymrightarrow}}\ e\isaliteral{5C3C5E697375623E}{}\isactrlisub {\isadigit{2}}}}$ \isa{{\isaliteral{7C}{\isacharbar}}\ e\ {\isaliteral{5C3C747269616E676C6572696768743E}{\isasymtriangleright}}\ {\isaliteral{5C3C67616D6D613E}{\isasymgamma}}}\smallskip\\
  \multicolumn{3}{@ {}l}{Patterns}\\
  \isa{p} & \isa{{\isaliteral{3A}{\isacharcolon}}{\isaliteral{3A}{\isacharcolon}}{\isaliteral{3D}{\isacharequal}}} & \isa{K}$\;\overline{\isa{a{\isaliteral{3A}{\isacharcolon}}{\isaliteral{5C3C6B617070613E}{\isasymkappa}}}}\;\overline{\isa{c{\isaliteral{3A}{\isacharcolon}}{\isaliteral{5C3C696F74613E}{\isasymiota}}}}\;\overline{\isa{x{\isaliteral{3A}{\isacharcolon}}{\isaliteral{5C3C7369676D613E}{\isasymsigma}}}}$\smallskip\\
  \multicolumn{3}{@ {}l}{Constants}\\
  & \isa{C} & coercion constants\\
  & \isa{T} & value type constructors\\
  & \isa{S\isaliteral{5C3C5E697375623E}{}\isactrlisub n} & n-ary type functions (which need to be fully applied)\\
  & \isa{K} & data constructors\smallskip\\
  \multicolumn{3}{@ {}l}{Variables}\\
  & \isa{a} & type variables\\
  & \isa{c} & coercion variables\\
  & \isa{x} & term variables\\
  \end{tabular}
  \end{center}
  \end{boxedminipage}
  \caption{The System \isa{F\isaliteral{5C3C5E697375623E}{}\isactrlisub C}
  \cite{CoreHaskell}, also often referred to as \emph{Core-Haskell}. In this
  version of \isa{F\isaliteral{5C3C5E697375623E}{}\isactrlisub C} we made a modification by separating the
  grammars for type kinds and coercion kinds, as well as for types and coercion
  types. For this paper the interesting term-constructor is \isa{{\isaliteral{5C3C434153453E}{\isasymCASE}}},
  which binds multiple type-, coercion- and term-variables (the overlines stand for lists).\label{corehas}}
  \end{figure}%
\end{isamarkuptext}%
\isamarkuptrue%
\isamarkupsection{A Short Review of the Nominal Logic Work%
}
\isamarkuptrue%
\begin{isamarkuptext}%
At its core, Nominal Isabelle is an adaptation of the nominal logic work by
  Pitts \cite{Pitts03}. This adaptation for Isabelle/HOL is described in
  \cite{HuffmanUrban10} (including proofs). We shall briefly review this work
  to aid the description of what follows. 

  Two central notions in the nominal logic work are sorted atoms and
  sort-respecting permutations of atoms. We will use the letters \isa{a{\isaliteral{2C}{\isacharcomma}}\ b{\isaliteral{2C}{\isacharcomma}}\ c{\isaliteral{2C}{\isacharcomma}}\ {\isaliteral{5C3C646F74733E}{\isasymdots}}} to stand for atoms and \isa{{\isaliteral{5C3C70693E}{\isasympi}}{\isaliteral{2C}{\isacharcomma}}\ {\isaliteral{5C3C70693E}{\isasympi}}\isaliteral{5C3C5E697375623E}{}\isactrlisub {\isadigit{1}}{\isaliteral{2C}{\isacharcomma}}\ {\isaliteral{5C3C646F74733E}{\isasymdots}}} to stand for permutations,
  which in Nominal Isabelle have type \isa{perm}. The purpose of atoms is to
  represent variables, be they bound or free. The sorts of atoms can be used
  to represent different kinds of variables, such as the term-, coercion- and
  type-variables in Core-Haskell.  It is assumed that there is an infinite
  supply of atoms for each sort. In the interest of brevity, we shall restrict
  ourselves in what follows to only one sort of atoms.

  Permutations are bijective functions from atoms to atoms that are 
  the identity everywhere except on a finite number of atoms. There is a 
  two-place permutation operation written
  \isa{{\isaliteral{5F}{\isacharunderscore}}\ {\isaliteral{5C3C62756C6C65743E}{\isasymbullet}}\ {\isaliteral{5F}{\isacharunderscore}}} and having the type \isa{perm\ {\isaliteral{5C3C52696768746172726F773E}{\isasymRightarrow}}\ {\isaliteral{5C3C626574613E}{\isasymbeta}}\ {\isaliteral{5C3C52696768746172726F773E}{\isasymRightarrow}}\ {\isaliteral{5C3C626574613E}{\isasymbeta}}}
  where the generic type \isa{{\isaliteral{5C3C626574613E}{\isasymbeta}}} is the type of the object 
  over which the permutation 
  acts. In Nominal Isabelle, the identity permutation is written as \isa{{\isadigit{0}}},
  the composition of two permutations \isa{{\isaliteral{5C3C70693E}{\isasympi}}\isaliteral{5C3C5E697375623E}{}\isactrlisub {\isadigit{1}}} and \isa{{\isaliteral{5C3C70693E}{\isasympi}}\isaliteral{5C3C5E697375623E}{}\isactrlisub {\isadigit{2}}} as \mbox{\isa{{\isaliteral{5C3C70693E}{\isasympi}}\isaliteral{5C3C5E697375623E}{}\isactrlisub {\isadigit{1}}\ {\isaliteral{2B}{\isacharplus}}\ {\isaliteral{5C3C70693E}{\isasympi}}\isaliteral{5C3C5E697375623E}{}\isactrlisub {\isadigit{2}}}} 
  (even if this operation is non-commutative), 
  and the inverse permutation of \isa{{\isaliteral{5C3C70693E}{\isasympi}}} as \isa{{\isaliteral{2D}{\isacharminus}}\ {\isaliteral{5C3C70693E}{\isasympi}}}. The permutation
  operation is defined over Isabelle/HOL's type-hierarchy \cite{HuffmanUrban10};
  for example permutations acting on atoms, products, lists, permutations, sets, 
  functions and booleans are given by:
  
  \begin{equation}\label{permute}
  \mbox{\begin{tabular}{@ {}c@ {\hspace{10mm}}c@ {}}
  \begin{tabular}{@ {}l@ {}}
  \isa{{\isaliteral{5C3C70693E}{\isasympi}}\ {\isaliteral{5C3C62756C6C65743E}{\isasymbullet}}\ a\ {\isaliteral{5C3C65717569763E}{\isasymequiv}}\ {\isaliteral{5C3C70693E}{\isasympi}}\ a}\\
  \isa{{\isaliteral{5C3C70693E}{\isasympi}}\ {\isaliteral{5C3C62756C6C65743E}{\isasymbullet}}\ {\isaliteral{28}{\isacharparenleft}}x{\isaliteral{2C}{\isacharcomma}}\ y{\isaliteral{29}{\isacharparenright}}\ {\isaliteral{5C3C65717569763E}{\isasymequiv}}\ {\isaliteral{28}{\isacharparenleft}}{\isaliteral{5C3C70693E}{\isasympi}}\ {\isaliteral{5C3C62756C6C65743E}{\isasymbullet}}\ x{\isaliteral{2C}{\isacharcomma}}\ {\isaliteral{5C3C70693E}{\isasympi}}\ {\isaliteral{5C3C62756C6C65743E}{\isasymbullet}}\ y{\isaliteral{29}{\isacharparenright}}}\\[2mm]
  \isa{{\isaliteral{5C3C70693E}{\isasympi}}\ {\isaliteral{5C3C62756C6C65743E}{\isasymbullet}}\ {\isaliteral{5B}{\isacharbrackleft}}{\isaliteral{5D}{\isacharbrackright}}\ {\isaliteral{5C3C65717569763E}{\isasymequiv}}\ {\isaliteral{5B}{\isacharbrackleft}}{\isaliteral{5D}{\isacharbrackright}}}\\
  \isa{{\isaliteral{5C3C70693E}{\isasympi}}\ {\isaliteral{5C3C62756C6C65743E}{\isasymbullet}}\ {\isaliteral{28}{\isacharparenleft}}x{\isaliteral{3A}{\isacharcolon}}{\isaliteral{3A}{\isacharcolon}}xs{\isaliteral{29}{\isacharparenright}}\ {\isaliteral{5C3C65717569763E}{\isasymequiv}}\ {\isaliteral{28}{\isacharparenleft}}{\isaliteral{5C3C70693E}{\isasympi}}\ {\isaliteral{5C3C62756C6C65743E}{\isasymbullet}}\ x{\isaliteral{29}{\isacharparenright}}{\isaliteral{3A}{\isacharcolon}}{\isaliteral{3A}{\isacharcolon}}{\isaliteral{28}{\isacharparenleft}}{\isaliteral{5C3C70693E}{\isasympi}}\ {\isaliteral{5C3C62756C6C65743E}{\isasymbullet}}\ xs{\isaliteral{29}{\isacharparenright}}}\\
  \end{tabular} &
  \begin{tabular}{@ {}l@ {}}
  \isa{{\isaliteral{5C3C70693E}{\isasympi}}\ {\isaliteral{5C3C62756C6C65743E}{\isasymbullet}}\ {\isaliteral{5C3C70693E}{\isasympi}}{\isaliteral{27}{\isacharprime}}\ {\isaliteral{5C3C65717569763E}{\isasymequiv}}\ {\isaliteral{5C3C70693E}{\isasympi}}\ {\isaliteral{2B}{\isacharplus}}\ {\isaliteral{5C3C70693E}{\isasympi}}{\isaliteral{27}{\isacharprime}}\ {\isaliteral{2D}{\isacharminus}}\ {\isaliteral{5C3C70693E}{\isasympi}}}\\
  \isa{{\isaliteral{5C3C70693E}{\isasympi}}\ {\isaliteral{5C3C62756C6C65743E}{\isasymbullet}}\ X\ {\isaliteral{5C3C65717569763E}{\isasymequiv}}\ {\isaliteral{7B}{\isacharbraceleft}}{\isaliteral{5C3C70693E}{\isasympi}}\ {\isaliteral{5C3C62756C6C65743E}{\isasymbullet}}\ x\ {\isaliteral{7C}{\isacharbar}}\ x\ {\isaliteral{5C3C696E3E}{\isasymin}}\ X{\isaliteral{7D}{\isacharbraceright}}}\\
  \isa{{\isaliteral{5C3C70693E}{\isasympi}}\ {\isaliteral{5C3C62756C6C65743E}{\isasymbullet}}\ f\ {\isaliteral{5C3C65717569763E}{\isasymequiv}}\ {\isaliteral{5C3C6C616D6264613E}{\isasymlambda}}x{\isaliteral{2E}{\isachardot}}\ {\isaliteral{5C3C70693E}{\isasympi}}\ {\isaliteral{5C3C62756C6C65743E}{\isasymbullet}}\ {\isaliteral{28}{\isacharparenleft}}f\ {\isaliteral{28}{\isacharparenleft}}{\isaliteral{2D}{\isacharminus}}\ {\isaliteral{5C3C70693E}{\isasympi}}\ {\isaliteral{5C3C62756C6C65743E}{\isasymbullet}}\ x{\isaliteral{29}{\isacharparenright}}{\isaliteral{29}{\isacharparenright}}}\\
  \isa{{\isaliteral{5C3C70693E}{\isasympi}}\ {\isaliteral{5C3C62756C6C65743E}{\isasymbullet}}\ b\ {\isaliteral{5C3C65717569763E}{\isasymequiv}}\ b}
  \end{tabular}
  \end{tabular}}
  \end{equation}\smallskip
  
  \noindent
  Concrete permutations in Nominal Isabelle are built up from swappings, 
  written as \mbox{\isa{{\isaliteral{28}{\isacharparenleft}}a\ b{\isaliteral{29}{\isacharparenright}}}}, which are permutations that behave 
  as follows:
  
  \[
  \isa{{\isaliteral{28}{\isacharparenleft}}a\ b{\isaliteral{29}{\isacharparenright}}\ {\isaliteral{3D}{\isacharequal}}\ {\isaliteral{5C3C6C616D6264613E}{\isasymlambda}}c{\isaliteral{2E}{\isachardot}}\ if\ a\ {\isaliteral{3D}{\isacharequal}}\ c\ then\ b\ else\ if\ b\ {\isaliteral{3D}{\isacharequal}}\ c\ then\ a\ else\ c}
  \]\smallskip

  The most original aspect of the nominal logic work of Pitts is a general
  definition for the notion of the `set of free variables of an object \isa{x}'.  This notion, written \isa{supp\ x}, is general in the sense that
  it applies not only to lambda-terms (alpha-equated or not), but also to lists,
  products, sets and even functions. Its definition depends only on the
  permutation operation and on the notion of equality defined for the type of
  \isa{x}, namely:
  
  \begin{equation}\label{suppdef}
  \isa{supp\ x\ {\isaliteral{5C3C65717569763E}{\isasymequiv}}\ {\isaliteral{7B}{\isacharbraceleft}}a\ {\isaliteral{7C}{\isacharbar}}\ infinite\ {\isaliteral{7B}{\isacharbraceleft}}b\ {\isaliteral{7C}{\isacharbar}}\ {\isaliteral{28}{\isacharparenleft}}a\ b{\isaliteral{29}{\isacharparenright}}\ {\isaliteral{5C3C62756C6C65743E}{\isasymbullet}}\ x\ {\isaliteral{5C3C6E6F7465713E}{\isasymnoteq}}\ x{\isaliteral{7D}{\isacharbraceright}}{\isaliteral{7D}{\isacharbraceright}}}
  \end{equation}\smallskip

  \noindent
  There is also the derived notion for when an atom \isa{a} is \emph{fresh}
  for an \isa{x}, defined as 

  \[
  \isa{a\ {\isaliteral{23}{\isacharhash}}\ x\ {\isaliteral{5C3C65717569763E}{\isasymequiv}}\ a\ {\isaliteral{5C3C6E6F74696E3E}{\isasymnotin}}\ supp\ x}
  \]\smallskip

  \noindent
  We use for sets of atoms the abbreviation 
  \isa{as\ {\isaliteral{23}{\isacharhash}}\isaliteral{5C3C5E7375703E}{}\isactrlsup {\isaliteral{2A}{\isacharasterisk}}\ x}, defined as 
  \isa{{\isaliteral{5C3C666F72616C6C3E}{\isasymforall}}a{\isaliteral{5C3C696E3E}{\isasymin}}as{\isaliteral{2E}{\isachardot}}\ a\ {\isaliteral{23}{\isacharhash}}\ x}.
  A striking consequence of these definitions is that we can prove
  without knowing anything about the structure of \isa{x} that
  swapping two fresh atoms, say \isa{a} and \isa{b}, leaves 
  \isa{x} unchanged, namely 
  
  \begin{prop}\label{swapfreshfresh}
  If \isa{a\ {\isaliteral{23}{\isacharhash}}\ x} and \isa{b\ {\isaliteral{23}{\isacharhash}}\ x}
  then \isa{{\isaliteral{28}{\isacharparenleft}}a\ b{\isaliteral{29}{\isacharparenright}}\ {\isaliteral{5C3C62756C6C65743E}{\isasymbullet}}\ x\ {\isaliteral{3D}{\isacharequal}}\ x}.
  \end{prop}
  
  While often the support of an object can be relatively easily 
  described, for example for atoms, products, lists, function applications, 
  booleans and permutations as follows
  
  \begin{equation}\label{supps}\mbox{
  \begin{tabular}{c@ {\hspace{10mm}}c}
  \begin{tabular}{rcl}
  \isa{supp\ a} & $=$ & \isa{{\isaliteral{7B}{\isacharbraceleft}}a{\isaliteral{7D}{\isacharbraceright}}}\\
  \isa{supp\ {\isaliteral{28}{\isacharparenleft}}x{\isaliteral{2C}{\isacharcomma}}\ y{\isaliteral{29}{\isacharparenright}}} & $=$ & \isa{supp\ x\ {\isaliteral{5C3C756E696F6E3E}{\isasymunion}}\ supp\ y}\\
  \isa{supp\ {\isaliteral{5B}{\isacharbrackleft}}{\isaliteral{5D}{\isacharbrackright}}} & $=$ & \isa{{\isaliteral{5C3C656D7074797365743E}{\isasymemptyset}}}\\
  \isa{supp\ {\isaliteral{28}{\isacharparenleft}}x{\isaliteral{3A}{\isacharcolon}}{\isaliteral{3A}{\isacharcolon}}xs{\isaliteral{29}{\isacharparenright}}} & $=$ & \isa{supp\ x\ {\isaliteral{5C3C756E696F6E3E}{\isasymunion}}\ supp\ xs}\\
  \end{tabular}
  &
  \begin{tabular}{rcl}
  \isa{supp\ {\isaliteral{28}{\isacharparenleft}}f\ x{\isaliteral{29}{\isacharparenright}}} & \isa{{\isaliteral{5C3C73756273657465713E}{\isasymsubseteq}}} & \isa{supp\ f\ {\isaliteral{5C3C756E696F6E3E}{\isasymunion}}\ supp\ x}\\
  \isa{supp\ b} & $=$ & \isa{{\isaliteral{5C3C656D7074797365743E}{\isasymemptyset}}}\\
  \isa{supp\ {\isaliteral{5C3C70693E}{\isasympi}}} & $=$ & \isa{{\isaliteral{7B}{\isacharbraceleft}}a\ {\isaliteral{7C}{\isacharbar}}\ {\isaliteral{5C3C70693E}{\isasympi}}\ {\isaliteral{5C3C62756C6C65743E}{\isasymbullet}}\ a\ {\isaliteral{5C3C6E6F7465713E}{\isasymnoteq}}\ a{\isaliteral{7D}{\isacharbraceright}}}
  \end{tabular}
  \end{tabular}}
  \end{equation}\smallskip
  
  \noindent 
  in some cases it can be difficult to characterise the support precisely, and
  only an approximation can be established (as for function applications
  above). Reasoning about such approximations can be simplified with the
  notion \emph{supports}, defined as follows:
  
  \begin{defi}
  A set \isa{S} \emph{supports} \isa{x}, if for all atoms \isa{a} and \isa{b}
  not in \isa{S} we have \isa{{\isaliteral{28}{\isacharparenleft}}a\ b{\isaliteral{29}{\isacharparenright}}\ {\isaliteral{5C3C62756C6C65743E}{\isasymbullet}}\ x\ {\isaliteral{3D}{\isacharequal}}\ x}.
  \end{defi}
  
  \noindent
  The main point of \isa{supports} is that we can establish the following 
  two properties.
  
  \begin{prop}\label{supportsprop}
  Given a set \isa{bs} of atoms.\\
  {\it (i)} If \isa{bs\ supports\ x}
  and \isa{finite\ bs} then 
  \isa{supp\ x\ {\isaliteral{5C3C73756273657465713E}{\isasymsubseteq}}\ bs}.\\
  {\it (ii)} \isa{{\isaliteral{28}{\isacharparenleft}}supp\ x{\isaliteral{29}{\isacharparenright}}\ supports\ x}.
  \end{prop}
  
  Another important notion in the nominal logic work is \emph{equivariance}.
  For a function \isa{f} to be equivariant 
  it is required that every permutation leaves \isa{f} unchanged, that is
  
  \begin{equation}\label{equivariancedef}
  \isa{{\isaliteral{5C3C666F72616C6C3E}{\isasymforall}}{\isaliteral{5C3C70693E}{\isasympi}}{\isaliteral{2E}{\isachardot}}\ {\isaliteral{5C3C70693E}{\isasympi}}\ {\isaliteral{5C3C62756C6C65743E}{\isasymbullet}}\ f\ {\isaliteral{3D}{\isacharequal}}\ f}\;.
  \end{equation}\smallskip
  
  \noindent
  If a function is of type \isa{{\isaliteral{5C3C616C7068613E}{\isasymalpha}}\ {\isaliteral{5C3C52696768746172726F773E}{\isasymRightarrow}}\ {\isaliteral{5C3C626574613E}{\isasymbeta}}}, say, this definition is equivalent to 
  the fact that a permutation applied to the application
  \isa{f\ x} can be moved to the argument \isa{x}. That means for 
  such functions, we have for all permutations \isa{{\isaliteral{5C3C70693E}{\isasympi}}}:
  
  \begin{equation}\label{equivariance}
  \isa{{\isaliteral{5C3C70693E}{\isasympi}}\ {\isaliteral{5C3C62756C6C65743E}{\isasymbullet}}\ f\ {\isaliteral{3D}{\isacharequal}}\ f} \;\;\;\;\textit{if and only if}\;\;\;\;
  \isa{{\isaliteral{5C3C666F72616C6C3E}{\isasymforall}}x{\isaliteral{2E}{\isachardot}}\ {\isaliteral{5C3C70693E}{\isasympi}}\ {\isaliteral{5C3C62756C6C65743E}{\isasymbullet}}\ {\isaliteral{28}{\isacharparenleft}}f\ x{\isaliteral{29}{\isacharparenright}}\ {\isaliteral{3D}{\isacharequal}}\ f\ {\isaliteral{28}{\isacharparenleft}}{\isaliteral{5C3C70693E}{\isasympi}}\ {\isaliteral{5C3C62756C6C65743E}{\isasymbullet}}\ x{\isaliteral{29}{\isacharparenright}}}\;.
  \end{equation}\smallskip
   
  \noindent
  There is
  also a similar property for relations, which are in HOL functions of type \isa{{\isaliteral{5C3C616C7068613E}{\isasymalpha}}\ {\isaliteral{5C3C52696768746172726F773E}{\isasymRightarrow}}\ {\isaliteral{5C3C626574613E}{\isasymbeta}}\ {\isaliteral{5C3C52696768746172726F773E}{\isasymRightarrow}}\ bool}.
  Suppose a relation \isa{R}, then for all permutations \isa{{\isaliteral{5C3C70693E}{\isasympi}}}:
  
  \[
  \isa{{\isaliteral{5C3C70693E}{\isasympi}}\ {\isaliteral{5C3C62756C6C65743E}{\isasymbullet}}\ R\ {\isaliteral{3D}{\isacharequal}}\ R} \;\;\;\;\textit{if and only if}\;\;\;\;
  \isa{{\isaliteral{5C3C666F72616C6C3E}{\isasymforall}}x\ y{\isaliteral{2E}{\isachardot}}}~~\isa{x\ R\ y} \;\textit{implies}\; \isa{{\isaliteral{28}{\isacharparenleft}}{\isaliteral{5C3C70693E}{\isasympi}}\ {\isaliteral{5C3C62756C6C65743E}{\isasymbullet}}\ x{\isaliteral{29}{\isacharparenright}}\ R\ {\isaliteral{28}{\isacharparenleft}}{\isaliteral{5C3C70693E}{\isasympi}}\ {\isaliteral{5C3C62756C6C65743E}{\isasymbullet}}\ y{\isaliteral{29}{\isacharparenright}}}\;.
  \]\smallskip

  \noindent
  Note that from property \eqref{equivariancedef} and the definition of \isa{supp}, we 
  can easily deduce that for a function being equivariant is equivalent to having empty support.

  Using freshness, the nominal logic work provides us with general means for renaming 
  binders. 
  
  \noindent
  While in the older version of Nominal Isabelle, we used extensively 
  Proposition~\ref{swapfreshfresh} to rename single binders, this property 
  proved too unwieldy for dealing with multiple binders. For such binders the 
  following generalisations turned out to be easier to use.

  \begin{prop}\label{supppermeq}
  \isa{{\normalsize{}If\,}\ supp\ x\ {\isaliteral{23}{\isacharhash}}\isaliteral{5C3C5E7375703E}{}\isactrlsup {\isaliteral{2A}{\isacharasterisk}}\ {\isaliteral{5C3C70693E}{\isasympi}}\ {\normalsize \,then\,}\ {\isaliteral{5C3C70693E}{\isasympi}}\ {\isaliteral{5C3C62756C6C65743E}{\isasymbullet}}\ x\ {\isaliteral{3D}{\isacharequal}}\ x{\isaliteral{2E}{\isachardot}}}
  \end{prop}

  \begin{prop}\label{avoiding}
  For a finite set \isa{as} and a finitely supported \isa{x} with
  \isa{as\ {\isaliteral{23}{\isacharhash}}\isaliteral{5C3C5E7375703E}{}\isactrlsup {\isaliteral{2A}{\isacharasterisk}}\ x} and also a finitely supported \isa{c}, there
  exists a permutation \isa{{\isaliteral{5C3C70693E}{\isasympi}}} such that \isa{{\isaliteral{5C3C70693E}{\isasympi}}\ {\isaliteral{5C3C62756C6C65743E}{\isasymbullet}}\ as\ {\isaliteral{23}{\isacharhash}}\isaliteral{5C3C5E7375703E}{}\isactrlsup {\isaliteral{2A}{\isacharasterisk}}\ c} and
  \isa{supp\ x\ {\isaliteral{23}{\isacharhash}}\isaliteral{5C3C5E7375703E}{}\isactrlsup {\isaliteral{2A}{\isacharasterisk}}\ {\isaliteral{5C3C70693E}{\isasympi}}}.
  \end{prop}

  \noindent
  The idea behind the second property is that given a finite set \isa{as}
  of binders (being bound, or fresh, in \isa{x} is ensured by the
  assumption \isa{as\ {\isaliteral{23}{\isacharhash}}\isaliteral{5C3C5E7375703E}{}\isactrlsup {\isaliteral{2A}{\isacharasterisk}}\ x}), then there exists a permutation \isa{{\isaliteral{5C3C70693E}{\isasympi}}} such that
  the renamed binders \isa{{\isaliteral{5C3C70693E}{\isasympi}}\ {\isaliteral{5C3C62756C6C65743E}{\isasymbullet}}\ as} avoid \isa{c} (which can be arbitrarily chosen
  as long as it is finitely supported) and also \isa{{\isaliteral{5C3C70693E}{\isasympi}}} does not affect anything
  in the support of \isa{x} (that is \isa{supp\ x\ {\isaliteral{23}{\isacharhash}}\isaliteral{5C3C5E7375703E}{}\isactrlsup {\isaliteral{2A}{\isacharasterisk}}\ {\isaliteral{5C3C70693E}{\isasympi}}}). The last 
  fact and Property~\ref{supppermeq} allow us to `rename' just the binders 
  \isa{as} in \isa{x}, because \isa{{\isaliteral{5C3C70693E}{\isasympi}}\ {\isaliteral{5C3C62756C6C65743E}{\isasymbullet}}\ x\ {\isaliteral{3D}{\isacharequal}}\ x}. 

  Note that \isa{supp\ x\ {\isaliteral{23}{\isacharhash}}\isaliteral{5C3C5E7375703E}{}\isactrlsup {\isaliteral{2A}{\isacharasterisk}}\ {\isaliteral{5C3C70693E}{\isasympi}}}
  is equivalent with \isa{supp\ {\isaliteral{5C3C70693E}{\isasympi}}\ {\isaliteral{23}{\isacharhash}}\isaliteral{5C3C5E7375703E}{}\isactrlsup {\isaliteral{2A}{\isacharasterisk}}\ x}, which means we could also formulate 
  Propositions \ref{supppermeq} and \ref{avoiding} in the other `direction'; however the 
  reasoning infrastructure of Nominal Isabelle is set up so that it provides more
  automation for the formulation given above.

  Most properties given in this section are described in detail in \cite{HuffmanUrban10}
  and all are formalised in Isabelle/HOL. In the next sections we will make 
  use of these properties in order to define alpha-equivalence in 
  the presence of multiple binders.%
\end{isamarkuptext}%
\isamarkuptrue%
\isamarkupsection{Abstractions\label{sec:binders}%
}
\isamarkuptrue%
\begin{isamarkuptext}%
In Nominal Isabelle, the user is expected to write down a specification of a
  term-calculus and then a reasoning infrastructure is automatically derived
  from this specification (remember that Nominal Isabelle is a definitional
  extension of Isabelle/HOL, which does not introduce any new axioms).

  In order to keep our work with deriving the reasoning infrastructure
  manageable, we will wherever possible state definitions and perform proofs
  on the `user-level' of Isabelle/HOL, as opposed to writing custom ML-code that
  generates them anew for each specification. 
  To that end, we will consider
  first pairs \isa{{\isaliteral{28}{\isacharparenleft}}as{\isaliteral{2C}{\isacharcomma}}\ x{\isaliteral{29}{\isacharparenright}}} of type \isa{{\isaliteral{28}{\isacharparenleft}}atom\ set{\isaliteral{29}{\isacharparenright}}\ {\isaliteral{5C3C74696D65733E}{\isasymtimes}}\ {\isaliteral{5C3C626574613E}{\isasymbeta}}}.  These pairs
  are intended to represent the abstraction, or binding, of the set of atoms \isa{as} in the body \isa{x}.

  The first question we have to answer is when two pairs \isa{{\isaliteral{28}{\isacharparenleft}}as{\isaliteral{2C}{\isacharcomma}}\ x{\isaliteral{29}{\isacharparenright}}} and
  \isa{{\isaliteral{28}{\isacharparenleft}}bs{\isaliteral{2C}{\isacharcomma}}\ y{\isaliteral{29}{\isacharparenright}}} are alpha-equivalent? (For the moment we are interested in
  the notion of alpha-equivalence that is \emph{not} preserved by adding
  vacuous binders.) To answer this question, we identify four conditions: {\it (i)}
  given a free-atom function \isa{fa} of type \mbox{\isa{{\isaliteral{5C3C626574613E}{\isasymbeta}}\ {\isaliteral{5C3C52696768746172726F773E}{\isasymRightarrow}}\ atom\ set}}, then \isa{{\isaliteral{28}{\isacharparenleft}}as{\isaliteral{2C}{\isacharcomma}}\ x{\isaliteral{29}{\isacharparenright}}} and \isa{{\isaliteral{28}{\isacharparenleft}}bs{\isaliteral{2C}{\isacharcomma}}\ y{\isaliteral{29}{\isacharparenright}}} need to have the same set of free
  atoms; moreover there must be a permutation \isa{{\isaliteral{5C3C70693E}{\isasympi}}} such that {\it
  (ii)} \isa{{\isaliteral{5C3C70693E}{\isasympi}}} leaves the free atoms of \isa{{\isaliteral{28}{\isacharparenleft}}as{\isaliteral{2C}{\isacharcomma}}\ x{\isaliteral{29}{\isacharparenright}}} and \isa{{\isaliteral{28}{\isacharparenleft}}bs{\isaliteral{2C}{\isacharcomma}}\ y{\isaliteral{29}{\isacharparenright}}} unchanged, but
  {\it (iii)} `moves' their bound names so that we obtain modulo a relation,
  say \mbox{\isa{{\isaliteral{5F}{\isacharunderscore}}\ R\ {\isaliteral{5F}{\isacharunderscore}}}}, two equivalent terms. We also require that {\it (iv)}
  \isa{{\isaliteral{5C3C70693E}{\isasympi}}} makes the sets of abstracted atoms \isa{as} and \isa{bs} equal. The
  requirements {\it (i)} to {\it (iv)} can be stated formally as:

  \begin{defi}[Alpha-Equivalence for Set-Bindings]\label{alphaset}\mbox{}\\
  \begin{tabular}{@ {\hspace{10mm}}l@ {\hspace{5mm}}rl}  
  \isa{{\isaliteral{28}{\isacharparenleft}}as{\isaliteral{2C}{\isacharcomma}}\ x{\isaliteral{29}{\isacharparenright}}\ {\isaliteral{5C3C617070726F783E}{\isasymapprox}}\,\raisebox{-1pt}{\makebox[0mm][l]{$_{\textit{set}}$}}\isaliteral{5C3C5E627375703E}{}\isactrlbsup R{\isaliteral{2C}{\isacharcomma}}\ fa\isaliteral{5C3C5E657375703E}{}\isactrlesup \ {\isaliteral{28}{\isacharparenleft}}bs{\isaliteral{2C}{\isacharcomma}}\ y{\isaliteral{29}{\isacharparenright}}}\hspace{2mm}\isa{{\isaliteral{5C3C65717569763E}{\isasymequiv}}} & 
    \multicolumn{2}{@ {}l}{if there exists a \isa{{\isaliteral{5C3C70693E}{\isasympi}}} such that:}\\ 
       & \mbox{\it (i)}   & \isa{fa\ x\ {\isaliteral{2D}{\isacharminus}}\ as\ {\isaliteral{3D}{\isacharequal}}\ fa\ y\ {\isaliteral{2D}{\isacharminus}}\ bs}\\
       & \mbox{\it (ii)}  & \isa{fa\ x\ {\isaliteral{2D}{\isacharminus}}\ as\ {\isaliteral{23}{\isacharhash}}\isaliteral{5C3C5E7375703E}{}\isactrlsup {\isaliteral{2A}{\isacharasterisk}}\ {\isaliteral{5C3C70693E}{\isasympi}}}\\
       & \mbox{\it (iii)} &  \isa{{\isaliteral{28}{\isacharparenleft}}{\isaliteral{5C3C70693E}{\isasympi}}\ {\isaliteral{5C3C62756C6C65743E}{\isasymbullet}}\ x{\isaliteral{29}{\isacharparenright}}\ R\ y} \\
       & \mbox{\it (iv)}  & \isa{{\isaliteral{5C3C70693E}{\isasympi}}\ {\isaliteral{5C3C62756C6C65743E}{\isasymbullet}}\ as\ {\isaliteral{3D}{\isacharequal}}\ bs} \\ 
  \end{tabular}
  \end{defi}
 
  \noindent
  Note that the relation is
  dependent on a free-atom function \isa{fa} and a relation \isa{R}. The reason for this extra generality is that we will use
  $\approx_{\,\textit{set}}^{\textit{R}, \textit{fa}}$ for both raw terms and 
  alpha-equated terms. In
  the latter case, \isa{R} will be replaced by equality \isa{{\isaliteral{3D}{\isacharequal}}} and we
  will prove that \isa{fa} is equal to \isa{supp}.

  Definition \ref{alphaset} does not make any distinction between the
  order of abstracted atoms. If we want this, then we can define alpha-equivalence 
  for pairs of the form \mbox{\isa{{\isaliteral{28}{\isacharparenleft}}as{\isaliteral{2C}{\isacharcomma}}\ x{\isaliteral{29}{\isacharparenright}}}} with type \isa{{\isaliteral{28}{\isacharparenleft}}atom\ list{\isaliteral{29}{\isacharparenright}}\ {\isaliteral{5C3C74696D65733E}{\isasymtimes}}\ {\isaliteral{5C3C626574613E}{\isasymbeta}}} 
  as follows
  
  \begin{defi}[Alpha-Equivalence for List-Bindings]\label{alphalist}\mbox{}\\
  \begin{tabular}{@ {\hspace{10mm}}l@ {\hspace{5mm}}rl}  
  \isa{{\isaliteral{28}{\isacharparenleft}}as{\isaliteral{2C}{\isacharcomma}}\ x{\isaliteral{29}{\isacharparenright}}\ {\isaliteral{5C3C617070726F783E}{\isasymapprox}}\,\raisebox{-1pt}{\makebox[0mm][l]{$_{\textit{list}}$}}\isaliteral{5C3C5E627375703E}{}\isactrlbsup R{\isaliteral{2C}{\isacharcomma}}\ fa\isaliteral{5C3C5E657375703E}{}\isactrlesup \ {\isaliteral{28}{\isacharparenleft}}bs{\isaliteral{2C}{\isacharcomma}}\ y{\isaliteral{29}{\isacharparenright}}}\hspace{2mm}\isa{{\isaliteral{5C3C65717569763E}{\isasymequiv}}} &
  \multicolumn{2}{@ {}l}{if there exists a \isa{{\isaliteral{5C3C70693E}{\isasympi}}} such that:}\\ 
         & \mbox{\it (i)}   & \isa{fa\ x\ {\isaliteral{2D}{\isacharminus}}\ set\ as\ {\isaliteral{3D}{\isacharequal}}\ fa\ y\ {\isaliteral{2D}{\isacharminus}}\ set\ bs}\\ 
         & \mbox{\it (ii)}  & \isa{fa\ x\ {\isaliteral{2D}{\isacharminus}}\ set\ as\ {\isaliteral{23}{\isacharhash}}\isaliteral{5C3C5E7375703E}{}\isactrlsup {\isaliteral{2A}{\isacharasterisk}}\ {\isaliteral{5C3C70693E}{\isasympi}}}\\
         & \mbox{\it (iii)} & \isa{{\isaliteral{28}{\isacharparenleft}}{\isaliteral{5C3C70693E}{\isasympi}}\ {\isaliteral{5C3C62756C6C65743E}{\isasymbullet}}\ x{\isaliteral{29}{\isacharparenright}}\ R\ y}\\
         & \mbox{\it (iv)}  & \isa{{\isaliteral{5C3C70693E}{\isasympi}}\ {\isaliteral{5C3C62756C6C65743E}{\isasymbullet}}\ as\ {\isaliteral{3D}{\isacharequal}}\ bs}\\
  \end{tabular}
  \end{defi}
  
  \noindent
  where \isa{set} is the function that coerces a list of atoms into a set of atoms.
  Now the last clause ensures that the order of the binders matters (since \isa{as}
  and \isa{bs} are lists of atoms).

  If we do not want to make any difference between the order of binders \emph{and}
  also allow vacuous binders, that means according to Pitts~\cite{Pitts04} 
  \emph{restrict} atoms, then we keep sets of binders, but drop 
  condition {\it (iv)} in Definition~\ref{alphaset}:

  \begin{defi}[Alpha-Equivalence for Set+-Bindings]\label{alphares}\mbox{}\\
  \begin{tabular}{@ {\hspace{10mm}}l@ {\hspace{5mm}}rl}  
  \isa{{\isaliteral{28}{\isacharparenleft}}as{\isaliteral{2C}{\isacharcomma}}\ x{\isaliteral{29}{\isacharparenright}}\ {\isaliteral{5C3C617070726F783E}{\isasymapprox}}\,\raisebox{-1pt}{\makebox[0mm][l]{$_{\textit{set+}}$}}\isaliteral{5C3C5E627375703E}{}\isactrlbsup R{\isaliteral{2C}{\isacharcomma}}\ fa\isaliteral{5C3C5E657375703E}{}\isactrlesup \ {\isaliteral{28}{\isacharparenleft}}bs{\isaliteral{2C}{\isacharcomma}}\ y{\isaliteral{29}{\isacharparenright}}}\hspace{2mm}\isa{{\isaliteral{5C3C65717569763E}{\isasymequiv}}} &
  \multicolumn{2}{@ {}l}{if there exists a \isa{{\isaliteral{5C3C70693E}{\isasympi}}} such that:}\\ 
             & \mbox{\it (i)}   & \isa{fa\ x\ {\isaliteral{2D}{\isacharminus}}\ as\ {\isaliteral{3D}{\isacharequal}}\ fa\ y\ {\isaliteral{2D}{\isacharminus}}\ bs}\\
             & \mbox{\it (ii)}  & \isa{fa\ x\ {\isaliteral{2D}{\isacharminus}}\ as\ {\isaliteral{23}{\isacharhash}}\isaliteral{5C3C5E7375703E}{}\isactrlsup {\isaliteral{2A}{\isacharasterisk}}\ {\isaliteral{5C3C70693E}{\isasympi}}}\\
             & \mbox{\it (iii)} & \isa{{\isaliteral{28}{\isacharparenleft}}{\isaliteral{5C3C70693E}{\isasympi}}\ {\isaliteral{5C3C62756C6C65743E}{\isasymbullet}}\ x{\isaliteral{29}{\isacharparenright}}\ R\ y}\\
  \end{tabular}
  \end{defi}

  It might be useful to consider first some examples how these definitions
  of alpha-equivalence pan out in practice.  For this consider the case of
  abstracting a set of atoms over types (as in type-schemes). We set
  \isa{R} to be the usual equality \isa{{\isaliteral{3D}{\isacharequal}}} and for \isa{fa{\isaliteral{28}{\isacharparenleft}}T{\isaliteral{29}{\isacharparenright}}} we
  define
  
  \[
  \isa{fa{\isaliteral{28}{\isacharparenleft}}x{\isaliteral{29}{\isacharparenright}}\ {\isaliteral{5C3C65717569763E}{\isasymequiv}}\ {\isaliteral{7B}{\isacharbraceleft}}x{\isaliteral{7D}{\isacharbraceright}}}  \hspace{10mm} \isa{fa{\isaliteral{28}{\isacharparenleft}}T\isaliteral{5C3C5E697375623E}{}\isactrlisub {\isadigit{1}}\ {\isaliteral{5C3C72696768746172726F773E}{\isasymrightarrow}}\ T\isaliteral{5C3C5E697375623E}{}\isactrlisub {\isadigit{2}}{\isaliteral{29}{\isacharparenright}}\ {\isaliteral{5C3C65717569763E}{\isasymequiv}}\ fa{\isaliteral{28}{\isacharparenleft}}T\isaliteral{5C3C5E697375623E}{}\isactrlisub {\isadigit{1}}{\isaliteral{29}{\isacharparenright}}\ {\isaliteral{5C3C756E696F6E3E}{\isasymunion}}\ fa{\isaliteral{28}{\isacharparenleft}}T\isaliteral{5C3C5E697375623E}{}\isactrlisub {\isadigit{2}}{\isaliteral{29}{\isacharparenright}}}
  \]\smallskip

  \noindent
  Now recall the examples shown in \eqref{ex1} and
  \eqref{ex3}. It can be easily checked that \isa{{\isaliteral{28}{\isacharparenleft}}{\isaliteral{7B}{\isacharbraceleft}}x{\isaliteral{2C}{\isacharcomma}}\ y{\isaliteral{7D}{\isacharbraceright}}{\isaliteral{2C}{\isacharcomma}}\ x\ {\isaliteral{5C3C72696768746172726F773E}{\isasymrightarrow}}\ y{\isaliteral{29}{\isacharparenright}}} and
  \isa{{\isaliteral{28}{\isacharparenleft}}{\isaliteral{7B}{\isacharbraceleft}}x{\isaliteral{2C}{\isacharcomma}}\ y{\isaliteral{7D}{\isacharbraceright}}{\isaliteral{2C}{\isacharcomma}}\ y\ {\isaliteral{5C3C72696768746172726F773E}{\isasymrightarrow}}\ x{\isaliteral{29}{\isacharparenright}}} are alpha-equivalent according to
  $\approx_{\,\textit{set}}$ and $\approx_{\,\textit{set+}}$ by taking \isa{{\isaliteral{5C3C70693E}{\isasympi}}} to
  be the swapping \isa{{\isaliteral{28}{\isacharparenleft}}x\ y{\isaliteral{29}{\isacharparenright}}}. In case of \isa{x\ {\isaliteral{5C3C6E6F7465713E}{\isasymnoteq}}\ y}, then \isa{{\isaliteral{28}{\isacharparenleft}}{\isaliteral{5B}{\isacharbrackleft}}x{\isaliteral{2C}{\isacharcomma}}\ y{\isaliteral{5D}{\isacharbrackright}}{\isaliteral{2C}{\isacharcomma}}\ x\ {\isaliteral{5C3C72696768746172726F773E}{\isasymrightarrow}}\ y{\isaliteral{29}{\isacharparenright}}} $\not\approx_{\,\textit{list}}$ \isa{{\isaliteral{28}{\isacharparenleft}}{\isaliteral{5B}{\isacharbrackleft}}y{\isaliteral{2C}{\isacharcomma}}\ x{\isaliteral{5D}{\isacharbrackright}}{\isaliteral{2C}{\isacharcomma}}\ x\ {\isaliteral{5C3C72696768746172726F773E}{\isasymrightarrow}}\ y{\isaliteral{29}{\isacharparenright}}}
  since there is no permutation that makes the lists \isa{{\isaliteral{5B}{\isacharbrackleft}}x{\isaliteral{2C}{\isacharcomma}}\ y{\isaliteral{5D}{\isacharbrackright}}} and
  \isa{{\isaliteral{5B}{\isacharbrackleft}}y{\isaliteral{2C}{\isacharcomma}}\ x{\isaliteral{5D}{\isacharbrackright}}} equal, and also leaves the type \mbox{\isa{x\ {\isaliteral{5C3C72696768746172726F773E}{\isasymrightarrow}}\ y}}
  unchanged. Another example is \isa{{\isaliteral{28}{\isacharparenleft}}{\isaliteral{7B}{\isacharbraceleft}}x{\isaliteral{7D}{\isacharbraceright}}{\isaliteral{2C}{\isacharcomma}}\ x{\isaliteral{29}{\isacharparenright}}} $\approx_{\,\textit{set+}}$
  \isa{{\isaliteral{28}{\isacharparenleft}}{\isaliteral{7B}{\isacharbraceleft}}x{\isaliteral{2C}{\isacharcomma}}\ y{\isaliteral{7D}{\isacharbraceright}}{\isaliteral{2C}{\isacharcomma}}\ x{\isaliteral{29}{\isacharparenright}}} which holds by taking \isa{{\isaliteral{5C3C70693E}{\isasympi}}} to be the identity
  permutation.  However, if \isa{x\ {\isaliteral{5C3C6E6F7465713E}{\isasymnoteq}}\ y}, then \isa{{\isaliteral{28}{\isacharparenleft}}{\isaliteral{7B}{\isacharbraceleft}}x{\isaliteral{7D}{\isacharbraceright}}{\isaliteral{2C}{\isacharcomma}}\ x{\isaliteral{29}{\isacharparenright}}}
  $\not\approx_{\,\textit{set}}$ \isa{{\isaliteral{28}{\isacharparenleft}}{\isaliteral{7B}{\isacharbraceleft}}x{\isaliteral{2C}{\isacharcomma}}\ y{\isaliteral{7D}{\isacharbraceright}}{\isaliteral{2C}{\isacharcomma}}\ x{\isaliteral{29}{\isacharparenright}}} since there is no
  permutation that makes the sets \isa{{\isaliteral{7B}{\isacharbraceleft}}x{\isaliteral{7D}{\isacharbraceright}}} and \isa{{\isaliteral{7B}{\isacharbraceleft}}x{\isaliteral{2C}{\isacharcomma}}\ y{\isaliteral{7D}{\isacharbraceright}}} equal
  (similarly for $\approx_{\,\textit{list}}$).  It can also relatively easily be
  shown that all three notions of alpha-equivalence coincide, if we only
  abstract a single atom. In this case they also agree with the alpha-equivalence
  used in older versions of Nominal Isabelle \cite{Urban08}.\footnote{We omit a
  proof of this fact since the details are hairy and not really important for the
  purpose of this paper.}

  In the rest of this section we are going to show that the alpha-equivalences
  really lead to abstractions where some atoms are bound (or more precisely
  removed from the support).  For this we will consider three abstraction
  types that are quotients of the relations

  \begin{equation}
  \begin{array}{r}
  \isa{{\isaliteral{28}{\isacharparenleft}}as{\isaliteral{2C}{\isacharcomma}}\ x{\isaliteral{29}{\isacharparenright}}\ {\isaliteral{5C3C617070726F783E}{\isasymapprox}}\,\raisebox{-1pt}{\makebox[0mm][l]{$_{\textit{set}}$}}\isaliteral{5C3C5E627375703E}{}\isactrlbsup {\isaliteral{3D}{\isacharequal}}{\isaliteral{2C}{\isacharcomma}}\ supp\isaliteral{5C3C5E657375703E}{}\isactrlesup \ {\isaliteral{28}{\isacharparenleft}}bs{\isaliteral{2C}{\isacharcomma}}\ y{\isaliteral{29}{\isacharparenright}}}\smallskip\\
  \isa{{\isaliteral{28}{\isacharparenleft}}as{\isaliteral{2C}{\isacharcomma}}\ x{\isaliteral{29}{\isacharparenright}}\ {\isaliteral{5C3C617070726F783E}{\isasymapprox}}\,\raisebox{-1pt}{\makebox[0mm][l]{$_{\textit{set+}}$}}\isaliteral{5C3C5E627375703E}{}\isactrlbsup {\isaliteral{3D}{\isacharequal}}{\isaliteral{2C}{\isacharcomma}}\ supp\isaliteral{5C3C5E657375703E}{}\isactrlesup \ {\isaliteral{28}{\isacharparenleft}}bs{\isaliteral{2C}{\isacharcomma}}\ y{\isaliteral{29}{\isacharparenright}}}\smallskip\\
  \isa{{\isaliteral{28}{\isacharparenleft}}as{\isaliteral{2C}{\isacharcomma}}\ x{\isaliteral{29}{\isacharparenright}}\ {\isaliteral{5C3C617070726F783E}{\isasymapprox}}\,\raisebox{-1pt}{\makebox[0mm][l]{$_{\textit{list}}$}}\isaliteral{5C3C5E627375703E}{}\isactrlbsup {\isaliteral{3D}{\isacharequal}}{\isaliteral{2C}{\isacharcomma}}\ supp\isaliteral{5C3C5E657375703E}{}\isactrlesup \ {\isaliteral{28}{\isacharparenleft}}bs{\isaliteral{2C}{\isacharcomma}}\ y{\isaliteral{29}{\isacharparenright}}}\\
  \end{array}
  \end{equation}\smallskip
  
  \noindent
  Note that in these relations we replaced the free-atom function \isa{fa}
  with \isa{supp} and the relation \isa{R} with equality. We can show
  the following two properties:

  \begin{lem}\label{alphaeq} 
  The relations $\approx_{\,\textit{set}}^{=, \textit{supp}}$, 
  $\approx_{\,\textit{set+}}^{=, \textit{supp}}$
  and $\approx_{\,\textit{list}}^{=, \textit{supp}}$ are 
  equivalence relations and equivariant. 
  \end{lem}

  \begin{proof}
  Reflexivity is by taking \isa{{\isaliteral{5C3C70693E}{\isasympi}}} to be \isa{{\isadigit{0}}}. For symmetry we have
  a permutation \isa{{\isaliteral{5C3C70693E}{\isasympi}}} and for the proof obligation take \isa{{\isaliteral{2D}{\isacharminus}}{\isaliteral{5C3C70693E}{\isasympi}}}. In case of transitivity, we have two permutations \isa{{\isaliteral{5C3C70693E}{\isasympi}}\isaliteral{5C3C5E697375623E}{}\isactrlisub {\isadigit{1}}}
  and \isa{{\isaliteral{5C3C70693E}{\isasympi}}\isaliteral{5C3C5E697375623E}{}\isactrlisub {\isadigit{2}}}, and for the proof obligation use \isa{{\isaliteral{5C3C70693E}{\isasympi}}\isaliteral{5C3C5E697375623E}{}\isactrlisub {\isadigit{1}}\ {\isaliteral{2B}{\isacharplus}}\ {\isaliteral{5C3C70693E}{\isasympi}}\isaliteral{5C3C5E697375623E}{}\isactrlisub {\isadigit{2}}}. Equivariance means \isa{{\isaliteral{28}{\isacharparenleft}}{\isaliteral{5C3C70693E}{\isasympi}}\ {\isaliteral{5C3C62756C6C65743E}{\isasymbullet}}\ as{\isaliteral{2C}{\isacharcomma}}\ {\isaliteral{5C3C70693E}{\isasympi}}\ {\isaliteral{5C3C62756C6C65743E}{\isasymbullet}}\ x{\isaliteral{29}{\isacharparenright}}\ {\isaliteral{5C3C617070726F783E}{\isasymapprox}}\,\raisebox{-1pt}{\makebox[0mm][l]{$_{\textit{set}}$}}\isaliteral{5C3C5E627375703E}{}\isactrlbsup {\isaliteral{3D}{\isacharequal}}{\isaliteral{2C}{\isacharcomma}}\ supp\isaliteral{5C3C5E657375703E}{}\isactrlesup \ {\isaliteral{28}{\isacharparenleft}}{\isaliteral{5C3C70693E}{\isasympi}}\ {\isaliteral{5C3C62756C6C65743E}{\isasymbullet}}\ bs{\isaliteral{2C}{\isacharcomma}}\ {\isaliteral{5C3C70693E}{\isasympi}}\ {\isaliteral{5C3C62756C6C65743E}{\isasymbullet}}\ y{\isaliteral{29}{\isacharparenright}}} holds provided \mbox{\isa{{\isaliteral{28}{\isacharparenleft}}as{\isaliteral{2C}{\isacharcomma}}\ x{\isaliteral{29}{\isacharparenright}}\ {\isaliteral{5C3C617070726F783E}{\isasymapprox}}\,\raisebox{-1pt}{\makebox[0mm][l]{$_{\textit{set}}$}}\isaliteral{5C3C5E627375703E}{}\isactrlbsup {\isaliteral{3D}{\isacharequal}}{\isaliteral{2C}{\isacharcomma}}\ supp\isaliteral{5C3C5E657375703E}{}\isactrlesup \ {\isaliteral{28}{\isacharparenleft}}bs{\isaliteral{2C}{\isacharcomma}}\ y{\isaliteral{29}{\isacharparenright}}}} holds. From the assumption we
  have a permutation \isa{{\isaliteral{5C3C70693E}{\isasympi}}{\isaliteral{27}{\isacharprime}}} and for the proof obligation use \isa{{\isaliteral{5C3C70693E}{\isasympi}}\ {\isaliteral{5C3C62756C6C65743E}{\isasymbullet}}\ {\isaliteral{5C3C70693E}{\isasympi}}{\isaliteral{27}{\isacharprime}}}. To show equivariance, we need to `pull out' the permutations,
  which is possible since all operators, namely as \isa{{\isaliteral{23}{\isacharhash}}\isaliteral{5C3C5E7375703E}{}\isactrlsup {\isaliteral{2A}{\isacharasterisk}}{\isaliteral{2C}{\isacharcomma}}\ {\isaliteral{2D}{\isacharminus}}{\isaliteral{2C}{\isacharcomma}}\ {\isaliteral{3D}{\isacharequal}}{\isaliteral{2C}{\isacharcomma}}\ {\isaliteral{5C3C62756C6C65743E}{\isasymbullet}}{\isaliteral{2C}{\isacharcomma}}\ set} and \isa{supp}, are equivariant (see
  \cite{HuffmanUrban10}). Finally, we apply the permutation operation on
  booleans.
  \end{proof}

  \noindent
  Recall the picture shown in \eqref{picture} about new types in HOL.
  The lemma above allows us to use our quotient package for introducing 
  new types \isa{{\isaliteral{5C3C626574613E}{\isasymbeta}}\ abs\isaliteral{5C3C5E627375623E}{}\isactrlbsub set\isaliteral{5C3C5E657375623E}{}\isactrlesub }, \isa{{\isaliteral{5C3C626574613E}{\isasymbeta}}\ abs\isaliteral{5C3C5E627375623E}{}\isactrlbsub set{\isaliteral{2B}{\isacharplus}}\isaliteral{5C3C5E657375623E}{}\isactrlesub } and \isa{{\isaliteral{5C3C626574613E}{\isasymbeta}}\ abs\isaliteral{5C3C5E627375623E}{}\isactrlbsub list\isaliteral{5C3C5E657375623E}{}\isactrlesub }
  representing alpha-equivalence classes of pairs of type 
  \isa{{\isaliteral{28}{\isacharparenleft}}atom\ set{\isaliteral{29}{\isacharparenright}}\ {\isaliteral{5C3C74696D65733E}{\isasymtimes}}\ {\isaliteral{5C3C626574613E}{\isasymbeta}}} (in the first two cases) and of type \isa{{\isaliteral{28}{\isacharparenleft}}atom\ list{\isaliteral{29}{\isacharparenright}}\ {\isaliteral{5C3C74696D65733E}{\isasymtimes}}\ {\isaliteral{5C3C626574613E}{\isasymbeta}}}
  (in the third case). 
  The elements in these types will be, respectively, written as
  
  \[
  \isa{{\isaliteral{5B}{\isacharbrackleft}}as{\isaliteral{5D}{\isacharbrackright}}\isaliteral{5C3C5E627375623E}{}\isactrlbsub set\isaliteral{5C3C5E657375623E}{}\isactrlesub {\isaliteral{2E}{\isachardot}}x} \hspace{10mm} 
  \isa{{\isaliteral{5B}{\isacharbrackleft}}as{\isaliteral{5D}{\isacharbrackright}}\isaliteral{5C3C5E627375623E}{}\isactrlbsub set{\isaliteral{2B}{\isacharplus}}\isaliteral{5C3C5E657375623E}{}\isactrlesub {\isaliteral{2E}{\isachardot}}x} \hspace{10mm}
  \isa{{\isaliteral{5B}{\isacharbrackleft}}as{\isaliteral{5D}{\isacharbrackright}}\isaliteral{5C3C5E627375623E}{}\isactrlbsub list\isaliteral{5C3C5E657375623E}{}\isactrlesub {\isaliteral{2E}{\isachardot}}x} 
  \]\smallskip
  
  \noindent
  indicating that a set (or list) of atoms \isa{as} is abstracted in \isa{x}. We will
  call the types \emph{abstraction types} and their elements
  \emph{abstractions}. The important property we need to derive is the support of 
  abstractions, namely:

  \begin{thm}[Support of Abstractions]\label{suppabs} 
  Assuming \isa{x} has finite support, then

  \[
  \begin{array}{l@ {\;=\;}l}
  \isa{supp\ {\isaliteral{5B}{\isacharbrackleft}}as{\isaliteral{5D}{\isacharbrackright}}\isaliteral{5C3C5E627375623E}{}\isactrlbsub set\isaliteral{5C3C5E657375623E}{}\isactrlesub {\isaliteral{2E}{\isachardot}}x} & \isa{supp\ x\ {\isaliteral{2D}{\isacharminus}}\ as}\\
  \isa{supp\ {\isaliteral{5B}{\isacharbrackleft}}as{\isaliteral{5D}{\isacharbrackright}}\isaliteral{5C3C5E627375623E}{}\isactrlbsub set{\isaliteral{2B}{\isacharplus}}\isaliteral{5C3C5E657375623E}{}\isactrlesub {\isaliteral{2E}{\isachardot}}x} & \isa{supp\ x\ {\isaliteral{2D}{\isacharminus}}\ as}\\
  \isa{supp\ {\isaliteral{5B}{\isacharbrackleft}}as{\isaliteral{5D}{\isacharbrackright}}\isaliteral{5C3C5E627375623E}{}\isactrlbsub list\isaliteral{5C3C5E657375623E}{}\isactrlesub {\isaliteral{2E}{\isachardot}}x} &
  \isa{supp\ x\ {\isaliteral{2D}{\isacharminus}}\ set\ as}\\
  \end{array}
  \]\smallskip
  \end{thm}

  \noindent
  In effect, this theorem states that the atoms \isa{as} are bound in the
  abstraction. As stated earlier, this can be seen as a litmus test that our
  Definitions \ref{alphaset}, \ref{alphalist} and \ref{alphares} capture the
  idea of alpha-equivalence relations. Below we will give the proof for the
  first equation of Theorem \ref{suppabs}. The others follow by similar
  arguments. By definition of the abstraction type \isa{abs\isaliteral{5C3C5E627375623E}{}\isactrlbsub set\isaliteral{5C3C5E657375623E}{}\isactrlesub } we have

  \begin{equation}\label{abseqiff}
  \isa{{\isaliteral{5B}{\isacharbrackleft}}as{\isaliteral{5D}{\isacharbrackright}}\isaliteral{5C3C5E627375623E}{}\isactrlbsub set\isaliteral{5C3C5E657375623E}{}\isactrlesub {\isaliteral{2E}{\isachardot}}x\ {\isaliteral{3D}{\isacharequal}}\ {\isaliteral{5B}{\isacharbrackleft}}bs{\isaliteral{5D}{\isacharbrackright}}\isaliteral{5C3C5E627375623E}{}\isactrlbsub set\isaliteral{5C3C5E657375623E}{}\isactrlesub {\isaliteral{2E}{\isachardot}}y} \;\;\;\text{if and only if}\;\;\; 
  \isa{{\isaliteral{28}{\isacharparenleft}}as{\isaliteral{2C}{\isacharcomma}}\ x{\isaliteral{29}{\isacharparenright}}\ {\isaliteral{5C3C617070726F783E}{\isasymapprox}}\,\raisebox{-1pt}{\makebox[0mm][l]{$_{\textit{set}}$}}\isaliteral{5C3C5E627375703E}{}\isactrlbsup {\isaliteral{3D}{\isacharequal}}{\isaliteral{2C}{\isacharcomma}}\ supp\isaliteral{5C3C5E657375703E}{}\isactrlesup \ {\isaliteral{28}{\isacharparenleft}}bs{\isaliteral{2C}{\isacharcomma}}\ y{\isaliteral{29}{\isacharparenright}}}
  \end{equation}\smallskip
  
  \noindent
  and also set
  
  \begin{equation}\label{absperm}
  \isa{{\isaliteral{5C3C70693E}{\isasympi}}\ {\isaliteral{5C3C62756C6C65743E}{\isasymbullet}}\ {\isaliteral{5B}{\isacharbrackleft}}as{\isaliteral{5D}{\isacharbrackright}}\isaliteral{5C3C5E627375623E}{}\isactrlbsub set\isaliteral{5C3C5E657375623E}{}\isactrlesub {\isaliteral{2E}{\isachardot}}x\ {\isaliteral{5C3C65717569763E}{\isasymequiv}}\ {\isaliteral{5B}{\isacharbrackleft}}{\isaliteral{5C3C70693E}{\isasympi}}\ {\isaliteral{5C3C62756C6C65743E}{\isasymbullet}}\ as{\isaliteral{5D}{\isacharbrackright}}\isaliteral{5C3C5E627375623E}{}\isactrlbsub set\isaliteral{5C3C5E657375623E}{}\isactrlesub {\isaliteral{2E}{\isachardot}}{\isaliteral{28}{\isacharparenleft}}{\isaliteral{5C3C70693E}{\isasympi}}\ {\isaliteral{5C3C62756C6C65743E}{\isasymbullet}}\ x{\isaliteral{29}{\isacharparenright}}}
  \end{equation}\smallskip

  \noindent
  With this at our disposal, we can show 
  the following lemma about swapping two atoms in an abstraction.
  
  \begin{lem}
  If \isa{a\ {\isaliteral{5C3C6E6F74696E3E}{\isasymnotin}}\ supp\ x\ {\isaliteral{2D}{\isacharminus}}\ as} and
  \isa{b\ {\isaliteral{5C3C6E6F74696E3E}{\isasymnotin}}\ supp\ x\ {\isaliteral{2D}{\isacharminus}}\ as} then 
  \isa{{\isaliteral{5B}{\isacharbrackleft}}as{\isaliteral{5D}{\isacharbrackright}}\isaliteral{5C3C5E627375623E}{}\isactrlbsub set\isaliteral{5C3C5E657375623E}{}\isactrlesub {\isaliteral{2E}{\isachardot}}x\ {\isaliteral{3D}{\isacharequal}}\ {\isaliteral{5B}{\isacharbrackleft}}{\isaliteral{28}{\isacharparenleft}}a\ b{\isaliteral{29}{\isacharparenright}}\ {\isaliteral{5C3C62756C6C65743E}{\isasymbullet}}\ as{\isaliteral{5D}{\isacharbrackright}}\isaliteral{5C3C5E627375623E}{}\isactrlbsub set\isaliteral{5C3C5E657375623E}{}\isactrlesub {\isaliteral{2E}{\isachardot}}{\isaliteral{28}{\isacharparenleft}}{\isaliteral{28}{\isacharparenleft}}a\ b{\isaliteral{29}{\isacharparenright}}\ {\isaliteral{5C3C62756C6C65743E}{\isasymbullet}}\ x{\isaliteral{29}{\isacharparenright}}}
  \end{lem}
  
  \begin{proof}
  If \isa{a\ {\isaliteral{3D}{\isacharequal}}\ b} the lemma is immediate, since \isa{{\isaliteral{28}{\isacharparenleft}}a\ b{\isaliteral{29}{\isacharparenright}}} is then
  the identity permutation.
  Also in the other case the lemma is straightforward using \eqref{abseqiff}
  and observing that the assumptions give us \isa{{\isaliteral{28}{\isacharparenleft}}a\ b{\isaliteral{29}{\isacharparenright}}\ {\isaliteral{5C3C62756C6C65743E}{\isasymbullet}}\ {\isaliteral{28}{\isacharparenleft}}supp\ x\ {\isaliteral{2D}{\isacharminus}}\ as{\isaliteral{29}{\isacharparenright}}\ {\isaliteral{3D}{\isacharequal}}\ supp\ x\ {\isaliteral{2D}{\isacharminus}}\ as}.  We therefore can use the swapping \isa{{\isaliteral{28}{\isacharparenleft}}a\ b{\isaliteral{29}{\isacharparenright}}} as
  the permutation for the proof obligation.
  \end{proof}
  
  \noindent
  This lemma together 
  with \eqref{absperm} allows us to show
  
  \begin{equation}\label{halfone}
  \isa{{\isaliteral{28}{\isacharparenleft}}supp\ x\ {\isaliteral{2D}{\isacharminus}}\ as{\isaliteral{29}{\isacharparenright}}\ supports\ {\isaliteral{5B}{\isacharbrackleft}}as{\isaliteral{5D}{\isacharbrackright}}\isaliteral{5C3C5E627375623E}{}\isactrlbsub set\isaliteral{5C3C5E657375623E}{}\isactrlesub {\isaliteral{2E}{\isachardot}}x}
  \end{equation}\smallskip
  
  \noindent
  which by Property~\ref{supportsprop} gives us `one half' of
  Theorem~\ref{suppabs}. To establish the `other half', we 
  use a trick from \cite{Pitts04} and first define an auxiliary 
  function \isa{aux}, taking an abstraction as argument

  \[
  \isa{aux\ {\isaliteral{28}{\isacharparenleft}}{\isaliteral{5B}{\isacharbrackleft}}as{\isaliteral{5D}{\isacharbrackright}}\isaliteral{5C3C5E627375623E}{}\isactrlbsub set\isaliteral{5C3C5E657375623E}{}\isactrlesub {\isaliteral{2E}{\isachardot}}x{\isaliteral{29}{\isacharparenright}}\ {\isaliteral{5C3C65717569763E}{\isasymequiv}}\ supp\ x\ {\isaliteral{2D}{\isacharminus}}\ as}
  \]\smallskip 

  \noindent
  Using the second equation in \eqref{equivariance}, we can show that 
  \isa{aux} is equivariant (since \isa{{\isaliteral{5C3C70693E}{\isasympi}}\ {\isaliteral{5C3C62756C6C65743E}{\isasymbullet}}\ {\isaliteral{28}{\isacharparenleft}}supp\ x\ {\isaliteral{2D}{\isacharminus}}\ as{\isaliteral{29}{\isacharparenright}}\ {\isaliteral{3D}{\isacharequal}}\ supp\ {\isaliteral{28}{\isacharparenleft}}{\isaliteral{5C3C70693E}{\isasympi}}\ {\isaliteral{5C3C62756C6C65743E}{\isasymbullet}}\ x{\isaliteral{29}{\isacharparenright}}\ {\isaliteral{2D}{\isacharminus}}\ {\isaliteral{5C3C70693E}{\isasympi}}\ {\isaliteral{5C3C62756C6C65743E}{\isasymbullet}}\ as}) 
  and therefore has empty support. 
  This in turn means
  
  \[
  \isa{supp\ {\isaliteral{28}{\isacharparenleft}}aux\ {\isaliteral{28}{\isacharparenleft}}{\isaliteral{5B}{\isacharbrackleft}}as{\isaliteral{5D}{\isacharbrackright}}\isaliteral{5C3C5E627375623E}{}\isactrlbsub set\isaliteral{5C3C5E657375623E}{}\isactrlesub {\isaliteral{2E}{\isachardot}}x{\isaliteral{29}{\isacharparenright}}{\isaliteral{29}{\isacharparenright}}\ {\isaliteral{5C3C73756273657465713E}{\isasymsubseteq}}\ supp\ {\isaliteral{5B}{\isacharbrackleft}}as{\isaliteral{5D}{\isacharbrackright}}\isaliteral{5C3C5E627375623E}{}\isactrlbsub set\isaliteral{5C3C5E657375623E}{}\isactrlesub {\isaliteral{2E}{\isachardot}}x}
  \]\smallskip
  
  \noindent
  using the fact about the support of function applications in \eqref{supps}. Assuming 
  \isa{supp\ x\ {\isaliteral{2D}{\isacharminus}}\ as} is a finite set, we further obtain
  
  \begin{equation}\label{halftwo}
  \isa{supp\ x\ {\isaliteral{2D}{\isacharminus}}\ as\ {\isaliteral{5C3C73756273657465713E}{\isasymsubseteq}}\ supp\ {\isaliteral{5B}{\isacharbrackleft}}as{\isaliteral{5D}{\isacharbrackright}}\isaliteral{5C3C5E627375623E}{}\isactrlbsub set\isaliteral{5C3C5E657375623E}{}\isactrlesub {\isaliteral{2E}{\isachardot}}x}
  \end{equation}\smallskip
  
  \noindent
  This is because for every finite set of atoms, say \isa{bs}, we have 
  \isa{supp\ bs\ {\isaliteral{3D}{\isacharequal}}\ bs}.\footnote{Note that this is not 
  the case for infinite sets.}
  Finally, taking \eqref{halfone} and \eqref{halftwo} together establishes 
  the first equation of Theorem~\ref{suppabs}. The others are similar.

  Recall the definition of support given in \eqref{suppdef}, and note the difference between 
  the support of a raw pair and an abstraction

  \[
  \isa{supp\ {\isaliteral{28}{\isacharparenleft}}as{\isaliteral{2C}{\isacharcomma}}\ x{\isaliteral{29}{\isacharparenright}}\ {\isaliteral{3D}{\isacharequal}}\ supp\ as\ {\isaliteral{5C3C756E696F6E3E}{\isasymunion}}\ supp\ x}\hspace{15mm}
  \isa{supp\ {\isaliteral{5B}{\isacharbrackleft}}as{\isaliteral{5D}{\isacharbrackright}}\isaliteral{5C3C5E627375623E}{}\isactrlbsub set\isaliteral{5C3C5E657375623E}{}\isactrlesub {\isaliteral{2E}{\isachardot}}x\ {\isaliteral{3D}{\isacharequal}}\ supp\ x\ {\isaliteral{2D}{\isacharminus}}\ as}
  \]\smallskip

  \noindent
  While the permutation operations behave in both cases the same (a permutation
  is just moved to the arguments), the notion of equality is different for pairs and
  abstractions. Therefore we have different supports. In case of abstractions,
  we have established in Theorem~\ref{suppabs} that bound atoms are removed from 
  the support of the abstractions' bodies.

  The method of first considering abstractions of the form \isa{{\isaliteral{5B}{\isacharbrackleft}}as{\isaliteral{5D}{\isacharbrackright}}\isaliteral{5C3C5E627375623E}{}\isactrlbsub set\isaliteral{5C3C5E657375623E}{}\isactrlesub {\isaliteral{2E}{\isachardot}}x} etc is motivated by the fact that we can conveniently establish at the
  Isabelle/HOL level properties about them.  It would be extremely laborious
  to write custom ML-code that derives automatically such properties for every
  term-constructor that binds some atoms. Also the generality of the
  definitions for alpha-equivalence will help us in the next sections.%
\end{isamarkuptext}%
\isamarkuptrue%
\isamarkupsection{Specifying General Bindings\label{sec:spec}%
}
\isamarkuptrue%
\begin{isamarkuptext}%
Our choice of syntax for specifications is influenced by the existing
  datatype package of Isabelle/HOL \cite{Berghofer99} 
  and by the syntax of the
  Ott-tool \cite{ott-jfp}. For us a specification of a term-calculus is a
  collection of (possibly mutually recursive) type declarations, say \isa{ty{\isaliteral{5C3C414C3E}{\isasymAL}}\isaliteral{5C3C5E697375623E}{}\isactrlisub {\isadigit{1}}{\isaliteral{2C}{\isacharcomma}}\ {\isaliteral{5C3C646F74733E}{\isasymdots}}{\isaliteral{2C}{\isacharcomma}}\ ty{\isaliteral{5C3C414C3E}{\isasymAL}}\isaliteral{5C3C5E697375623E}{}\isactrlisub n}, and an associated collection of
  binding functions, say \isa{bn{\isaliteral{5C3C414C3E}{\isasymAL}}\isaliteral{5C3C5E697375623E}{}\isactrlisub {\isadigit{1}}{\isaliteral{2C}{\isacharcomma}}\ {\isaliteral{5C3C646F74733E}{\isasymdots}}{\isaliteral{2C}{\isacharcomma}}\ bn{\isaliteral{5C3C414C3E}{\isasymAL}}\isaliteral{5C3C5E697375623E}{}\isactrlisub m}. The
  syntax in Nominal Isabelle for such specifications is schematically as follows:
  
  \begin{equation}\label{scheme}
  \mbox{\begin{tabular}{@ {}p{2.5cm}l}
  type \mbox{declaration part} &
  $\begin{cases}
  \mbox{\begin{tabular}{l}
  \isacommand{nominal\_datatype} \isa{ty{\isaliteral{5C3C414C3E}{\isasymAL}}\isaliteral{5C3C5E697375623E}{}\isactrlisub {\isadigit{1}}\ {\isaliteral{3D}{\isacharequal}}\ {\isaliteral{5C3C646F74733E}{\isasymdots}}}\\
  \isacommand{and} \isa{ty{\isaliteral{5C3C414C3E}{\isasymAL}}\isaliteral{5C3C5E697375623E}{}\isactrlisub {\isadigit{2}}\ {\isaliteral{3D}{\isacharequal}}\ {\isaliteral{5C3C646F74733E}{\isasymdots}}}\\
  \raisebox{2mm}{$\ldots$}\\[-2mm] 
  \isacommand{and} \isa{ty{\isaliteral{5C3C414C3E}{\isasymAL}}\isaliteral{5C3C5E697375623E}{}\isactrlisub n\ {\isaliteral{3D}{\isacharequal}}\ {\isaliteral{5C3C646F74733E}{\isasymdots}}}\\ 
  \end{tabular}}
  \end{cases}$\\[2mm]
  binding \mbox{function part} &
  $\begin{cases}
  \mbox{\begin{tabular}{l}
  \isacommand{binder} \isa{bn{\isaliteral{5C3C414C3E}{\isasymAL}}\isaliteral{5C3C5E697375623E}{}\isactrlisub {\isadigit{1}}} \isacommand{and} \ldots \isacommand{and} \isa{bn{\isaliteral{5C3C414C3E}{\isasymAL}}\isaliteral{5C3C5E697375623E}{}\isactrlisub m}\\
  \isacommand{where}\\
  \raisebox{2mm}{$\ldots$}\\[-2mm]
  \end{tabular}}
  \end{cases}$\\
  \end{tabular}}
  \end{equation}\smallskip

  \noindent
  Every type declaration \isa{ty}$^\alpha_{1..n}$ consists of a collection
  of term-constructors, each of which comes with a list of labelled types that
  stand for the types of the arguments of the term-constructor.  For example a
  term-constructor \isa{C\isaliteral{5C3C5E7375703E}{}\isactrlsup {\isaliteral{5C3C616C7068613E}{\isasymalpha}}} might be specified with

  \[
  \isa{C\isaliteral{5C3C5E7375703E}{}\isactrlsup {\isaliteral{5C3C616C7068613E}{\isasymalpha}}\ label\isaliteral{5C3C5E697375623E}{}\isactrlisub {\isadigit{1}}{\isaliteral{3A}{\isacharcolon}}{\isaliteral{3A}{\isacharcolon}}ty}\mbox{$'_1$} \isa{{\isaliteral{5C3C646F74733E}{\isasymdots}}\ label\isaliteral{5C3C5E697375623E}{}\isactrlisub l{\isaliteral{3A}{\isacharcolon}}{\isaliteral{3A}{\isacharcolon}}ty}\mbox{$'_l\;\;\;\;\;$}  
  \isa{binding{\isaliteral{5F}{\isacharunderscore}}clauses} 
  \]\smallskip
  
  \noindent
  whereby some of the \isa{ty}$'_{1..l}$ (or their components) can be
  contained in the collection of \isa{ty}$^\alpha_{1..n}$ declared in
  \eqref{scheme}. In this case we will call the corresponding argument a
  \emph{recursive argument} of \isa{C\isaliteral{5C3C5E7375703E}{}\isactrlsup {\isaliteral{5C3C616C7068613E}{\isasymalpha}}}. The types of such
  recursive arguments need to satisfy a `positivity' restriction, which
  ensures that the type has a set-theoretic semantics (see
  \cite{Berghofer99}). If the types are polymorphic, we require the
  type variables to stand for types that are finitely supported and over which 
  a permutation operation is defined.
  The labels \isa{label}$_{1..l}$ annotated on the types are optional. Their
  purpose is to be used in the (possibly empty) list of \emph{binding
  clauses}, which indicate the binders and their scope in a term-constructor.
  They come in three \emph{modes}:

  \[\mbox{
  \begin{tabular}{@ {}l@ {}}
  \isacommand{binds} {\it binders} \isacommand{in} {\it bodies}\\
  \isacommand{binds (set)} {\it binders} \isacommand{in} {\it bodies}\\
  \isacommand{binds (set+)} {\it binders} \isacommand{in} {\it bodies}
  \end{tabular}}
  \]\smallskip
  
  \noindent
  The first mode is for binding lists of atoms (the order of bound atoms
  matters); the second is for sets of binders (the order does not matter, but
  the cardinality does) and the last is for sets of binders (with vacuous
  binders preserving alpha-equivalence). As indicated, the labels in the
  `\isacommand{in}-part' of a binding clause will be called \emph{bodies};
  the `\isacommand{binds}-part' will be called \emph{binders}. In contrast to
  Ott, we allow multiple labels in binders and bodies.  For example we allow
  binding clauses of the form:
 
  \[\mbox{
  \begin{tabular}{@ {}ll@ {}}
  \isa{Foo\isaliteral{5C3C5E697375623E}{}\isactrlisub {\isadigit{1}}\ x{\isaliteral{3A}{\isacharcolon}}{\isaliteral{3A}{\isacharcolon}}name\ y{\isaliteral{3A}{\isacharcolon}}{\isaliteral{3A}{\isacharcolon}}name\ t{\isaliteral{3A}{\isacharcolon}}{\isaliteral{3A}{\isacharcolon}}trm\ s{\isaliteral{3A}{\isacharcolon}}{\isaliteral{3A}{\isacharcolon}}trm} &  
      \isacommand{binds} \isa{x\ y} \isacommand{in} \isa{t\ s}\\
  \isa{Foo\isaliteral{5C3C5E697375623E}{}\isactrlisub {\isadigit{2}}\ x{\isaliteral{3A}{\isacharcolon}}{\isaliteral{3A}{\isacharcolon}}name\ y{\isaliteral{3A}{\isacharcolon}}{\isaliteral{3A}{\isacharcolon}}name\ t{\isaliteral{3A}{\isacharcolon}}{\isaliteral{3A}{\isacharcolon}}trm\ s{\isaliteral{3A}{\isacharcolon}}{\isaliteral{3A}{\isacharcolon}}trm} &  
      \isacommand{binds} \isa{x\ y} \isacommand{in} \isa{t}, 
      \isacommand{binds} \isa{x\ y} \isacommand{in} \isa{s}\\
  \end{tabular}}
  \]\smallskip

  \noindent
  Similarly for the other binding modes. Interestingly, in case of
  \isacommand{binds (set)} and \isacommand{binds (set+)} the binding clauses
  above will make a difference to the semantics of the specifications (the
  corresponding alpha-equivalence will differ). We will show this later with
  an example.

  There are also some restrictions we need to impose on our binding clauses in
  comparison to Ott. The main idea behind these restrictions is
  that we obtain a notion of alpha-equivalence where it is ensured
  that within a given scope an atom occurrence cannot be both bound and free
  at the same time.  The first restriction is that a body can only occur in
  \emph{one} binding clause of a term constructor. So for example

  \[\mbox{
  \isa{Foo\ x{\isaliteral{3A}{\isacharcolon}}{\isaliteral{3A}{\isacharcolon}}name\ y{\isaliteral{3A}{\isacharcolon}}{\isaliteral{3A}{\isacharcolon}}name\ t{\isaliteral{3A}{\isacharcolon}}{\isaliteral{3A}{\isacharcolon}}trm}\hspace{3mm}  
  \isacommand{binds} \isa{x} \isacommand{in} \isa{t},
  \isacommand{binds} \isa{y} \isacommand{in} \isa{t}}
  \]\smallskip

  \noindent
  is not allowed. This ensures that the bound atoms of a body cannot be free
  at the same time by specifying an alternative binder for the same body.

  For binders we distinguish between \emph{shallow} and \emph{deep} binders.
  Shallow binders are just labels. The restriction we need to impose on them
  is that in case of \isacommand{binds (set)} and \isacommand{binds (set+)} the
  labels must either refer to atom types or to sets of atom types; in case of
  \isacommand{binds} the labels must refer to atom types or to lists of atom
  types. Two examples for the use of shallow binders are the specification of
  lambda-terms, where a single name is bound, and type-schemes, where a finite
  set of names is bound:

  \[\mbox{
  \begin{tabular}{@ {}c@ {\hspace{8mm}}c@ {}}
  \begin{tabular}{@ {}l}
  \isacommand{nominal\_datatype} \isa{lam} $=$\\
  \hspace{2mm}\phantom{$\mid$}~\isa{Var\ name}\\
  \hspace{2mm}$\mid$~\isa{App\ lam\ lam}\\
  \hspace{2mm}$\mid$~\isa{Lam\ x{\isaliteral{3A}{\isacharcolon}}{\isaliteral{3A}{\isacharcolon}}name\ t{\isaliteral{3A}{\isacharcolon}}{\isaliteral{3A}{\isacharcolon}}lam}\hspace{3mm}%
  \isacommand{binds} \isa{x} \isacommand{in} \isa{t}\\
  \\
  \end{tabular} &
  \begin{tabular}{@ {}l@ {}}
  \isacommand{nominal\_datatype}~\isa{ty} $=$\\
  \hspace{2mm}\phantom{$\mid$}~\isa{TVar\ name}\\
  \hspace{2mm}$\mid$~\isa{TFun\ ty\ ty}\\
  \isacommand{and}~\isa{tsc\ {\isaliteral{3D}{\isacharequal}}}\\
  \hspace{2mm}\phantom{$\mid$}~\isa{TAll\ xs{\isaliteral{3A}{\isacharcolon}}{\isaliteral{3A}{\isacharcolon}}{\isaliteral{28}{\isacharparenleft}}name\ fset{\isaliteral{29}{\isacharparenright}}\ T{\isaliteral{3A}{\isacharcolon}}{\isaliteral{3A}{\isacharcolon}}ty}\hspace{3mm}%
  \isacommand{binds (set+)} \isa{xs} \isacommand{in} \isa{T}\\
  \end{tabular}
  \end{tabular}}
  \]\smallskip

  \noindent
  In these specifications \isa{name} refers to a (concrete) atom type, and \isa{fset} to the type of finite sets.  Note that for \isa{Lam} it does not
  matter which binding mode we use. The reason is that we bind only a single
  \isa{name}, in which case all three binding modes coincide. However, having 
  \isacommand{binds (set)} or just \isacommand{binds}
  in the second case makes a difference to the semantics of the specification
  (which we will define in the next section).

  A \emph{deep} binder uses an auxiliary binding function that `picks' out
  the atoms in one argument of the term-constructor, which can be bound in
  other arguments and also in the same argument (we will call such binders
  \emph{recursive}, see below). The binding functions are
  expected to return either a set of atoms (for \isacommand{binds (set)} and
  \isacommand{binds (set+)}) or a list of atoms (for \isacommand{binds}). They need
  to be defined by recursion over the corresponding type; the equations
  must be given in the binding function part of the scheme shown in
  \eqref{scheme}. For example a term-calculus containing \isa{Let}s with
  tuple patterns may be specified as:

  \begin{equation}\label{letpat}
  \mbox{%
  \begin{tabular}{l}
  \isacommand{nominal\_datatype} \isa{trm} $=$\\
  \hspace{5mm}\phantom{$\mid$}~\isa{Var\ name}\\
  \hspace{5mm}$\mid$~\isa{App\ trm\ trm}\\
  \hspace{5mm}$\mid$~\isa{Lam\ x{\isaliteral{3A}{\isacharcolon}}{\isaliteral{3A}{\isacharcolon}}name\ t{\isaliteral{3A}{\isacharcolon}}{\isaliteral{3A}{\isacharcolon}}trm} 
     \;\;\isacommand{binds} \isa{x} \isacommand{in} \isa{t}\\
  \hspace{5mm}$\mid$~\isa{Let{\isaliteral{5F}{\isacharunderscore}}pat\ p{\isaliteral{3A}{\isacharcolon}}{\isaliteral{3A}{\isacharcolon}}pat\ trm\ t{\isaliteral{3A}{\isacharcolon}}{\isaliteral{3A}{\isacharcolon}}trm} 
     \;\;\isacommand{binds} \isa{bn{\isaliteral{28}{\isacharparenleft}}p{\isaliteral{29}{\isacharparenright}}} \isacommand{in} \isa{t}\\
  \isacommand{and} \isa{pat} $=$\\
  \hspace{5mm}\phantom{$\mid$}~\isa{PVar\ name}\\
  \hspace{5mm}$\mid$~\isa{PTup\ pat\ pat}\\ 
  \isacommand{binder}~\isa{bn{\isaliteral{3A}{\isacharcolon}}{\isaliteral{3A}{\isacharcolon}}pat\ {\isaliteral{5C3C52696768746172726F773E}{\isasymRightarrow}}\ atom\ list}\\
  \isacommand{where}~\isa{bn{\isaliteral{28}{\isacharparenleft}}PVar\ x{\isaliteral{29}{\isacharparenright}}\ {\isaliteral{3D}{\isacharequal}}\ {\isaliteral{5B}{\isacharbrackleft}}atom\ x{\isaliteral{5D}{\isacharbrackright}}}\\
  \hspace{5mm}$\mid$~\isa{bn{\isaliteral{28}{\isacharparenleft}}PTup\ p\isaliteral{5C3C5E697375623E}{}\isactrlisub {\isadigit{1}}\ p\isaliteral{5C3C5E697375623E}{}\isactrlisub {\isadigit{2}}{\isaliteral{29}{\isacharparenright}}\ {\isaliteral{3D}{\isacharequal}}\ bn{\isaliteral{28}{\isacharparenleft}}p\isaliteral{5C3C5E697375623E}{}\isactrlisub {\isadigit{1}}{\isaliteral{29}{\isacharparenright}}\ {\isaliteral{40}{\isacharat}}\ bn{\isaliteral{28}{\isacharparenleft}}p\isaliteral{5C3C5E697375623E}{}\isactrlisub {\isadigit{2}}{\isaliteral{29}{\isacharparenright}}}\smallskip\\ 
  \end{tabular}}
  \end{equation}\smallskip

  \noindent
  In this specification the function \isa{bn} determines which atoms of
  the pattern \isa{p} (fifth line) are bound in the argument \isa{t}. Note that in the
  second-last \isa{bn}-clause the function \isa{atom} coerces a name
  into the generic atom type of Nominal Isabelle \cite{HuffmanUrban10}. This
  allows us to treat binders of different atom type uniformly.

  For deep binders we allow binding clauses such as
  
  \[\mbox{
  \begin{tabular}{ll}
  \isa{Bar\ p{\isaliteral{3A}{\isacharcolon}}{\isaliteral{3A}{\isacharcolon}}pat\ t{\isaliteral{3A}{\isacharcolon}}{\isaliteral{3A}{\isacharcolon}}trm} &  
     \isacommand{binds} \isa{bn{\isaliteral{28}{\isacharparenleft}}p{\isaliteral{29}{\isacharparenright}}} \isacommand{in} \isa{p\ t} \\
  \end{tabular}}
  \]\smallskip

  \noindent
  where the argument of the deep binder also occurs in the body. We call such
  binders \emph{recursive}.  To see the purpose of such recursive binders,
  compare `plain' \isa{Let}s and \isa{Let{\isaliteral{5F}{\isacharunderscore}}rec}s in the following
  specification:
 
  \begin{equation}\label{letrecs}
  \mbox{%
  \begin{tabular}{@ {}l@ {}l}
  \isacommand{nominal\_datatype}~\isa{trm\ {\isaliteral{3D}{\isacharequal}}}\\
  \hspace{5mm}\phantom{$\mid$}~\ldots\\
  \hspace{5mm}$\mid$~\isa{Let\ as{\isaliteral{3A}{\isacharcolon}}{\isaliteral{3A}{\isacharcolon}}assn\ t{\isaliteral{3A}{\isacharcolon}}{\isaliteral{3A}{\isacharcolon}}trm} 
     & \hspace{-19mm}\isacommand{binds} \isa{bn{\isaliteral{28}{\isacharparenleft}}as{\isaliteral{29}{\isacharparenright}}} \isacommand{in} \isa{t}\\
  \hspace{5mm}$\mid$~\isa{Let{\isaliteral{5F}{\isacharunderscore}}rec\ as{\isaliteral{3A}{\isacharcolon}}{\isaliteral{3A}{\isacharcolon}}assn\ t{\isaliteral{3A}{\isacharcolon}}{\isaliteral{3A}{\isacharcolon}}trm}
     & \hspace{-19mm}\isacommand{binds} \isa{bn{\isaliteral{28}{\isacharparenleft}}as{\isaliteral{29}{\isacharparenright}}} \isacommand{in} \isa{as\ t}\\
  \isacommand{and} \isa{assn} $=$\\
  \hspace{5mm}\phantom{$\mid$}~\isa{ANil}\\
  \hspace{5mm}$\mid$~\isa{ACons\ name\ trm\ assn}\\
  \isacommand{binder} \isa{bn{\isaliteral{3A}{\isacharcolon}}{\isaliteral{3A}{\isacharcolon}}assn\ {\isaliteral{5C3C52696768746172726F773E}{\isasymRightarrow}}\ atom\ list}\\
  \isacommand{where}~\isa{bn{\isaliteral{28}{\isacharparenleft}}ANil{\isaliteral{29}{\isacharparenright}}\ {\isaliteral{3D}{\isacharequal}}\ {\isaliteral{5B}{\isacharbrackleft}}{\isaliteral{5D}{\isacharbrackright}}}\\
  \hspace{5mm}$\mid$~\isa{bn{\isaliteral{28}{\isacharparenleft}}ACons\ a\ t\ as{\isaliteral{29}{\isacharparenright}}\ {\isaliteral{3D}{\isacharequal}}\ {\isaliteral{5B}{\isacharbrackleft}}atom\ a{\isaliteral{5D}{\isacharbrackright}}\ {\isaliteral{40}{\isacharat}}\ bn{\isaliteral{28}{\isacharparenleft}}as{\isaliteral{29}{\isacharparenright}}}\\
  \end{tabular}}
  \end{equation}\smallskip
  
  \noindent
  The difference is that with \isa{Let} we only want to bind the atoms \isa{bn{\isaliteral{28}{\isacharparenleft}}as{\isaliteral{29}{\isacharparenright}}} in the term \isa{t}, but with \isa{Let{\isaliteral{5F}{\isacharunderscore}}rec} we also want to bind the atoms
  inside the assignment. This difference has consequences for the associated
  notions of free-atoms and alpha-equivalence.
  
  To make sure that atoms bound by deep binders cannot be free at the
  same time, we cannot have more than one binding function for a deep binder. 
  Consequently we exclude specifications such as

  \[\mbox{
  \begin{tabular}{@ {}l@ {\hspace{2mm}}l@ {}}
  \isa{Baz\isaliteral{5C3C5E697375623E}{}\isactrlisub {\isadigit{1}}\ p{\isaliteral{3A}{\isacharcolon}}{\isaliteral{3A}{\isacharcolon}}pat\ t{\isaliteral{3A}{\isacharcolon}}{\isaliteral{3A}{\isacharcolon}}trm} & 
     \isacommand{binds} \isa{bn\isaliteral{5C3C5E697375623E}{}\isactrlisub {\isadigit{1}}{\isaliteral{28}{\isacharparenleft}}p{\isaliteral{29}{\isacharparenright}}\ bn\isaliteral{5C3C5E697375623E}{}\isactrlisub {\isadigit{2}}{\isaliteral{28}{\isacharparenleft}}p{\isaliteral{29}{\isacharparenright}}} \isacommand{in} \isa{p\ t}\\
  \isa{Baz\isaliteral{5C3C5E697375623E}{}\isactrlisub {\isadigit{2}}\ p{\isaliteral{3A}{\isacharcolon}}{\isaliteral{3A}{\isacharcolon}}pat\ t\isaliteral{5C3C5E697375623E}{}\isactrlisub {\isadigit{1}}{\isaliteral{3A}{\isacharcolon}}{\isaliteral{3A}{\isacharcolon}}trm\ t\isaliteral{5C3C5E697375623E}{}\isactrlisub {\isadigit{2}}{\isaliteral{3A}{\isacharcolon}}{\isaliteral{3A}{\isacharcolon}}trm} & 
     \isacommand{binds} \isa{bn\isaliteral{5C3C5E697375623E}{}\isactrlisub {\isadigit{1}}{\isaliteral{28}{\isacharparenleft}}p{\isaliteral{29}{\isacharparenright}}} \isacommand{in} \isa{p\ t\isaliteral{5C3C5E697375623E}{}\isactrlisub {\isadigit{1}}},
     \isacommand{binds} \isa{bn\isaliteral{5C3C5E697375623E}{}\isactrlisub {\isadigit{2}}{\isaliteral{28}{\isacharparenleft}}p{\isaliteral{29}{\isacharparenright}}} \isacommand{in} \isa{p\ t\isaliteral{5C3C5E697375623E}{}\isactrlisub {\isadigit{2}}}\\
  \end{tabular}}
  \]\smallskip

  \noindent
  Otherwise it is possible that \isa{bn\isaliteral{5C3C5E697375623E}{}\isactrlisub {\isadigit{1}}} and \isa{bn\isaliteral{5C3C5E697375623E}{}\isactrlisub {\isadigit{2}}}  pick 
  out different atoms to become bound, respectively be free, 
  in \isa{p}.\footnote{Since the Ott-tool does not derive a reasoning 
  infrastructure for 
  alpha-equated terms with deep binders, it can permit such specifications.}

  We also need to restrict the form of the binding functions in order to
  ensure the \isa{bn}-functions can be defined for alpha-equated
  terms. The main restriction is that we cannot return an atom in a binding
  function that is also bound in the corresponding term-constructor.
  Consider again the specification for \isa{trm} and a contrived
  version for assignments \isa{assn}:

  \begin{equation}\label{bnexp}
  \mbox{%
  \begin{tabular}{@ {}l@ {}}
  \isacommand{nominal\_datatype}~\isa{trm\ {\isaliteral{3D}{\isacharequal}}}~\ldots\\
  \isacommand{and} \isa{assn} $=$\\
  \hspace{5mm}\phantom{$\mid$}~\isa{ANil{\isaliteral{27}{\isacharprime}}}\\
  \hspace{5mm}$\mid$~\isa{ACons{\isaliteral{27}{\isacharprime}}\ x{\isaliteral{3A}{\isacharcolon}}{\isaliteral{3A}{\isacharcolon}}name\ y{\isaliteral{3A}{\isacharcolon}}{\isaliteral{3A}{\isacharcolon}}name\ t{\isaliteral{3A}{\isacharcolon}}{\isaliteral{3A}{\isacharcolon}}trm\ assn}
     \;\;\isacommand{binds} \isa{y} \isacommand{in} \isa{t}\\
  \isacommand{binder} \isa{bn{\isaliteral{3A}{\isacharcolon}}{\isaliteral{3A}{\isacharcolon}}assn\ {\isaliteral{5C3C52696768746172726F773E}{\isasymRightarrow}}\ atom\ list}\\
  \isacommand{where}~\isa{bn{\isaliteral{28}{\isacharparenleft}}ANil{\isaliteral{27}{\isacharprime}}{\isaliteral{29}{\isacharparenright}}\ {\isaliteral{3D}{\isacharequal}}\ {\isaliteral{5B}{\isacharbrackleft}}{\isaliteral{5D}{\isacharbrackright}}}\\
  \hspace{5mm}$\mid$~\isa{bn{\isaliteral{28}{\isacharparenleft}}ACons{\isaliteral{27}{\isacharprime}}\ x\ y\ t\ as{\isaliteral{29}{\isacharparenright}}\ {\isaliteral{3D}{\isacharequal}}\ {\isaliteral{5B}{\isacharbrackleft}}atom\ x{\isaliteral{5D}{\isacharbrackright}}\ {\isaliteral{40}{\isacharat}}\ bn{\isaliteral{28}{\isacharparenleft}}as{\isaliteral{29}{\isacharparenright}}}\\
  \end{tabular}}
  \end{equation}\smallskip

  \noindent
  In this example the term constructor \isa{ACons{\isaliteral{27}{\isacharprime}}} has four arguments with
  a binding clause involving two of them. This constructor is also used in the definition
  of the binding function. The restriction we have to impose is that the
  binding function can only return free atoms, that is the ones that are \emph{not}
  mentioned in a binding clause.  Therefore \isa{y} cannot be used in the
  binding function \isa{bn} (since it is bound in \isa{ACons{\isaliteral{27}{\isacharprime}}} by the
  binding clause), but \isa{x} can (since it is a free atom). This
  restriction is sufficient for lifting the binding function to alpha-equated
  terms. If we would permit \isa{bn} to return \isa{y},
  then it would not be respectful and therefore cannot be lifted to
  alpha-equated lambda-terms.

  In the version of Nominal Isabelle described here, we also adopted the
  restriction from the Ott-tool that binding functions can only return: the
  empty set or empty list (as in case \isa{ANil{\isaliteral{27}{\isacharprime}}}), a singleton set or
  singleton list containing an atom (case \isa{PVar} in \eqref{letpat}), or
  unions of atom sets or appended atom lists (case \isa{ACons{\isaliteral{27}{\isacharprime}}}). This
  restriction will simplify some automatic definitions and proofs later on.
  
  To sum up this section, we introduced nominal datatype
  specifications, which are like standard datatype specifications in
  Isabelle/HOL but extended with binding clauses and specifications for binding
  functions. Each constructor argument in our specification can also
  have an optional label. These labels are used in the binding clauses
  of a constructor; there can be several binding clauses for each
  constructor, but bodies of binding clauses can only occur in a
  single one. Binding clauses come in three modes: \isacommand{binds},
  \isacommand{binds (set)} and \isacommand{binds (set+)}.  Binders
  fall into two categories: shallow binders and deep binders. Shallow
  binders can occur in more than one binding clause and only have to
  respect the binding mode (i.e.~be of the right type). Deep binders
  can also occur in more than one binding clause, unless they are
  recursive in which case they can only occur once. Each of the deep
  binders can only have a single binding function.  Binding functions
  are defined by recursion over a nominal datatype.  They can
  return the empty set, singleton atoms and unions of sets of atoms
  (for binding modes \isacommand{binds (set)} and \isacommand{binds
  (set+)}), and the empty list, singleton atoms and appended lists of
  atoms (for mode \isacommand{bind}). However, they can only return
  atoms that are not mentioned in any binding clause.  

  In order to
  simplify our definitions of free atoms and alpha-equivalence we define next, we
  shall assume specifications of term-calculi are implicitly
  \emph{completed}. By this we mean that for every argument of a
  term-constructor that is \emph{not} already part of a binding clause
  given by the user, we add implicitly a special \emph{empty} binding
  clause, written \isacommand{binds}~\isa{{\isaliteral{5C3C656D7074797365743E}{\isasymemptyset}}}~\isacommand{in}~\isa{labels}. In case of the lambda-terms,
  the completion produces

  \[\mbox{
  \begin{tabular}{@ {}l@ {\hspace{-1mm}}}
  \isacommand{nominal\_datatype} \isa{lam} =\\
  \hspace{5mm}\phantom{$\mid$}~\isa{Var\ x{\isaliteral{3A}{\isacharcolon}}{\isaliteral{3A}{\isacharcolon}}name}
    \;\;\isacommand{binds}~\isa{{\isaliteral{5C3C656D7074797365743E}{\isasymemptyset}}}~\isacommand{in}~\isa{x}\\
  \hspace{5mm}$\mid$~\isa{App\ t\isaliteral{5C3C5E697375623E}{}\isactrlisub {\isadigit{1}}{\isaliteral{3A}{\isacharcolon}}{\isaliteral{3A}{\isacharcolon}}lam\ t\isaliteral{5C3C5E697375623E}{}\isactrlisub {\isadigit{2}}{\isaliteral{3A}{\isacharcolon}}{\isaliteral{3A}{\isacharcolon}}lam}
    \;\;\isacommand{binds}~\isa{{\isaliteral{5C3C656D7074797365743E}{\isasymemptyset}}}~\isacommand{in}~\isa{t\isaliteral{5C3C5E697375623E}{}\isactrlisub {\isadigit{1}}\ t\isaliteral{5C3C5E697375623E}{}\isactrlisub {\isadigit{2}}}\\
  \hspace{5mm}$\mid$~\isa{Lam\ x{\isaliteral{3A}{\isacharcolon}}{\isaliteral{3A}{\isacharcolon}}name\ t{\isaliteral{3A}{\isacharcolon}}{\isaliteral{3A}{\isacharcolon}}lam}
    \;\;\isacommand{binds}~\isa{x} \isacommand{in} \isa{t}\\
  \end{tabular}}
  \]\smallskip

  \noindent 
  The point of completion is that we can make definitions over the binding
  clauses and be sure to have captured all arguments of a term constructor.%
\end{isamarkuptext}%
\isamarkuptrue%
\isamarkupsection{Alpha-Equivalence and Free Atoms\label{sec:alpha}%
}
\isamarkuptrue%
\begin{isamarkuptext}%
Having dealt with all syntax matters, the problem now is how we can turn
  specifications into actual type definitions in Isabelle/HOL and then
  establish a reasoning infrastructure for them. As Pottier and Cheney pointed
  out \cite{Cheney05,Pottier06}, just re-arranging the arguments of
  term-constructors so that binders and their bodies are next to each other
  will result in inadequate representations in cases like \mbox{\isa{Let\ x\isaliteral{5C3C5E697375623E}{}\isactrlisub {\isadigit{1}}\ {\isaliteral{3D}{\isacharequal}}\ t\isaliteral{5C3C5E697375623E}{}\isactrlisub {\isadigit{1}}{\isaliteral{5C3C646F74733E}{\isasymdots}}x\isaliteral{5C3C5E697375623E}{}\isactrlisub n\ {\isaliteral{3D}{\isacharequal}}\ t\isaliteral{5C3C5E697375623E}{}\isactrlisub n\ in\ s}}. Therefore we will
  first extract `raw' datatype definitions from the specification and then
  define explicitly an alpha-equivalence relation over them. We subsequently
  construct the quotient of the datatypes according to our alpha-equivalence.

  The `raw' datatype definition can be obtained by stripping off the 
  binding clauses and the labels from the types given by the user. We also have to invent
  new names for the  types \isa{ty\isaliteral{5C3C5E7375703E}{}\isactrlsup {\isaliteral{5C3C616C7068613E}{\isasymalpha}}} and the term-constructors \isa{C\isaliteral{5C3C5E7375703E}{}\isactrlsup {\isaliteral{5C3C616C7068613E}{\isasymalpha}}}. 
  In our implementation we just use the affix ``\isa{{\isaliteral{5F}{\isacharunderscore}}raw}''.
  But for the purpose of this paper, we use the superscript \isa{{\isaliteral{5F}{\isacharunderscore}}\isaliteral{5C3C5E7375703E}{}\isactrlsup {\isaliteral{5C3C616C7068613E}{\isasymalpha}}} to indicate 
  that a notion is given for alpha-equivalence classes and leave it out 
  for the corresponding notion given on the raw level. So for example 
  we have \isa{ty\isaliteral{5C3C5E7375703E}{}\isactrlsup {\isaliteral{5C3C616C7068613E}{\isasymalpha}}\ {\isaliteral{2F}{\isacharslash}}\ ty} and \isa{C\isaliteral{5C3C5E7375703E}{}\isactrlsup {\isaliteral{5C3C616C7068613E}{\isasymalpha}}\ {\isaliteral{2F}{\isacharslash}}\ C}
  where \isa{ty} is the type used in the quotient construction for 
  \isa{ty\isaliteral{5C3C5E7375703E}{}\isactrlsup {\isaliteral{5C3C616C7068613E}{\isasymalpha}}} and \isa{C} is the term-constructor of the raw type \isa{ty},
  respectively \isa{C\isaliteral{5C3C5E7375703E}{}\isactrlsup {\isaliteral{5C3C616C7068613E}{\isasymalpha}}} is the corresponding term-constructor of \isa{ty\isaliteral{5C3C5E7375703E}{}\isactrlsup {\isaliteral{5C3C616C7068613E}{\isasymalpha}}}. 

  The resulting datatype definition is legal in Isabelle/HOL provided the datatypes are 
  non-empty and the types in the constructors only occur in positive 
  position (see \cite{Berghofer99} for an in-depth description of the datatype package
  in Isabelle/HOL). 
  We subsequently define each of the user-specified binding 
  functions \isa{bn}$_{1..m}$ by recursion over the corresponding 
  raw datatype. We also define permutation operations by 
  recursion so that for each term constructor \isa{C} we have that
  
  \begin{equation}\label{ceqvt}
  \isa{{\isaliteral{5C3C70693E}{\isasympi}}\ {\isaliteral{5C3C62756C6C65743E}{\isasymbullet}}\ {\isaliteral{28}{\isacharparenleft}}C\ z\isaliteral{5C3C5E697375623E}{}\isactrlisub {\isadigit{1}}\ {\isaliteral{5C3C646F74733E}{\isasymdots}}\ z\isaliteral{5C3C5E697375623E}{}\isactrlisub n{\isaliteral{29}{\isacharparenright}}\ {\isaliteral{3D}{\isacharequal}}\ C\ {\isaliteral{28}{\isacharparenleft}}{\isaliteral{5C3C70693E}{\isasympi}}\ {\isaliteral{5C3C62756C6C65743E}{\isasymbullet}}\ z\isaliteral{5C3C5E697375623E}{}\isactrlisub {\isadigit{1}}{\isaliteral{29}{\isacharparenright}}\ {\isaliteral{5C3C646F74733E}{\isasymdots}}\ {\isaliteral{28}{\isacharparenleft}}{\isaliteral{5C3C70693E}{\isasympi}}\ {\isaliteral{5C3C62756C6C65743E}{\isasymbullet}}\ z\isaliteral{5C3C5E697375623E}{}\isactrlisub n{\isaliteral{29}{\isacharparenright}}}
  \end{equation}\smallskip

  \noindent
  We will need this operation later when we define the notion of alpha-equivalence.

  The first non-trivial step we have to perform is the generation of
  \emph{free-atom functions} from the specifications.\footnote{Admittedly, the
  details of our definitions will be somewhat involved. However they are still
  conceptually simple in comparison with the `positional' approach taken in
  Ott \cite[Pages 88--95]{ott-jfp}, which uses the notions of \emph{occurrences} and
  \emph{partial equivalence relations} over sets of occurrences.} For the
  \emph{raw} types \isa{ty}$_{1..n}$ we define the free-atom functions

  \begin{equation}\label{fvars}
  \mbox{\isa{fa{\isaliteral{5F}{\isacharunderscore}}ty}$_{1..n}$}
  \end{equation}\smallskip
  
  \noindent
  by recursion.
  We define these functions together with auxiliary free-atom functions for
  the binding functions. Given raw binding functions \isa{bn}$_{1..m}$ 
  we define
  
  \[
  \isa{fa{\isaliteral{5F}{\isacharunderscore}}bn}\mbox{$_{1..m}$}.
  \]\smallskip
  
  \noindent
  The reason for this setup is that in a deep binder not all atoms have to be
  bound, as we saw in \eqref{letrecs} with the example of `plain' \isa{Let}s. We need
  therefore functions that calculate those free atoms in deep binders.

  While the idea behind these free-atom functions is simple (they just
  collect all atoms that are not bound), because of our rather complicated
  binding mechanisms their definitions are somewhat involved.  Given
  a raw term-constructor \isa{C} of type \isa{ty} and some associated
  binding clauses \isa{bc\isaliteral{5C3C5E697375623E}{}\isactrlisub {\isadigit{1}}{\isaliteral{5C3C646F74733E}{\isasymdots}}bc\isaliteral{5C3C5E697375623E}{}\isactrlisub k}, the result of \isa{fa{\isaliteral{5F}{\isacharunderscore}}ty\ {\isaliteral{28}{\isacharparenleft}}C\ z\isaliteral{5C3C5E697375623E}{}\isactrlisub {\isadigit{1}}\ {\isaliteral{5C3C646F74733E}{\isasymdots}}\ z\isaliteral{5C3C5E697375623E}{}\isactrlisub n{\isaliteral{29}{\isacharparenright}}} will be the union \isa{fa{\isaliteral{28}{\isacharparenleft}}bc\isaliteral{5C3C5E697375623E}{}\isactrlisub {\isadigit{1}}{\isaliteral{29}{\isacharparenright}}\ {\isaliteral{5C3C756E696F6E3E}{\isasymunion}}\ {\isaliteral{5C3C646F74733E}{\isasymdots}}\ {\isaliteral{5C3C756E696F6E3E}{\isasymunion}}\ fa{\isaliteral{28}{\isacharparenleft}}bc\isaliteral{5C3C5E697375623E}{}\isactrlisub k{\isaliteral{29}{\isacharparenright}}} where we will define below what \isa{fa} for a binding
  clause means. We only show the details for the mode \isacommand{binds (set)} (the other modes are similar). 
  Suppose a binding clause \isa{bc\isaliteral{5C3C5E697375623E}{}\isactrlisub i} is of the form 
  
  \[
  \mbox{\isacommand{binds (set)} \isa{b\isaliteral{5C3C5E697375623E}{}\isactrlisub {\isadigit{1}}{\isaliteral{5C3C646F74733E}{\isasymdots}}b\isaliteral{5C3C5E697375623E}{}\isactrlisub p} \isacommand{in} \isa{d\isaliteral{5C3C5E697375623E}{}\isactrlisub {\isadigit{1}}{\isaliteral{5C3C646F74733E}{\isasymdots}}d\isaliteral{5C3C5E697375623E}{}\isactrlisub q}}
  \]\smallskip
  
  \noindent
  in which the body-labels \isa{d}$_{1..q}$ refer to types \isa{ty}$_{1..q}$, and the binders \isa{b}$_{1..p}$ either refer to labels of
  atom types (in case of shallow binders) or to binding functions taking a
  single label as argument (in case of deep binders). Assuming \isa{D}
  stands for the set of free atoms of the bodies, \isa{B} for the set of
  binding atoms in the binders and \isa{B{\isaliteral{27}{\isacharprime}}} for the set of free atoms in
  non-recursive deep binders, then the free atoms of the binding clause \isa{bc\isaliteral{5C3C5E697375623E}{}\isactrlisub i} are

  \begin{equation}\label{fadef}
  \mbox{\isa{fa{\isaliteral{28}{\isacharparenleft}}bc\isaliteral{5C3C5E697375623E}{}\isactrlisub i{\isaliteral{29}{\isacharparenright}}\ {\isaliteral{5C3C65717569763E}{\isasymequiv}}\ {\isaliteral{28}{\isacharparenleft}}D\ {\isaliteral{2D}{\isacharminus}}\ B{\isaliteral{29}{\isacharparenright}}\ {\isaliteral{5C3C756E696F6E3E}{\isasymunion}}\ B{\isaliteral{27}{\isacharprime}}}}.
  \end{equation}\smallskip
  
  \noindent
  The set \isa{D} is formally defined as
  
  \[
  \isa{D\ {\isaliteral{5C3C65717569763E}{\isasymequiv}}\ fa{\isaliteral{5F}{\isacharunderscore}}ty\isaliteral{5C3C5E697375623E}{}\isactrlisub {\isadigit{1}}\ d\isaliteral{5C3C5E697375623E}{}\isactrlisub {\isadigit{1}}\ {\isaliteral{5C3C756E696F6E3E}{\isasymunion}}\ {\isaliteral{2E}{\isachardot}}{\isaliteral{2E}{\isachardot}}{\isaliteral{2E}{\isachardot}}\ {\isaliteral{5C3C756E696F6E3E}{\isasymunion}}\ fa{\isaliteral{5F}{\isacharunderscore}}ty\isaliteral{5C3C5E697375623E}{}\isactrlisub q\ d\isaliteral{5C3C5E697375623E}{}\isactrlisub q}
  \]\smallskip
  
  \noindent
  where in case \isa{d\isaliteral{5C3C5E697375623E}{}\isactrlisub i} refers to one of the raw types \isa{ty}$_{1..n}$ from the 
  specification, the function \isa{fa{\isaliteral{5F}{\isacharunderscore}}ty\isaliteral{5C3C5E697375623E}{}\isactrlisub i} is the corresponding free-atom function 
  we are defining by recursion; otherwise we set \mbox{\isa{fa{\isaliteral{5F}{\isacharunderscore}}ty\isaliteral{5C3C5E697375623E}{}\isactrlisub i\ {\isaliteral{5C3C65717569763E}{\isasymequiv}}\ supp}}. The reason
  for the latter is that \isa{ty}$_i$ is not a type that is part of the specification, and
  we assume \isa{supp} is the generic function that characterises the free variables of 
  a type (in fact in the next section we will show that the free-variable functions we
  define here, are equal to the support once lifted to alpha-equivalence classes).
  
  In order to formally define the set \isa{B} we use the following auxiliary \isa{bn}-functions
  for atom types to which shallow binders may refer\\[-4mm]
  
  \begin{equation}\label{bnaux}\mbox{
  \begin{tabular}{r@ {\hspace{2mm}}c@ {\hspace{2mm}}l}
  \isa{bn\isaliteral{5C3C5E627375623E}{}\isactrlbsub atom\isaliteral{5C3C5E657375703E}{}\isactrlesup \ a} & \isa{{\isaliteral{5C3C65717569763E}{\isasymequiv}}} & \isa{{\isaliteral{7B}{\isacharbraceleft}}atom\ a{\isaliteral{7D}{\isacharbraceright}}}\\
  \isa{bn\isaliteral{5C3C5E627375623E}{}\isactrlbsub atom{\isaliteral{5F}{\isacharunderscore}}set\isaliteral{5C3C5E657375703E}{}\isactrlesup \ as} & \isa{{\isaliteral{5C3C65717569763E}{\isasymequiv}}} & \isa{atoms\ as}\\
  \isa{bn\isaliteral{5C3C5E627375623E}{}\isactrlbsub atom{\isaliteral{5F}{\isacharunderscore}}list\isaliteral{5C3C5E657375623E}{}\isactrlesub \ as} & \isa{{\isaliteral{5C3C65717569763E}{\isasymequiv}}} & \isa{atoms\ {\isaliteral{28}{\isacharparenleft}}set\ as{\isaliteral{29}{\isacharparenright}}}
  \end{tabular}}
  \end{equation}\smallskip
  
  \noindent 
  Like the function \isa{atom}, the function \isa{atoms} coerces 
  a set of atoms to a set of the generic atom type. 
  It is defined as  \isa{atoms\ as\ {\isaliteral{5C3C65717569763E}{\isasymequiv}}\ {\isaliteral{7B}{\isacharbraceleft}}atom\ a\ {\isaliteral{7C}{\isacharbar}}\ a\ {\isaliteral{5C3C696E3E}{\isasymin}}\ as{\isaliteral{7D}{\isacharbraceright}}}. 
  The set \isa{B} in \eqref{fadef} is then formally defined as
  
  \begin{equation}\label{bdef}
  \isa{B\ {\isaliteral{5C3C65717569763E}{\isasymequiv}}\ bn{\isaliteral{5F}{\isacharunderscore}}ty\isaliteral{5C3C5E697375623E}{}\isactrlisub {\isadigit{1}}\ b\isaliteral{5C3C5E697375623E}{}\isactrlisub {\isadigit{1}}\ {\isaliteral{5C3C756E696F6E3E}{\isasymunion}}\ {\isaliteral{2E}{\isachardot}}{\isaliteral{2E}{\isachardot}}{\isaliteral{2E}{\isachardot}}\ {\isaliteral{5C3C756E696F6E3E}{\isasymunion}}\ bn{\isaliteral{5F}{\isacharunderscore}}ty\isaliteral{5C3C5E697375623E}{}\isactrlisub p\ b\isaliteral{5C3C5E697375623E}{}\isactrlisub p}
  \end{equation}\smallskip

  \noindent 
  where we use the auxiliary binding functions from \eqref{bnaux} for shallow 
  binders (that means when \isa{ty}$_i$ is of type \isa{atom}, \isa{atom\ set} or
  \isa{atom\ list}). 

  The set \isa{B{\isaliteral{27}{\isacharprime}}} in \eqref{fadef} collects all free atoms in
  non-recursive deep binders. Let us assume these binders in the binding 
  clause \isa{bc\isaliteral{5C3C5E697375623E}{}\isactrlisub i} are

  \[
  \mbox{\isa{bn\isaliteral{5C3C5E697375623E}{}\isactrlisub {\isadigit{1}}\ l\isaliteral{5C3C5E697375623E}{}\isactrlisub {\isadigit{1}}{\isaliteral{2C}{\isacharcomma}}\ {\isaliteral{5C3C646F74733E}{\isasymdots}}{\isaliteral{2C}{\isacharcomma}}\ bn\isaliteral{5C3C5E697375623E}{}\isactrlisub r\ l\isaliteral{5C3C5E697375623E}{}\isactrlisub r}}
  \]\smallskip
  
  \noindent
  with \isa{l}$_{1..r}$ $\subseteq$ \isa{b}$_{1..p}$ and 
  none of the \isa{l}$_{1..r}$ being among the bodies
  \isa{d}$_{1..q}$. The set \isa{B{\isaliteral{27}{\isacharprime}}} is defined as
  
  \begin{equation}\label{bprimedef}
  \isa{B{\isaliteral{27}{\isacharprime}}\ {\isaliteral{5C3C65717569763E}{\isasymequiv}}\ fa{\isaliteral{5F}{\isacharunderscore}}bn\isaliteral{5C3C5E697375623E}{}\isactrlisub {\isadigit{1}}\ l\isaliteral{5C3C5E697375623E}{}\isactrlisub {\isadigit{1}}\ {\isaliteral{5C3C756E696F6E3E}{\isasymunion}}\ {\isaliteral{2E}{\isachardot}}{\isaliteral{2E}{\isachardot}}{\isaliteral{2E}{\isachardot}}\ {\isaliteral{5C3C756E696F6E3E}{\isasymunion}}\ fa{\isaliteral{5F}{\isacharunderscore}}bn\isaliteral{5C3C5E697375623E}{}\isactrlisub r\ l\isaliteral{5C3C5E697375623E}{}\isactrlisub r}
  \end{equation}\smallskip
  
  \noindent
  This completes all clauses for the free-atom functions \isa{fa{\isaliteral{5F}{\isacharunderscore}}ty}$_{1..n}$.

  Note that for non-recursive deep binders, we have to add in \eqref{fadef}
  the set of atoms that are left unbound by the binding functions \isa{bn}$_{1..m}$. We used for
  the definition of this set the functions \isa{fa{\isaliteral{5F}{\isacharunderscore}}bn}$_{1..m}$. The
  definition for those functions needs to be extracted from the clauses the
  user provided for \isa{bn}$_{1..m}$ Assume the user specified a \isa{bn}-clause of the form
  
  \[
  \isa{bn\ {\isaliteral{28}{\isacharparenleft}}C\ z\isaliteral{5C3C5E697375623E}{}\isactrlisub {\isadigit{1}}\ {\isaliteral{5C3C646F74733E}{\isasymdots}}\ z\isaliteral{5C3C5E697375623E}{}\isactrlisub s{\isaliteral{29}{\isacharparenright}}\ {\isaliteral{3D}{\isacharequal}}\ rhs}
  \]\smallskip
  
  \noindent
  where the \isa{z}$_{1..s}$ are of types \isa{ty}$_{1..s}$. For 
  each of the arguments we calculate the free atoms as follows:
  
  \[\mbox{
  \begin{tabular}{c@ {\hspace{2mm}}p{0.9\textwidth}}
  $\bullet$ & \isa{fa{\isaliteral{5F}{\isacharunderscore}}ty\isaliteral{5C3C5E697375623E}{}\isactrlisub i\ z\isaliteral{5C3C5E697375623E}{}\isactrlisub i} provided \isa{z\isaliteral{5C3C5E697375623E}{}\isactrlisub i} does not occur in \isa{rhs}\\ 
  & (that means nothing is bound in \isa{z\isaliteral{5C3C5E697375623E}{}\isactrlisub i} by the binding function),\smallskip\\
  $\bullet$ & \isa{fa{\isaliteral{5F}{\isacharunderscore}}bn\isaliteral{5C3C5E697375623E}{}\isactrlisub i\ z\isaliteral{5C3C5E697375623E}{}\isactrlisub i} provided \isa{z\isaliteral{5C3C5E697375623E}{}\isactrlisub i} occurs in  \isa{rhs}
  with the recursive call \isa{bn\isaliteral{5C3C5E697375623E}{}\isactrlisub i\ z\isaliteral{5C3C5E697375623E}{}\isactrlisub i}\\
  & (that means whatever is `left over' from the \isa{bn}-function is free)\smallskip\\
  $\bullet$ & \isa{{\isaliteral{5C3C656D7074797365743E}{\isasymemptyset}}} provided \isa{z\isaliteral{5C3C5E697375623E}{}\isactrlisub i} occurs in  \isa{rhs},
  but without a recursive call\\
  & (that means \isa{z\isaliteral{5C3C5E697375623E}{}\isactrlisub i} is supposed to become bound by the binding function)\\
  \end{tabular}}
  \]\smallskip
  
  \noindent
  For defining \isa{fa{\isaliteral{5F}{\isacharunderscore}}bn\ {\isaliteral{28}{\isacharparenleft}}C\ z\isaliteral{5C3C5E697375623E}{}\isactrlisub {\isadigit{1}}\ {\isaliteral{5C3C646F74733E}{\isasymdots}}\ z\isaliteral{5C3C5E697375623E}{}\isactrlisub n{\isaliteral{29}{\isacharparenright}}} we just union up all these sets.
 
  To see how these definitions work in practice, let us reconsider the
  term-constructors \isa{Let} and \isa{Let{\isaliteral{5F}{\isacharunderscore}}rec} shown in
  \eqref{letrecs} together with the term-constructors for assignments \isa{ANil} and \isa{ACons}. Since there is a binding function defined for
  assignments, we have three free-atom functions, namely \isa{fa\isaliteral{5C3C5E627375623E}{}\isactrlbsub trm\isaliteral{5C3C5E657375623E}{}\isactrlesub }, \isa{fa\isaliteral{5C3C5E627375623E}{}\isactrlbsub assn\isaliteral{5C3C5E657375623E}{}\isactrlesub } and \isa{fa\isaliteral{5C3C5E627375623E}{}\isactrlbsub bn\isaliteral{5C3C5E657375623E}{}\isactrlesub } as follows:
  
  \[\mbox{
  \begin{tabular}{@ {}l@ {\hspace{1mm}}c@ {\hspace{1mm}}l@ {}}
  \isa{fa\isaliteral{5C3C5E627375623E}{}\isactrlbsub trm\isaliteral{5C3C5E657375623E}{}\isactrlesub \ {\isaliteral{28}{\isacharparenleft}}Let\ as\ t{\isaliteral{29}{\isacharparenright}}} & \isa{{\isaliteral{5C3C65717569763E}{\isasymequiv}}} & \isa{{\isaliteral{28}{\isacharparenleft}}fa\isaliteral{5C3C5E627375623E}{}\isactrlbsub trm\isaliteral{5C3C5E657375623E}{}\isactrlesub \ t\ {\isaliteral{2D}{\isacharminus}}\ set\ {\isaliteral{28}{\isacharparenleft}}bn\ as{\isaliteral{29}{\isacharparenright}}{\isaliteral{29}{\isacharparenright}}\ {\isaliteral{5C3C756E696F6E3E}{\isasymunion}}\ fa\isaliteral{5C3C5E627375623E}{}\isactrlbsub bn\isaliteral{5C3C5E657375623E}{}\isactrlesub \ as}\\
  \isa{fa\isaliteral{5C3C5E627375623E}{}\isactrlbsub trm\isaliteral{5C3C5E657375623E}{}\isactrlesub \ {\isaliteral{28}{\isacharparenleft}}Let{\isaliteral{5F}{\isacharunderscore}}rec\ as\ t{\isaliteral{29}{\isacharparenright}}} & \isa{{\isaliteral{5C3C65717569763E}{\isasymequiv}}} & \isa{{\isaliteral{28}{\isacharparenleft}}fa\isaliteral{5C3C5E627375623E}{}\isactrlbsub assn\isaliteral{5C3C5E657375623E}{}\isactrlesub \ as\ {\isaliteral{5C3C756E696F6E3E}{\isasymunion}}\ fa\isaliteral{5C3C5E627375623E}{}\isactrlbsub trm\isaliteral{5C3C5E657375623E}{}\isactrlesub \ t{\isaliteral{29}{\isacharparenright}}\ {\isaliteral{2D}{\isacharminus}}\ set\ {\isaliteral{28}{\isacharparenleft}}bn\ as{\isaliteral{29}{\isacharparenright}}}\smallskip\\

  \isa{fa\isaliteral{5C3C5E627375623E}{}\isactrlbsub assn\isaliteral{5C3C5E657375623E}{}\isactrlesub \ {\isaliteral{28}{\isacharparenleft}}ANil{\isaliteral{29}{\isacharparenright}}} & \isa{{\isaliteral{5C3C65717569763E}{\isasymequiv}}} & \isa{{\isaliteral{5C3C656D7074797365743E}{\isasymemptyset}}}\\
  \isa{fa\isaliteral{5C3C5E627375623E}{}\isactrlbsub assn\isaliteral{5C3C5E657375623E}{}\isactrlesub \ {\isaliteral{28}{\isacharparenleft}}ACons\ a\ t\ as{\isaliteral{29}{\isacharparenright}}} & \isa{{\isaliteral{5C3C65717569763E}{\isasymequiv}}} & \isa{{\isaliteral{28}{\isacharparenleft}}supp\ a{\isaliteral{29}{\isacharparenright}}\ {\isaliteral{5C3C756E696F6E3E}{\isasymunion}}\ {\isaliteral{28}{\isacharparenleft}}fa\isaliteral{5C3C5E627375623E}{}\isactrlbsub trm\isaliteral{5C3C5E657375623E}{}\isactrlesub \ t{\isaliteral{29}{\isacharparenright}}\ {\isaliteral{5C3C756E696F6E3E}{\isasymunion}}\ {\isaliteral{28}{\isacharparenleft}}fa\isaliteral{5C3C5E627375623E}{}\isactrlbsub assn\isaliteral{5C3C5E657375623E}{}\isactrlesub \ as{\isaliteral{29}{\isacharparenright}}}\smallskip\\

  \isa{fa\isaliteral{5C3C5E627375623E}{}\isactrlbsub bn\isaliteral{5C3C5E657375623E}{}\isactrlesub \ {\isaliteral{28}{\isacharparenleft}}ANil{\isaliteral{29}{\isacharparenright}}} & \isa{{\isaliteral{5C3C65717569763E}{\isasymequiv}}} & \isa{{\isaliteral{5C3C656D7074797365743E}{\isasymemptyset}}}\\
  \isa{fa\isaliteral{5C3C5E627375623E}{}\isactrlbsub bn\isaliteral{5C3C5E657375623E}{}\isactrlesub \ {\isaliteral{28}{\isacharparenleft}}ACons\ a\ t\ as{\isaliteral{29}{\isacharparenright}}} & \isa{{\isaliteral{5C3C65717569763E}{\isasymequiv}}} & \isa{{\isaliteral{28}{\isacharparenleft}}fa\isaliteral{5C3C5E627375623E}{}\isactrlbsub trm\isaliteral{5C3C5E657375623E}{}\isactrlesub \ t{\isaliteral{29}{\isacharparenright}}\ {\isaliteral{5C3C756E696F6E3E}{\isasymunion}}\ {\isaliteral{28}{\isacharparenleft}}fa\isaliteral{5C3C5E627375623E}{}\isactrlbsub bn\isaliteral{5C3C5E657375623E}{}\isactrlesub \ as{\isaliteral{29}{\isacharparenright}}}
  \end{tabular}}
  \]\smallskip

  \noindent
  Recall that \isa{ANil} and \isa{ACons} have no binding clause in the
  specification. The corresponding free-atom function \isa{fa\isaliteral{5C3C5E627375623E}{}\isactrlbsub assn\isaliteral{5C3C5E657375623E}{}\isactrlesub } therefore returns all free atoms of an assignment
  (in case of \isa{ACons}, they are given in terms of \isa{supp}, \isa{fa\isaliteral{5C3C5E627375623E}{}\isactrlbsub trm\isaliteral{5C3C5E657375623E}{}\isactrlesub } and \isa{fa\isaliteral{5C3C5E627375623E}{}\isactrlbsub assn\isaliteral{5C3C5E657375623E}{}\isactrlesub }). The binding
  only takes place in \isa{Let} and \isa{Let{\isaliteral{5F}{\isacharunderscore}}rec}. In case of \isa{Let}, the binding clause specifies that all atoms given by \isa{set\ {\isaliteral{28}{\isacharparenleft}}bn\ as{\isaliteral{29}{\isacharparenright}}} have to be bound in \isa{t}. Therefore we have to subtract \isa{set\ {\isaliteral{28}{\isacharparenleft}}bn\ as{\isaliteral{29}{\isacharparenright}}} from \isa{fa\isaliteral{5C3C5E627375623E}{}\isactrlbsub trm\isaliteral{5C3C5E657375623E}{}\isactrlesub \ t}. However, we also need
  to add all atoms that are free in \isa{as}. This is in contrast with
  \isa{Let{\isaliteral{5F}{\isacharunderscore}}rec} where we have a recursive binder to bind all occurrences
  of the atoms in \isa{set\ {\isaliteral{28}{\isacharparenleft}}bn\ as{\isaliteral{29}{\isacharparenright}}} also inside \isa{as}. Therefore we
  have to subtract \isa{set\ {\isaliteral{28}{\isacharparenleft}}bn\ as{\isaliteral{29}{\isacharparenright}}} from both \isa{fa\isaliteral{5C3C5E627375623E}{}\isactrlbsub trm\isaliteral{5C3C5E657375623E}{}\isactrlesub \ t} and \isa{fa\isaliteral{5C3C5E627375623E}{}\isactrlbsub assn\isaliteral{5C3C5E657375623E}{}\isactrlesub \ as}. Like the
  function \isa{bn}, the function \isa{fa\isaliteral{5C3C5E627375623E}{}\isactrlbsub bn\isaliteral{5C3C5E657375623E}{}\isactrlesub } traverses
  the list of assignments, but instead returns the free atoms, which means in
  this example the free atoms in the argument \isa{t}.

  An interesting point in this example is that a `naked' assignment (\isa{ANil} or \isa{ACons}) does not bind any atoms, even if the binding
  function is specified over assignments. Only in the context of a \isa{Let}
  or \isa{Let{\isaliteral{5F}{\isacharunderscore}}rec}, where the binding clauses are given, will some atoms
  actually become bound.  This is a phenomenon that has also been pointed out
  in \cite{ott-jfp}. For us this observation is crucial, because we would not
  be able to lift the \isa{bn}-functions to alpha-equated terms if they
  act on atoms that are bound. In that case, these functions would \emph{not}
  respect alpha-equivalence.

  Having the free-atom functions at our disposal, we can next define the 
  alpha-equivalence relations for the raw types \isa{ty}$_{1..n}$. We write them as
  
  \[
  \mbox{\isa{{\isaliteral{5C3C617070726F783E}{\isasymapprox}}ty}$_{1..n}$}.
  \]\smallskip
  
  \noindent
  Like with the free-atom functions, we also need to
  define auxiliary alpha-equivalence relations 
  
  \[
  \mbox{\isa{{\isaliteral{5C3C617070726F783E}{\isasymapprox}}bn\isaliteral{5C3C5E697375623E}{}\isactrlisub }$_{1..m}$}
  \]\smallskip
  
  \noindent
  for the binding functions \isa{bn}$_{1..m}$, 
  To simplify our definitions we will use the following abbreviations for
  \emph{compound equivalence relations} and \emph{compound free-atom functions} acting on tuples.
  
  \[\mbox{
  \begin{tabular}{r@ {\hspace{2mm}}c@ {\hspace{2mm}}l}
  \isa{{\isaliteral{28}{\isacharparenleft}}x\isaliteral{5C3C5E697375623E}{}\isactrlisub {\isadigit{1}}{\isaliteral{2C}{\isacharcomma}}{\isaliteral{5C3C646F74733E}{\isasymdots}}{\isaliteral{2C}{\isacharcomma}}\ x\isaliteral{5C3C5E697375623E}{}\isactrlisub n{\isaliteral{29}{\isacharparenright}}\ {\isaliteral{28}{\isacharparenleft}}R\isaliteral{5C3C5E697375623E}{}\isactrlisub {\isadigit{1}}{\isaliteral{2C}{\isacharcomma}}{\isaliteral{5C3C646F74733E}{\isasymdots}}{\isaliteral{2C}{\isacharcomma}}\ R\isaliteral{5C3C5E697375623E}{}\isactrlisub n{\isaliteral{29}{\isacharparenright}}\ {\isaliteral{28}{\isacharparenleft}}y\isaliteral{5C3C5E697375623E}{}\isactrlisub {\isadigit{1}}{\isaliteral{2C}{\isacharcomma}}{\isaliteral{5C3C646F74733E}{\isasymdots}}{\isaliteral{2C}{\isacharcomma}}\ y\isaliteral{5C3C5E697375623E}{}\isactrlisub n{\isaliteral{29}{\isacharparenright}}} & \isa{{\isaliteral{5C3C65717569763E}{\isasymequiv}}} &
  \isa{x\isaliteral{5C3C5E697375623E}{}\isactrlisub {\isadigit{1}}\ R\isaliteral{5C3C5E697375623E}{}\isactrlisub {\isadigit{1}}\ y\isaliteral{5C3C5E697375623E}{}\isactrlisub {\isadigit{1}}\ {\isaliteral{5C3C616E643E}{\isasymand}}\ {\isaliteral{5C3C646F74733E}{\isasymdots}}\ {\isaliteral{5C3C616E643E}{\isasymand}}\ x\isaliteral{5C3C5E697375623E}{}\isactrlisub n\ R\isaliteral{5C3C5E697375623E}{}\isactrlisub n\ y\isaliteral{5C3C5E697375623E}{}\isactrlisub n}\\
  \isa{{\isaliteral{28}{\isacharparenleft}}fa\isaliteral{5C3C5E697375623E}{}\isactrlisub {\isadigit{1}}{\isaliteral{2C}{\isacharcomma}}{\isaliteral{5C3C646F74733E}{\isasymdots}}{\isaliteral{2C}{\isacharcomma}}\ fa\isaliteral{5C3C5E697375623E}{}\isactrlisub n{\isaliteral{29}{\isacharparenright}}\ {\isaliteral{28}{\isacharparenleft}}x\isaliteral{5C3C5E697375623E}{}\isactrlisub {\isadigit{1}}{\isaliteral{2C}{\isacharcomma}}{\isaliteral{5C3C646F74733E}{\isasymdots}}{\isaliteral{2C}{\isacharcomma}}\ x\isaliteral{5C3C5E697375623E}{}\isactrlisub n{\isaliteral{29}{\isacharparenright}}} & \isa{{\isaliteral{5C3C65717569763E}{\isasymequiv}}} & \isa{fa\isaliteral{5C3C5E697375623E}{}\isactrlisub {\isadigit{1}}\ x\isaliteral{5C3C5E697375623E}{}\isactrlisub {\isadigit{1}}\ {\isaliteral{5C3C756E696F6E3E}{\isasymunion}}\ {\isaliteral{5C3C646F74733E}{\isasymdots}}\ {\isaliteral{5C3C756E696F6E3E}{\isasymunion}}\ fa\isaliteral{5C3C5E697375623E}{}\isactrlisub n\ x\isaliteral{5C3C5E697375623E}{}\isactrlisub n}\\
  \end{tabular}}
  \]\smallskip

  The alpha-equivalence relations are defined as inductive predicates
  having a single clause for each term-constructor. Assuming a
  term-constructor \isa{C} is of type \isa{ty} and has the binding clauses
  \isa{bc}$_{1..k}$, then the alpha-equivalence clause has the form
  
  \begin{equation}\label{gform}
  \mbox{\infer{\isa{C\ z\isaliteral{5C3C5E697375623E}{}\isactrlisub {\isadigit{1}}\ {\isaliteral{5C3C646F74733E}{\isasymdots}}\ z\isaliteral{5C3C5E697375623E}{}\isactrlisub n\ \ {\isaliteral{5C3C617070726F783E}{\isasymapprox}}ty\ \ C\ z{\isaliteral{5C3C5052494D453E}{\isasymPRIME}}\isaliteral{5C3C5E697375623E}{}\isactrlisub {\isadigit{1}}\ {\isaliteral{5C3C646F74733E}{\isasymdots}}\ z{\isaliteral{5C3C5052494D453E}{\isasymPRIME}}\isaliteral{5C3C5E697375623E}{}\isactrlisub n}}
  {\isa{prems{\isaliteral{28}{\isacharparenleft}}bc\isaliteral{5C3C5E697375623E}{}\isactrlisub {\isadigit{1}}{\isaliteral{29}{\isacharparenright}}\ {\isaliteral{5C3C646F74733E}{\isasymdots}}\ prems{\isaliteral{28}{\isacharparenleft}}bc\isaliteral{5C3C5E697375623E}{}\isactrlisub k{\isaliteral{29}{\isacharparenright}}}}} 
  \end{equation}\smallskip

  \noindent
  The task below is to specify what the premises corresponding to a binding
  clause are. To understand better what the general pattern is, let us first 
  treat the special instance where \isa{bc\isaliteral{5C3C5E697375623E}{}\isactrlisub i} is the empty binding clause 
  of the form

  \[
  \mbox{\isacommand{binds (set)} \isa{{\isaliteral{5C3C656D7074797365743E}{\isasymemptyset}}} \isacommand{in} \isa{d\isaliteral{5C3C5E697375623E}{}\isactrlisub {\isadigit{1}}{\isaliteral{5C3C646F74733E}{\isasymdots}}d\isaliteral{5C3C5E697375623E}{}\isactrlisub q}.}
  \]\smallskip

  \noindent
  In this binding clause no atom is bound and we only have to `alpha-relate'
  the bodies. For this we build first the tuples \isa{D\ {\isaliteral{5C3C65717569763E}{\isasymequiv}}\ {\isaliteral{28}{\isacharparenleft}}d\isaliteral{5C3C5E697375623E}{}\isactrlisub {\isadigit{1}}{\isaliteral{2C}{\isacharcomma}}{\isaliteral{5C3C646F74733E}{\isasymdots}}{\isaliteral{2C}{\isacharcomma}}\ d\isaliteral{5C3C5E697375623E}{}\isactrlisub q{\isaliteral{29}{\isacharparenright}}} and \isa{D{\isaliteral{27}{\isacharprime}}\ {\isaliteral{5C3C65717569763E}{\isasymequiv}}\ {\isaliteral{28}{\isacharparenleft}}d{\isaliteral{5C3C5052494D453E}{\isasymPRIME}}\isaliteral{5C3C5E697375623E}{}\isactrlisub {\isadigit{1}}{\isaliteral{2C}{\isacharcomma}}{\isaliteral{5C3C646F74733E}{\isasymdots}}{\isaliteral{2C}{\isacharcomma}}\ d{\isaliteral{5C3C5052494D453E}{\isasymPRIME}}\isaliteral{5C3C5E697375623E}{}\isactrlisub q{\isaliteral{29}{\isacharparenright}}}
  whereby the labels \isa{d}$_{1..q}$ refer to some of the arguments \isa{z}$_{1..n}$ and respectively \isa{d{\isaliteral{5C3C5052494D453E}{\isasymPRIME}}}$_{1..q}$ to some of the \isa{z{\isaliteral{5C3C5052494D453E}{\isasymPRIME}}}$_{1..n}$ in \eqref{gform}. In order to relate two such
  tuples we define the compound alpha-equivalence relation \isa{R} as
  follows

  \begin{equation}\label{rempty}
  \mbox{\isa{R\ {\isaliteral{5C3C65717569763E}{\isasymequiv}}\ {\isaliteral{28}{\isacharparenleft}}R\isaliteral{5C3C5E697375623E}{}\isactrlisub {\isadigit{1}}{\isaliteral{2C}{\isacharcomma}}{\isaliteral{5C3C646F74733E}{\isasymdots}}{\isaliteral{2C}{\isacharcomma}}\ R\isaliteral{5C3C5E697375623E}{}\isactrlisub q{\isaliteral{29}{\isacharparenright}}}}
  \end{equation}\smallskip

  \noindent
  with \isa{R\isaliteral{5C3C5E697375623E}{}\isactrlisub i} being \isa{{\isaliteral{5C3C617070726F783E}{\isasymapprox}}ty\isaliteral{5C3C5E697375623E}{}\isactrlisub i} if the corresponding
  labels \isa{d\isaliteral{5C3C5E697375623E}{}\isactrlisub i} and \isa{d{\isaliteral{5C3C5052494D453E}{\isasymPRIME}}\isaliteral{5C3C5E697375623E}{}\isactrlisub i} refer to a
  recursive argument of \isa{C} and have type \isa{ty\isaliteral{5C3C5E697375623E}{}\isactrlisub i}; otherwise
  we take \isa{R\isaliteral{5C3C5E697375623E}{}\isactrlisub i} to be the equality \isa{{\isaliteral{3D}{\isacharequal}}}. Again the
  latter is because \isa{ty\isaliteral{5C3C5E697375623E}{}\isactrlisub i} is then not part of the specified types
  and alpha-equivalence of any previously defined type is supposed to coincide
  with equality.  This lets us now define the premise for an empty binding
  clause succinctly as \isa{prems{\isaliteral{28}{\isacharparenleft}}bc\isaliteral{5C3C5E697375623E}{}\isactrlisub i{\isaliteral{29}{\isacharparenright}}\ {\isaliteral{5C3C65717569763E}{\isasymequiv}}\ D\ R\ D{\isaliteral{27}{\isacharprime}}}, which can be
  unfolded to the series of premises
  
  \[
  \isa{d\isaliteral{5C3C5E697375623E}{}\isactrlisub {\isadigit{1}}\ R\isaliteral{5C3C5E697375623E}{}\isactrlisub {\isadigit{1}}\ d{\isaliteral{5C3C5052494D453E}{\isasymPRIME}}\isaliteral{5C3C5E697375623E}{}\isactrlisub {\isadigit{1}}\ \ {\isaliteral{5C3C646F74733E}{\isasymdots}}\ d\isaliteral{5C3C5E697375623E}{}\isactrlisub q\ R\isaliteral{5C3C5E697375623E}{}\isactrlisub q\ d{\isaliteral{5C3C5052494D453E}{\isasymPRIME}}\isaliteral{5C3C5E697375623E}{}\isactrlisub q}.
  \]\smallskip
  
  \noindent
  We will use the unfolded version in the examples below.

  Now suppose the binding clause \isa{bc\isaliteral{5C3C5E697375623E}{}\isactrlisub i} is of the general form 
  
  \begin{equation}\label{nonempty}
  \mbox{\isacommand{binds (set)} \isa{b\isaliteral{5C3C5E697375623E}{}\isactrlisub {\isadigit{1}}{\isaliteral{5C3C646F74733E}{\isasymdots}}b\isaliteral{5C3C5E697375623E}{}\isactrlisub p} \isacommand{in} \isa{d\isaliteral{5C3C5E697375623E}{}\isactrlisub {\isadigit{1}}{\isaliteral{5C3C646F74733E}{\isasymdots}}d\isaliteral{5C3C5E697375623E}{}\isactrlisub q}.}
  \end{equation}\smallskip

  \noindent
  In this case we define a premise \isa{P} using the relation
  $\approx_{\,\textit{set}}^{\textit{R}, \textit{fa}}$ given in Section~\ref{sec:binders} (similarly
  $\approx_{\,\textit{set+}}^{\textit{R}, \textit{fa}}$ and 
  $\approx_{\,\textit{list}}^{\textit{R}, \textit{fa}}$ for the other
  binding modes). As above, we first build the tuples \isa{D} and
  \isa{D{\isaliteral{27}{\isacharprime}}} for the bodies \isa{d}$_{1..q}$, and the corresponding
  compound alpha-relation \isa{R} (shown in \eqref{rempty}). 
  For $\approx_{\,\textit{set}}^{\textit{R}, \textit{fa}}$  we also need
  a compound free-atom function for the bodies defined as
  
  \[
  \mbox{\isa{fa\ {\isaliteral{5C3C65717569763E}{\isasymequiv}}\ {\isaliteral{28}{\isacharparenleft}}fa{\isaliteral{5F}{\isacharunderscore}}ty\isaliteral{5C3C5E697375623E}{}\isactrlisub {\isadigit{1}}{\isaliteral{2C}{\isacharcomma}}{\isaliteral{5C3C646F74733E}{\isasymdots}}{\isaliteral{2C}{\isacharcomma}}\ fa{\isaliteral{5F}{\isacharunderscore}}ty\isaliteral{5C3C5E697375623E}{}\isactrlisub q{\isaliteral{29}{\isacharparenright}}}}
  \]\smallskip

  \noindent
  with the assumption that the \isa{d}$_{1..q}$ refer to arguments of types \isa{ty}$_{1..q}$.
  The last ingredient we need are the sets of atoms bound in the bodies.
  For this we take

  \[
  \isa{B\ {\isaliteral{5C3C65717569763E}{\isasymequiv}}\ bn{\isaliteral{5F}{\isacharunderscore}}ty\isaliteral{5C3C5E697375623E}{}\isactrlisub {\isadigit{1}}\ b\isaliteral{5C3C5E697375623E}{}\isactrlisub {\isadigit{1}}\ {\isaliteral{5C3C756E696F6E3E}{\isasymunion}}\ {\isaliteral{5C3C646F74733E}{\isasymdots}}\ {\isaliteral{5C3C756E696F6E3E}{\isasymunion}}\ bn{\isaliteral{5F}{\isacharunderscore}}ty\isaliteral{5C3C5E697375623E}{}\isactrlisub p\ b\isaliteral{5C3C5E697375623E}{}\isactrlisub p}\;.\\
  \]\smallskip

  \noindent
  Similarly for \isa{B{\isaliteral{27}{\isacharprime}}} using the labels \isa{b{\isaliteral{5C3C5052494D453E}{\isasymPRIME}}}$_{1..p}$. This 
  lets us formally define the premise \isa{P} for a non-empty binding clause as:
  
  \[
  \mbox{\isa{P\ {\isaliteral{5C3C65717569763E}{\isasymequiv}}\ {\isaliteral{28}{\isacharparenleft}}B{\isaliteral{2C}{\isacharcomma}}\ D{\isaliteral{29}{\isacharparenright}}\ {\isaliteral{5C3C617070726F783E}{\isasymapprox}}\,\raisebox{-1pt}{\makebox[0mm][l]{$_{\textit{set}}$}}\isaliteral{5C3C5E627375703E}{}\isactrlbsup R{\isaliteral{2C}{\isacharcomma}}\ fa\isaliteral{5C3C5E657375703E}{}\isactrlesup \ {\isaliteral{28}{\isacharparenleft}}B{\isaliteral{27}{\isacharprime}}{\isaliteral{2C}{\isacharcomma}}\ D{\isaliteral{27}{\isacharprime}}{\isaliteral{29}{\isacharparenright}}}}\;.
  \]\smallskip

  \noindent
  This premise accounts for alpha-equivalence of the bodies of the binding
  clause. However, in case the binders have non-recursive deep binders, this
  premise is not enough: we also have to `propagate' alpha-equivalence
  inside the structure of these binders. An example is \isa{Let} where we
  have to make sure the right-hand sides of assignments are
  alpha-equivalent. For this we use relations \isa{{\isaliteral{5C3C617070726F783E}{\isasymapprox}}bn}$_{1..m}$ (which we
  will define shortly).  Let us assume the non-recursive deep binders
  in \isa{bc\isaliteral{5C3C5E697375623E}{}\isactrlisub i} are
  
  \[
  \isa{bn\isaliteral{5C3C5E697375623E}{}\isactrlisub {\isadigit{1}}\ l\isaliteral{5C3C5E697375623E}{}\isactrlisub {\isadigit{1}}{\isaliteral{2C}{\isacharcomma}}\ {\isaliteral{5C3C646F74733E}{\isasymdots}}{\isaliteral{2C}{\isacharcomma}}\ bn\isaliteral{5C3C5E697375623E}{}\isactrlisub r\ l\isaliteral{5C3C5E697375623E}{}\isactrlisub r}.
  \]\smallskip
  
  \noindent
  The tuple \isa{L} consists then of all these binders \isa{{\isaliteral{28}{\isacharparenleft}}l\isaliteral{5C3C5E697375623E}{}\isactrlisub {\isadigit{1}}{\isaliteral{2C}{\isacharcomma}}{\isaliteral{5C3C646F74733E}{\isasymdots}}{\isaliteral{2C}{\isacharcomma}}l\isaliteral{5C3C5E697375623E}{}\isactrlisub r{\isaliteral{29}{\isacharparenright}}} 
  (similarly \isa{L{\isaliteral{27}{\isacharprime}}}) and the compound equivalence relation \isa{R{\isaliteral{27}{\isacharprime}}} 
  is \isa{{\isaliteral{28}{\isacharparenleft}}{\isaliteral{5C3C617070726F783E}{\isasymapprox}}bn\isaliteral{5C3C5E697375623E}{}\isactrlisub {\isadigit{1}}{\isaliteral{2C}{\isacharcomma}}{\isaliteral{5C3C646F74733E}{\isasymdots}}{\isaliteral{2C}{\isacharcomma}}{\isaliteral{5C3C617070726F783E}{\isasymapprox}}bn\isaliteral{5C3C5E697375623E}{}\isactrlisub r{\isaliteral{29}{\isacharparenright}}}.  All premises for \isa{bc\isaliteral{5C3C5E697375623E}{}\isactrlisub i} are then given by
  
  \[
  \isa{prems{\isaliteral{28}{\isacharparenleft}}bc\isaliteral{5C3C5E697375623E}{}\isactrlisub i{\isaliteral{29}{\isacharparenright}}\ {\isaliteral{5C3C65717569763E}{\isasymequiv}}\ P\ \ {\isaliteral{5C3C616E643E}{\isasymand}}\ \ \ L\ R{\isaliteral{27}{\isacharprime}}\ L{\isaliteral{27}{\isacharprime}}}
  \]\smallskip

  \noindent 
  The auxiliary alpha-equivalence relations \isa{{\isaliteral{5C3C617070726F783E}{\isasymapprox}}bn}$_{1..m}$ 
  in \isa{R{\isaliteral{27}{\isacharprime}}} are defined as follows: assuming a \isa{bn}-clause is of the form
  
  \[
  \isa{bn\ {\isaliteral{28}{\isacharparenleft}}C\ z\isaliteral{5C3C5E697375623E}{}\isactrlisub {\isadigit{1}}\ {\isaliteral{5C3C646F74733E}{\isasymdots}}\ z\isaliteral{5C3C5E697375623E}{}\isactrlisub s{\isaliteral{29}{\isacharparenright}}\ {\isaliteral{3D}{\isacharequal}}\ rhs}
  \]\smallskip
  
  \noindent
  where the \isa{z}$_{1..s}$ are of types \isa{ty}$_{1..s}$,
  then the corresponding alpha-equivalence clause for \isa{{\isaliteral{5C3C617070726F783E}{\isasymapprox}}bn} has the form
  
  \[
  \mbox{\infer{\isa{C\ z\isaliteral{5C3C5E697375623E}{}\isactrlisub {\isadigit{1}}\ {\isaliteral{5C3C646F74733E}{\isasymdots}}\ z\isaliteral{5C3C5E697375623E}{}\isactrlisub s\ {\isaliteral{5C3C617070726F783E}{\isasymapprox}}bn\ C\ z{\isaliteral{5C3C5052494D453E}{\isasymPRIME}}\isaliteral{5C3C5E697375623E}{}\isactrlisub {\isadigit{1}}\ {\isaliteral{5C3C646F74733E}{\isasymdots}}\ z{\isaliteral{5C3C5052494D453E}{\isasymPRIME}}\isaliteral{5C3C5E697375623E}{}\isactrlisub s}}
  {\isa{z\isaliteral{5C3C5E697375623E}{}\isactrlisub {\isadigit{1}}\ R\isaliteral{5C3C5E697375623E}{}\isactrlisub {\isadigit{1}}\ z{\isaliteral{5C3C5052494D453E}{\isasymPRIME}}\isaliteral{5C3C5E697375623E}{}\isactrlisub {\isadigit{1}}\ {\isaliteral{5C3C646F74733E}{\isasymdots}}\ z\isaliteral{5C3C5E697375623E}{}\isactrlisub s\ R\isaliteral{5C3C5E697375623E}{}\isactrlisub s\ z{\isaliteral{5C3C5052494D453E}{\isasymPRIME}}\isaliteral{5C3C5E697375623E}{}\isactrlisub s}}}
  \]\smallskip
  
  \noindent
  In this clause the relations \isa{R}$_{1..s}$ are given by 

  \[\mbox{
  \begin{tabular}{c@ {\hspace{2mm}}p{0.9\textwidth}}
  $\bullet$ & \isa{z\isaliteral{5C3C5E697375623E}{}\isactrlisub i\ {\isaliteral{5C3C617070726F783E}{\isasymapprox}}ty\ z{\isaliteral{5C3C5052494D453E}{\isasymPRIME}}\isaliteral{5C3C5E697375623E}{}\isactrlisub i} provided \isa{z\isaliteral{5C3C5E697375623E}{}\isactrlisub i} does not occur in \isa{rhs} and 
  is a recursive argument of \isa{C},\smallskip\\
  $\bullet$ & \isa{z\isaliteral{5C3C5E697375623E}{}\isactrlisub i\ {\isaliteral{3D}{\isacharequal}}\ z{\isaliteral{5C3C5052494D453E}{\isasymPRIME}}\isaliteral{5C3C5E697375623E}{}\isactrlisub i} provided \isa{z\isaliteral{5C3C5E697375623E}{}\isactrlisub i} does not occur in \isa{rhs}
  and is a non-recursive argument of \isa{C},\smallskip\\
  $\bullet$ & \isa{z\isaliteral{5C3C5E697375623E}{}\isactrlisub i\ {\isaliteral{5C3C617070726F783E}{\isasymapprox}}bn\isaliteral{5C3C5E697375623E}{}\isactrlisub i\ z{\isaliteral{5C3C5052494D453E}{\isasymPRIME}}\isaliteral{5C3C5E697375623E}{}\isactrlisub i} provided \isa{z\isaliteral{5C3C5E697375623E}{}\isactrlisub i} occurs in \isa{rhs}
  with the recursive call \isa{bn\isaliteral{5C3C5E697375623E}{}\isactrlisub i\ x\isaliteral{5C3C5E697375623E}{}\isactrlisub i} and\smallskip\\
  $\bullet$ & \isa{True} provided \isa{z\isaliteral{5C3C5E697375623E}{}\isactrlisub i} occurs in \isa{rhs} but without a
  recursive call.
  \end{tabular}}
  \]\smallskip

  \noindent
  This completes the definition of alpha-equivalence. As a sanity check, we can show
  that the premises of empty binding clauses are a special case of the clauses for 
  non-empty ones (we just have to unfold the definition of 
  $\approx_{\,\textit{set}}^{\textit{R}, \textit{fa}}$ and take \isa{{\isadigit{0}}}
  for the existentially quantified permutation).

  Again let us take a look at a concrete example for these definitions. For 
  the specification shown in \eqref{letrecs}
  we have three relations $\approx_{\textit{trm}}$, $\approx_{\textit{assn}}$ and
  $\approx_{\textit{bn}}$ with the following rules:

  \begin{equation}\label{rawalpha}\mbox{
  \begin{tabular}{@ {}c @ {}}
  \infer{\isa{Let\ as\ t\ {\isaliteral{5C3C617070726F783E}{\isasymapprox}}\isaliteral{5C3C5E627375623E}{}\isactrlbsub trm\isaliteral{5C3C5E657375623E}{}\isactrlesub \ Let\ as{\isaliteral{27}{\isacharprime}}\ t{\isaliteral{27}{\isacharprime}}}}
  {\isa{{\isaliteral{28}{\isacharparenleft}}bn\ as{\isaliteral{2C}{\isacharcomma}}\ t{\isaliteral{29}{\isacharparenright}}\ {\isaliteral{5C3C617070726F783E}{\isasymapprox}}\,\raisebox{-1pt}{\makebox[0mm][l]{$_{\textit{list}}$}}\isaliteral{5C3C5E627375703E}{}\isactrlbsup {\isaliteral{5C3C617070726F783E}{\isasymapprox}}\isaliteral{5C3C5E627375623E}{}\isactrlbsub trm\isaliteral{5C3C5E657375623E}{}\isactrlesub {\isaliteral{2C}{\isacharcomma}}\ fa\isaliteral{5C3C5E627375623E}{}\isactrlbsub trm\isaliteral{5C3C5E657375623E}{}\isactrlesub \isaliteral{5C3C5E657375703E}{}\isactrlesup \ {\isaliteral{28}{\isacharparenleft}}bn\ as{\isaliteral{27}{\isacharprime}}{\isaliteral{2C}{\isacharcomma}}\ t{\isaliteral{27}{\isacharprime}}{\isaliteral{29}{\isacharparenright}}} & 
  \hspace{5mm}\isa{as\ {\isaliteral{5C3C617070726F783E}{\isasymapprox}}\isaliteral{5C3C5E627375623E}{}\isactrlbsub bn\isaliteral{5C3C5E657375623E}{}\isactrlesub \ as{\isaliteral{27}{\isacharprime}}}}\\
  \\
  \makebox[0mm]{\infer{\isa{Let{\isaliteral{5F}{\isacharunderscore}}rec\ as\ t\ {\isaliteral{5C3C617070726F783E}{\isasymapprox}}\isaliteral{5C3C5E627375623E}{}\isactrlbsub trm\isaliteral{5C3C5E657375623E}{}\isactrlesub \ Let{\isaliteral{5F}{\isacharunderscore}}rec\ as{\isaliteral{27}{\isacharprime}}\ t{\isaliteral{27}{\isacharprime}}}}
  {\isa{{\isaliteral{28}{\isacharparenleft}}bn\ as{\isaliteral{2C}{\isacharcomma}}\ {\isaliteral{28}{\isacharparenleft}}as{\isaliteral{2C}{\isacharcomma}}\ t{\isaliteral{29}{\isacharparenright}}{\isaliteral{29}{\isacharparenright}}\ {\isaliteral{5C3C617070726F783E}{\isasymapprox}}\,\raisebox{-1pt}{\makebox[0mm][l]{$_{\textit{list}}$}}\isaliteral{5C3C5E627375703E}{}\isactrlbsup {\isaliteral{28}{\isacharparenleft}}{\isaliteral{5C3C617070726F783E}{\isasymapprox}}\isaliteral{5C3C5E627375623E}{}\isactrlbsub assn\isaliteral{5C3C5E657375623E}{}\isactrlesub {\isaliteral{2C}{\isacharcomma}}\ {\isaliteral{5C3C617070726F783E}{\isasymapprox}}\isaliteral{5C3C5E627375623E}{}\isactrlbsub trm\isaliteral{5C3C5E657375623E}{}\isactrlesub {\isaliteral{29}{\isacharparenright}}{\isaliteral{2C}{\isacharcomma}}\ {\isaliteral{28}{\isacharparenleft}}fa\isaliteral{5C3C5E627375623E}{}\isactrlbsub assn\isaliteral{5C3C5E657375623E}{}\isactrlesub {\isaliteral{2C}{\isacharcomma}}\ fa\isaliteral{5C3C5E627375623E}{}\isactrlbsub trm\isaliteral{5C3C5E657375623E}{}\isactrlesub {\isaliteral{29}{\isacharparenright}}\isaliteral{5C3C5E657375703E}{}\isactrlesup \ {\isaliteral{28}{\isacharparenleft}}bn\ as{\isaliteral{27}{\isacharprime}}{\isaliteral{2C}{\isacharcomma}}\ {\isaliteral{28}{\isacharparenleft}}as{\isaliteral{2C}{\isacharcomma}}\ t{\isaliteral{5C3C5052494D453E}{\isasymPRIME}}\ {\isaliteral{29}{\isacharparenright}}{\isaliteral{29}{\isacharparenright}}}}}\\
  \\

  \begin{tabular}{@ {}c @ {}}
  \infer{\isa{ANil\ {\isaliteral{5C3C617070726F783E}{\isasymapprox}}\isaliteral{5C3C5E627375623E}{}\isactrlbsub assn\isaliteral{5C3C5E657375623E}{}\isactrlesub \ ANil}}{}\hspace{9mm}
  \infer{\isa{ACons\ a\ t\ as\ {\isaliteral{5C3C617070726F783E}{\isasymapprox}}\isaliteral{5C3C5E627375623E}{}\isactrlbsub assn\isaliteral{5C3C5E657375623E}{}\isactrlesub \ ACons\ a{\isaliteral{27}{\isacharprime}}\ t{\isaliteral{27}{\isacharprime}}\ as}}
  {\isa{a\ {\isaliteral{3D}{\isacharequal}}\ a{\isaliteral{27}{\isacharprime}}} & \hspace{5mm}\isa{t\ {\isaliteral{5C3C617070726F783E}{\isasymapprox}}\isaliteral{5C3C5E627375623E}{}\isactrlbsub trm\isaliteral{5C3C5E657375623E}{}\isactrlesub \ t{\isaliteral{27}{\isacharprime}}} & \hspace{5mm}\isa{as\ {\isaliteral{5C3C617070726F783E}{\isasymapprox}}\isaliteral{5C3C5E627375623E}{}\isactrlbsub assn\isaliteral{5C3C5E657375623E}{}\isactrlesub \ as{\isaliteral{27}{\isacharprime}}}}
  \end{tabular}\\
  \\

  \begin{tabular}{@ {}c @ {}}
  \infer{\isa{ANil\ {\isaliteral{5C3C617070726F783E}{\isasymapprox}}\isaliteral{5C3C5E627375623E}{}\isactrlbsub bn\isaliteral{5C3C5E657375623E}{}\isactrlesub \ ANil}}{}\hspace{9mm}
  \infer{\isa{ACons\ a\ t\ as\ {\isaliteral{5C3C617070726F783E}{\isasymapprox}}\isaliteral{5C3C5E627375623E}{}\isactrlbsub bn\isaliteral{5C3C5E657375623E}{}\isactrlesub \ ACons\ a{\isaliteral{27}{\isacharprime}}\ t{\isaliteral{27}{\isacharprime}}\ as}}
  {\isa{t\ {\isaliteral{5C3C617070726F783E}{\isasymapprox}}\isaliteral{5C3C5E627375623E}{}\isactrlbsub trm\isaliteral{5C3C5E657375623E}{}\isactrlesub \ t{\isaliteral{27}{\isacharprime}}} & \hspace{5mm}\isa{as\ {\isaliteral{5C3C617070726F783E}{\isasymapprox}}\isaliteral{5C3C5E627375623E}{}\isactrlbsub bn\isaliteral{5C3C5E657375623E}{}\isactrlesub \ as{\isaliteral{27}{\isacharprime}}}}
  \end{tabular}
  \end{tabular}}
  \end{equation}\smallskip

  \noindent
  Notice the difference between  $\approx_{\textit{assn}}$ and
  $\approx_{\textit{bn}}$: the latter only `tracks' alpha-equivalence of 
  the components in an assignment that are \emph{not} bound. This is needed in the 
  clause for \isa{Let} (which has
  a non-recursive binder). 
  The underlying reason is that the terms inside an assignment are not meant 
  to be `under' the binder. Such a premise is \emph{not} needed in \isa{Let{\isaliteral{5F}{\isacharunderscore}}rec}, 
  because there all components of an assignment are `under' the binder. 
  Note also that in case of more than one body (that is in the \isa{Let{\isaliteral{5F}{\isacharunderscore}}rec}-case above)
  we need to parametrise the relation $\approx_{\textit{list}}$ with a compound
  equivalence relation and a compound free-atom function. This is because the
  corresponding binding clause specifies a binder with two bodies, namely
  \isa{as} and \isa{t}.%
\end{isamarkuptext}%
\isamarkuptrue%
\isamarkupsection{Establishing the Reasoning Infrastructure%
}
\isamarkuptrue%
\begin{isamarkuptext}%
Having made all necessary definitions for raw terms, we can start with
  establishing the reasoning infrastructure for the alpha-equated types \isa{ty{\isaliteral{5C3C414C3E}{\isasymAL}}}$_{1..n}$, that is the types the user originally specified. We
  give in this section and the next the proofs we need for establishing this
  infrastructure. One point of our work is that we have completely
  automated these proofs in Isabelle/HOL.

  First we establish that the free-variable functions, the binding functions and the
  alpha-equi\-va\-lences are equivariant.

  \begin{lem}\mbox{}\\
  \isa{{\isaliteral{28}{\isacharparenleft}}i{\isaliteral{29}{\isacharparenright}}} The functions \isa{fa{\isaliteral{5F}{\isacharunderscore}}ty}$_{1..n}$, \isa{fa{\isaliteral{5F}{\isacharunderscore}}bn}$_{1..m}$ and
  \isa{bn}$_{1..m}$ are equivariant.\\
  \isa{{\isaliteral{28}{\isacharparenleft}}ii{\isaliteral{29}{\isacharparenright}}} The relations \isa{{\isaliteral{5C3C617070726F783E}{\isasymapprox}}ty}$_{1..n}$ and
  \isa{{\isaliteral{5C3C617070726F783E}{\isasymapprox}}bn}$_{1..m}$ are equivariant.
  \end{lem}

  \begin{proof}
  The function package of Isabelle/HOL allows us to prove the first part by
  mutual induction over the definitions of the functions.\footnote{We have
  that the free-atom functions are terminating. From this the function
  package derives an induction principle~\cite{Krauss09}.} The second is by a
  straightforward induction over the rules of \isa{{\isaliteral{5C3C617070726F783E}{\isasymapprox}}ty}$_{1..n}$ and
  \isa{{\isaliteral{5C3C617070726F783E}{\isasymapprox}}bn}$_{1..m}$ using the first part.
  \end{proof}

  \noindent
  Next we establish that the alpha-equivalence relations defined in the
  previous section are indeed equivalence relations.

  \begin{lem}\label{equiv} 
  The relations \isa{{\isaliteral{5C3C617070726F783E}{\isasymapprox}}ty}$_{1..n}$ and \isa{{\isaliteral{5C3C617070726F783E}{\isasymapprox}}bn}$_{1..m}$ are
  equivalence relations.
  \end{lem}

  \begin{proof} 
  The proofs are by induction. The non-trivial
  cases involve premises built up by $\approx_{\textit{set}}$, 
  $\approx_{\textit{set+}}$ and $\approx_{\textit{list}}$. They 
  can be dealt with as in Lemma~\ref{alphaeq}. However, the transitivity
  case needs in addition the fact that the relations are equivariant. 
  \end{proof}

  \noindent 
  We can feed the last lemma into our quotient package and obtain new types
  \isa{ty}$^\alpha_{1..n}$ representing alpha-equated terms of types
  \isa{ty}$_{1..n}$. We also obtain definitions for the term-constructors
  \isa{C}$^\alpha_{1..k}$ from the raw term-constructors \isa{C}$_{1..k}$, and similar definitions for the free-atom functions \isa{fa{\isaliteral{5F}{\isacharunderscore}}ty}$^\alpha_{1..n}$ and \isa{fa{\isaliteral{5F}{\isacharunderscore}}bn}$^\alpha_{1..m}$ as well as the
  binding functions \isa{bn}$^\alpha_{1..m}$. However, these definitions
  are not really useful to the user, since they are given in terms of the
  isomorphisms we obtained by creating new types in Isabelle/HOL (recall the
  picture shown in the Introduction).

  The first useful property for the user is the fact that distinct 
  term-constructors are not equal, that is the property
  
  \begin{equation}\label{distinctalpha}
  \mbox{\isa{C}$^\alpha$~\isa{x\isaliteral{5C3C5E697375623E}{}\isactrlisub {\isadigit{1}}\ {\isaliteral{5C3C646F74733E}{\isasymdots}}\ x\isaliteral{5C3C5E697375623E}{}\isactrlisub r}~\isa{{\isaliteral{5C3C6E6F7465713E}{\isasymnoteq}}}~%
  \isa{D}$^\alpha$~\isa{y\isaliteral{5C3C5E697375623E}{}\isactrlisub {\isadigit{1}}\ {\isaliteral{5C3C646F74733E}{\isasymdots}}\ y\isaliteral{5C3C5E697375623E}{}\isactrlisub s}} 
  \end{equation}\smallskip
  
  \noindent
  whenever \isa{C}$^\alpha$~\isa{{\isaliteral{5C3C6E6F7465713E}{\isasymnoteq}}}~\isa{D}$^\alpha$.
  In order to derive this property, we use the definition of alpha-equivalence
  and establish that
  
  \begin{equation}\label{distinctraw}
  \mbox{\isa{C\ x\isaliteral{5C3C5E697375623E}{}\isactrlisub {\isadigit{1}}\ {\isaliteral{5C3C646F74733E}{\isasymdots}}\ x\isaliteral{5C3C5E697375623E}{}\isactrlisub r}\;$\not\approx$\isa{ty}\;\isa{D\ y\isaliteral{5C3C5E697375623E}{}\isactrlisub {\isadigit{1}}\ {\isaliteral{5C3C646F74733E}{\isasymdots}}\ y\isaliteral{5C3C5E697375623E}{}\isactrlisub s}}
  \end{equation}\smallskip

  \noindent
  holds for the corresponding raw term-constructors.
  In order to deduce \eqref{distinctalpha} from \eqref{distinctraw}, our quotient
  package needs to know that the raw term-constructors \isa{C} and \isa{D} 
  are \emph{respectful} w.r.t.~the alpha-equivalence relations (see \cite{Homeier05}).
  Given, for example, \isa{C} is of type \isa{ty} with argument types
  \isa{ty}$_{1..r}$, respectfulness amounts to showing that
  
  \[\mbox{
  \isa{C\ x\isaliteral{5C3C5E697375623E}{}\isactrlisub {\isadigit{1}}\ {\isaliteral{5C3C646F74733E}{\isasymdots}}\ x\isaliteral{5C3C5E697375623E}{}\isactrlisub r\ {\isaliteral{5C3C617070726F783E}{\isasymapprox}}ty\ C\ x{\isaliteral{5C3C5052494D453E}{\isasymPRIME}}\isaliteral{5C3C5E697375623E}{}\isactrlisub {\isadigit{1}}\ {\isaliteral{5C3C646F74733E}{\isasymdots}}\ x{\isaliteral{5C3C5052494D453E}{\isasymPRIME}}\isaliteral{5C3C5E697375623E}{}\isactrlisub r}
  }\]\smallskip

  \noindent
  holds under the assumptions \mbox{\isa{x\isaliteral{5C3C5E697375623E}{}\isactrlisub i\ {\isaliteral{5C3C617070726F783E}{\isasymapprox}}ty\isaliteral{5C3C5E697375623E}{}\isactrlisub i\ x{\isaliteral{5C3C5052494D453E}{\isasymPRIME}}\isaliteral{5C3C5E697375623E}{}\isactrlisub i}} whenever \isa{x\isaliteral{5C3C5E697375623E}{}\isactrlisub i}
  and \isa{x{\isaliteral{5C3C5052494D453E}{\isasymPRIME}}\isaliteral{5C3C5E697375623E}{}\isactrlisub i} are recursive arguments of \isa{C}, and
  \isa{x\isaliteral{5C3C5E697375623E}{}\isactrlisub i\ {\isaliteral{3D}{\isacharequal}}\ x{\isaliteral{5C3C5052494D453E}{\isasymPRIME}}\isaliteral{5C3C5E697375623E}{}\isactrlisub i} whenever they are non-recursive arguments 
  (similarly for \isa{D}). For this we have to show
  by induction over the definitions of alpha-equivalences the following 
  auxiliary implications

  \begin{equation}\label{fnresp}\mbox{
  \begin{tabular}{lll}
  \isa{x\ {\isaliteral{5C3C617070726F783E}{\isasymapprox}}ty\isaliteral{5C3C5E697375623E}{}\isactrlisub i\ x{\isaliteral{27}{\isacharprime}}} & implies & \isa{fa{\isaliteral{5F}{\isacharunderscore}}ty\isaliteral{5C3C5E697375623E}{}\isactrlisub i\ x\ {\isaliteral{3D}{\isacharequal}}\ fa{\isaliteral{5F}{\isacharunderscore}}ty\isaliteral{5C3C5E697375623E}{}\isactrlisub i\ x{\isaliteral{27}{\isacharprime}}}\\
  \isa{x\ {\isaliteral{5C3C617070726F783E}{\isasymapprox}}ty\isaliteral{5C3C5E697375623E}{}\isactrlisub l\ x{\isaliteral{27}{\isacharprime}}} & implies & \isa{fa{\isaliteral{5F}{\isacharunderscore}}bn\isaliteral{5C3C5E697375623E}{}\isactrlisub j\ x\ {\isaliteral{3D}{\isacharequal}}\ fa{\isaliteral{5F}{\isacharunderscore}}bn\isaliteral{5C3C5E697375623E}{}\isactrlisub j\ x{\isaliteral{27}{\isacharprime}}}\\
  \isa{x\ {\isaliteral{5C3C617070726F783E}{\isasymapprox}}ty\isaliteral{5C3C5E697375623E}{}\isactrlisub l\ x{\isaliteral{27}{\isacharprime}}} & implies & \isa{bn\isaliteral{5C3C5E697375623E}{}\isactrlisub j\ x\ {\isaliteral{3D}{\isacharequal}}\ bn\isaliteral{5C3C5E697375623E}{}\isactrlisub j\ x{\isaliteral{27}{\isacharprime}}}\\
  \isa{x\ {\isaliteral{5C3C617070726F783E}{\isasymapprox}}ty\isaliteral{5C3C5E697375623E}{}\isactrlisub l\ x{\isaliteral{27}{\isacharprime}}} & implies & \isa{x\ {\isaliteral{5C3C617070726F783E}{\isasymapprox}}bn\isaliteral{5C3C5E697375623E}{}\isactrlisub j\ x{\isaliteral{27}{\isacharprime}}}\\
  \end{tabular}
  }\end{equation}\smallskip
  
  \noindent
  whereby \isa{ty\isaliteral{5C3C5E697375623E}{}\isactrlisub l} is the type over which \isa{bn\isaliteral{5C3C5E697375623E}{}\isactrlisub j}
  is defined. Whereas the first, second and last implication are true by
  how we stated our definitions, the third \emph{only} holds because of our
  restriction imposed on the form of the binding functions---namely \emph{not}
  to return any bound atoms. In Ott, in contrast, the user may define \isa{bn}$_{1..m}$ so that they return bound atoms and in this case the third
  implication is \emph{not} true. A result is that in general the lifting of the
  corresponding binding functions in Ott to alpha-equated terms is impossible.
  Having established respectfulness for the raw term-constructors, the 
  quotient package is able to automatically deduce \eqref{distinctalpha} from 
  \eqref{distinctraw}.

  Next we can lift the permutation operations defined in \eqref{ceqvt}. In
  order to make this lifting to go through, we have to show that the
  permutation operations are respectful. This amounts to showing that the
  alpha-equivalence relations are equivariant, which
  we already established in Lemma~\ref{equiv}. As a result we can add the
  equations
  
  \begin{equation}\label{calphaeqvt}
  \isa{{\isaliteral{5C3C70693E}{\isasympi}}\ {\isaliteral{5C3C62756C6C65743E}{\isasymbullet}}\ {\isaliteral{28}{\isacharparenleft}}C\isaliteral{5C3C5E7375703E}{}\isactrlsup {\isaliteral{5C3C616C7068613E}{\isasymalpha}}\ x\isaliteral{5C3C5E697375623E}{}\isactrlisub {\isadigit{1}}\ {\isaliteral{5C3C646F74733E}{\isasymdots}}\ x\isaliteral{5C3C5E697375623E}{}\isactrlisub r{\isaliteral{29}{\isacharparenright}}\ {\isaliteral{3D}{\isacharequal}}\ C\isaliteral{5C3C5E7375703E}{}\isactrlsup {\isaliteral{5C3C616C7068613E}{\isasymalpha}}\ {\isaliteral{28}{\isacharparenleft}}{\isaliteral{5C3C70693E}{\isasympi}}\ {\isaliteral{5C3C62756C6C65743E}{\isasymbullet}}\ x\isaliteral{5C3C5E697375623E}{}\isactrlisub {\isadigit{1}}{\isaliteral{29}{\isacharparenright}}\ {\isaliteral{5C3C646F74733E}{\isasymdots}}\ {\isaliteral{28}{\isacharparenleft}}{\isaliteral{5C3C70693E}{\isasympi}}\ {\isaliteral{5C3C62756C6C65743E}{\isasymbullet}}\ x\isaliteral{5C3C5E697375623E}{}\isactrlisub r{\isaliteral{29}{\isacharparenright}}}
  \end{equation}\smallskip

  \noindent
  to our infrastructure. In a similar fashion we can lift the defining equations
  of the free-atom functions \isa{fa{\isaliteral{5F}{\isacharunderscore}}ty{\isaliteral{5C3C414C3E}{\isasymAL}}}$_{1..n}$ and
  \isa{fa{\isaliteral{5F}{\isacharunderscore}}bn{\isaliteral{5C3C414C3E}{\isasymAL}}}$_{1..m}$ as well as of the binding functions \isa{bn{\isaliteral{5C3C414C3E}{\isasymAL}}}$_{1..m}$ and size functions \isa{size{\isaliteral{5F}{\isacharunderscore}}ty{\isaliteral{5C3C414C3E}{\isasymAL}}}$_{1..n}$.
  The latter are defined automatically for the raw types \isa{ty}$_{1..n}$
  by the datatype package of Isabelle/HOL.

  We also need to lift the properties that characterise when two raw terms of the form
  
  \[
  \mbox{\isa{C\ x\isaliteral{5C3C5E697375623E}{}\isactrlisub {\isadigit{1}}\ {\isaliteral{5C3C646F74733E}{\isasymdots}}\ x\isaliteral{5C3C5E697375623E}{}\isactrlisub r\ {\isaliteral{5C3C617070726F783E}{\isasymapprox}}ty\ C\ x{\isaliteral{5C3C5052494D453E}{\isasymPRIME}}\isaliteral{5C3C5E697375623E}{}\isactrlisub {\isadigit{1}}\ {\isaliteral{5C3C646F74733E}{\isasymdots}}\ x{\isaliteral{5C3C5052494D453E}{\isasymPRIME}}\isaliteral{5C3C5E697375623E}{}\isactrlisub r}}
  \]\smallskip

  \noindent
  are alpha-equivalent. This gives us conditions when the corresponding
  alpha-equated terms are \emph{equal}, namely
  
  \[
  \isa{C\isaliteral{5C3C5E7375703E}{}\isactrlsup {\isaliteral{5C3C616C7068613E}{\isasymalpha}}\ x\isaliteral{5C3C5E697375623E}{}\isactrlisub {\isadigit{1}}\ {\isaliteral{5C3C646F74733E}{\isasymdots}}\ x\isaliteral{5C3C5E697375623E}{}\isactrlisub r\ {\isaliteral{3D}{\isacharequal}}\ C\isaliteral{5C3C5E7375703E}{}\isactrlsup {\isaliteral{5C3C616C7068613E}{\isasymalpha}}\ x{\isaliteral{5C3C5052494D453E}{\isasymPRIME}}\isaliteral{5C3C5E697375623E}{}\isactrlisub {\isadigit{1}}\ {\isaliteral{5C3C646F74733E}{\isasymdots}}\ x{\isaliteral{5C3C5052494D453E}{\isasymPRIME}}\isaliteral{5C3C5E697375623E}{}\isactrlisub r}.
  \]\smallskip
  
  \noindent
  We call these conditions \emph{quasi-injectivity}. They correspond to the
  premises in our alpha-equiva\-lence relations, except that the
  relations \isa{{\isaliteral{5C3C617070726F783E}{\isasymapprox}}ty}$_{1..n}$ are all replaced by equality (and similarly
  the free-atom and binding functions are replaced by their lifted
  counterparts). Recall the alpha-equivalence rules for \isa{Let} and
  \isa{Let{\isaliteral{5F}{\isacharunderscore}}rec} shown in \eqref{rawalpha}. For \isa{Let\isaliteral{5C3C5E7375703E}{}\isactrlsup {\isaliteral{5C3C616C7068613E}{\isasymalpha}}} and
  \isa{Let{\isaliteral{5F}{\isacharunderscore}}rec\isaliteral{5C3C5E7375703E}{}\isactrlsup {\isaliteral{5C3C616C7068613E}{\isasymalpha}}} we have

  \begin{equation}\label{alphalift}\mbox{
  \begin{tabular}{@ {}c @ {}}
  \infer{\isa{Let\isaliteral{5C3C5E7375703E}{}\isactrlsup {\isaliteral{5C3C616C7068613E}{\isasymalpha}}\ as\ t\ {\isaliteral{3D}{\isacharequal}}\ Let\isaliteral{5C3C5E7375703E}{}\isactrlsup {\isaliteral{5C3C616C7068613E}{\isasymalpha}}\ as{\isaliteral{27}{\isacharprime}}\ t{\isaliteral{27}{\isacharprime}}}}
  {\isa{{\isaliteral{28}{\isacharparenleft}}bn\isaliteral{5C3C5E7375703E}{}\isactrlsup {\isaliteral{5C3C616C7068613E}{\isasymalpha}}\ as{\isaliteral{2C}{\isacharcomma}}\ t{\isaliteral{29}{\isacharparenright}}\ {\isaliteral{5C3C617070726F783E}{\isasymapprox}}\,\raisebox{-1pt}{\makebox[0mm][l]{$_{\textit{list}}$}}\isaliteral{5C3C5E627375703E}{}\isactrlbsup {\isaliteral{3D}{\isacharequal}}{\isaliteral{2C}{\isacharcomma}}\ fa{\isaliteral{5C3C414C3E}{\isasymAL}}\isaliteral{5C3C5E627375623E}{}\isactrlbsub trm\isaliteral{5C3C5E657375623E}{}\isactrlesub \isaliteral{5C3C5E657375703E}{}\isactrlesup \ {\isaliteral{28}{\isacharparenleft}}bn\ as{\isaliteral{27}{\isacharprime}}{\isaliteral{2C}{\isacharcomma}}\ t{\isaliteral{27}{\isacharprime}}{\isaliteral{29}{\isacharparenright}}} & 
  \hspace{5mm}\isa{as\ {\isaliteral{5C3C617070726F783E}{\isasymapprox}}{\isaliteral{5C3C414C3E}{\isasymAL}}\isaliteral{5C3C5E627375623E}{}\isactrlbsub bn\isaliteral{5C3C5E657375623E}{}\isactrlesub \ as{\isaliteral{27}{\isacharprime}}}}\\
  \\
  \makebox[0mm]{\infer{\isa{Let{\isaliteral{5F}{\isacharunderscore}}rec\isaliteral{5C3C5E7375703E}{}\isactrlsup {\isaliteral{5C3C616C7068613E}{\isasymalpha}}\ as\ t\ {\isaliteral{3D}{\isacharequal}}\ Let{\isaliteral{5F}{\isacharunderscore}}rec\isaliteral{5C3C5E7375703E}{}\isactrlsup {\isaliteral{5C3C616C7068613E}{\isasymalpha}}\ as{\isaliteral{27}{\isacharprime}}\ t{\isaliteral{27}{\isacharprime}}}}
  {\isa{{\isaliteral{28}{\isacharparenleft}}bn\isaliteral{5C3C5E7375703E}{}\isactrlsup {\isaliteral{5C3C616C7068613E}{\isasymalpha}}\ as{\isaliteral{2C}{\isacharcomma}}\ {\isaliteral{28}{\isacharparenleft}}as{\isaliteral{2C}{\isacharcomma}}\ t{\isaliteral{29}{\isacharparenright}}{\isaliteral{29}{\isacharparenright}}\ {\isaliteral{5C3C617070726F783E}{\isasymapprox}}\,\raisebox{-1pt}{\makebox[0mm][l]{$_{\textit{list}}$}}\isaliteral{5C3C5E627375703E}{}\isactrlbsup {\isaliteral{28}{\isacharparenleft}}{\isaliteral{3D}{\isacharequal}}{\isaliteral{2C}{\isacharcomma}}\ {\isaliteral{3D}{\isacharequal}}{\isaliteral{29}{\isacharparenright}}{\isaliteral{2C}{\isacharcomma}}\ {\isaliteral{28}{\isacharparenleft}}fa{\isaliteral{5C3C414C3E}{\isasymAL}}\isaliteral{5C3C5E627375623E}{}\isactrlbsub assn\isaliteral{5C3C5E657375623E}{}\isactrlesub {\isaliteral{2C}{\isacharcomma}}\ fa{\isaliteral{5C3C414C3E}{\isasymAL}}\isaliteral{5C3C5E627375623E}{}\isactrlbsub trm\isaliteral{5C3C5E657375623E}{}\isactrlesub {\isaliteral{29}{\isacharparenright}}\isaliteral{5C3C5E657375703E}{}\isactrlesup \ {\isaliteral{28}{\isacharparenleft}}bn\isaliteral{5C3C5E7375703E}{}\isactrlsup {\isaliteral{5C3C616C7068613E}{\isasymalpha}}\ as{\isaliteral{27}{\isacharprime}}{\isaliteral{2C}{\isacharcomma}}\ {\isaliteral{28}{\isacharparenleft}}as{\isaliteral{2C}{\isacharcomma}}\ t{\isaliteral{5C3C5052494D453E}{\isasymPRIME}}\ {\isaliteral{29}{\isacharparenright}}{\isaliteral{29}{\isacharparenright}}}}}\\
  \end{tabular}}
  \end{equation}\smallskip

  We can also add to our infrastructure cases lemmas and a (mutual)
  induction principle for the types \isa{ty{\isaliteral{5C3C414C3E}{\isasymAL}}}$_{1..n}$. The cases
  lemmas allow the user to deduce a property \isa{P} by exhaustively
  analysing how an element of a type, say \isa{ty{\isaliteral{5C3C414C3E}{\isasymAL}}}$_i$, can be
  constructed (that means one case for each of the term-constructors in \isa{ty{\isaliteral{5C3C414C3E}{\isasymAL}}}$_i\,$). The lifted cases lemma for a type \isa{ty{\isaliteral{5C3C414C3E}{\isasymAL}}}$_i\,$ looks as follows

  \begin{equation}\label{cases}
  \infer{P}
  {\begin{array}{l}
  \isa{{\isaliteral{5C3C666F72616C6C3E}{\isasymforall}}x\isaliteral{5C3C5E697375623E}{}\isactrlisub {\isadigit{1}}{\isaliteral{5C3C646F74733E}{\isasymdots}}x\isaliteral{5C3C5E697375623E}{}\isactrlisub k{\isaliteral{2E}{\isachardot}}\ y\ {\isaliteral{3D}{\isacharequal}}\ C{\isaliteral{5C3C414C3E}{\isasymAL}}\isaliteral{5C3C5E697375623E}{}\isactrlisub {\isadigit{1}}\ x\isaliteral{5C3C5E697375623E}{}\isactrlisub {\isadigit{1}}\ {\isaliteral{5C3C646F74733E}{\isasymdots}}\ x\isaliteral{5C3C5E697375623E}{}\isactrlisub k\ {\isaliteral{5C3C52696768746172726F773E}{\isasymRightarrow}}\ P}\\
  \hspace{5mm}\vdots\\
  \isa{{\isaliteral{5C3C666F72616C6C3E}{\isasymforall}}x\isaliteral{5C3C5E697375623E}{}\isactrlisub {\isadigit{1}}{\isaliteral{5C3C646F74733E}{\isasymdots}}x\isaliteral{5C3C5E697375623E}{}\isactrlisub l{\isaliteral{2E}{\isachardot}}\ y\ {\isaliteral{3D}{\isacharequal}}\ C{\isaliteral{5C3C414C3E}{\isasymAL}}\isaliteral{5C3C5E697375623E}{}\isactrlisub m\ x\isaliteral{5C3C5E697375623E}{}\isactrlisub {\isadigit{1}}\ {\isaliteral{5C3C646F74733E}{\isasymdots}}\ x\isaliteral{5C3C5E697375623E}{}\isactrlisub l\ {\isaliteral{5C3C52696768746172726F773E}{\isasymRightarrow}}\ P}\\
  \end{array}}
  \end{equation}\smallskip

  \noindent
  where \isa{y} is a variable of type \isa{ty{\isaliteral{5C3C414C3E}{\isasymAL}}}$_i$ and \isa{P} is the 
  property that is established by the case analysis. Similarly, we have a (mutual) 
  induction principle for the types \isa{ty{\isaliteral{5C3C414C3E}{\isasymAL}}}$_{1..n}$, which is of the 
  form

   \begin{equation}\label{induct}
  \infer{\isa{P\isaliteral{5C3C5E697375623E}{}\isactrlisub {\isadigit{1}}\ y\isaliteral{5C3C5E697375623E}{}\isactrlisub {\isadigit{1}}\ {\isaliteral{5C3C616E643E}{\isasymand}}\ {\isaliteral{5C3C646F74733E}{\isasymdots}}\ {\isaliteral{5C3C616E643E}{\isasymand}}\ P\isaliteral{5C3C5E697375623E}{}\isactrlisub n\ y\isaliteral{5C3C5E697375623E}{}\isactrlisub n}}
  {\begin{array}{l}
  \isa{{\isaliteral{5C3C666F72616C6C3E}{\isasymforall}}x\isaliteral{5C3C5E697375623E}{}\isactrlisub {\isadigit{1}}{\isaliteral{5C3C646F74733E}{\isasymdots}}x\isaliteral{5C3C5E697375623E}{}\isactrlisub k{\isaliteral{2E}{\isachardot}}\ P\isaliteral{5C3C5E697375623E}{}\isactrlisub i\ x\isaliteral{5C3C5E697375623E}{}\isactrlisub i\ {\isaliteral{5C3C616E643E}{\isasymand}}\ {\isaliteral{5C3C646F74733E}{\isasymdots}}\ {\isaliteral{5C3C616E643E}{\isasymand}}\ P\isaliteral{5C3C5E697375623E}{}\isactrlisub j\ x\isaliteral{5C3C5E697375623E}{}\isactrlisub j\ {\isaliteral{5C3C52696768746172726F773E}{\isasymRightarrow}}\ P\ {\isaliteral{28}{\isacharparenleft}}C{\isaliteral{5C3C414C3E}{\isasymAL}}\isaliteral{5C3C5E697375623E}{}\isactrlisub {\isadigit{1}}\ x\isaliteral{5C3C5E697375623E}{}\isactrlisub {\isadigit{1}}\ {\isaliteral{5C3C646F74733E}{\isasymdots}}\ x\isaliteral{5C3C5E697375623E}{}\isactrlisub k{\isaliteral{29}{\isacharparenright}}}\\
  \hspace{5mm}\vdots\\
  \isa{{\isaliteral{5C3C666F72616C6C3E}{\isasymforall}}x\isaliteral{5C3C5E697375623E}{}\isactrlisub {\isadigit{1}}{\isaliteral{5C3C646F74733E}{\isasymdots}}x\isaliteral{5C3C5E697375623E}{}\isactrlisub l{\isaliteral{2E}{\isachardot}}\ P\isaliteral{5C3C5E697375623E}{}\isactrlisub r\ x\isaliteral{5C3C5E697375623E}{}\isactrlisub r\ {\isaliteral{5C3C616E643E}{\isasymand}}\ {\isaliteral{5C3C646F74733E}{\isasymdots}}\ {\isaliteral{5C3C616E643E}{\isasymand}}\ P\isaliteral{5C3C5E697375623E}{}\isactrlisub s\ x\isaliteral{5C3C5E697375623E}{}\isactrlisub s\ {\isaliteral{5C3C52696768746172726F773E}{\isasymRightarrow}}\ P\ {\isaliteral{28}{\isacharparenleft}}C{\isaliteral{5C3C414C3E}{\isasymAL}}\isaliteral{5C3C5E697375623E}{}\isactrlisub m\ x\isaliteral{5C3C5E697375623E}{}\isactrlisub {\isadigit{1}}\ {\isaliteral{5C3C646F74733E}{\isasymdots}}\ x\isaliteral{5C3C5E697375623E}{}\isactrlisub l{\isaliteral{29}{\isacharparenright}}}\\
  \end{array}}
  \end{equation}\smallskip

  \noindent
  whereby the \isa{P}$_{1..n}$ are the properties established by the
  induction, and the \isa{y}$_{1..n}$ are of type \isa{ty{\isaliteral{5C3C414C3E}{\isasymAL}}}$_{1..n}$. Note that for the term constructor \isa{C}$^\alpha_1$ the induction principle has a hypothesis of the form

  \[
  \mbox{\isa{{\isaliteral{5C3C666F72616C6C3E}{\isasymforall}}x\isaliteral{5C3C5E697375623E}{}\isactrlisub {\isadigit{1}}{\isaliteral{5C3C646F74733E}{\isasymdots}}x\isaliteral{5C3C5E697375623E}{}\isactrlisub k{\isaliteral{2E}{\isachardot}}\ P\isaliteral{5C3C5E697375623E}{}\isactrlisub i\ x\isaliteral{5C3C5E697375623E}{}\isactrlisub i\ {\isaliteral{5C3C616E643E}{\isasymand}}\ {\isaliteral{5C3C646F74733E}{\isasymdots}}\ {\isaliteral{5C3C616E643E}{\isasymand}}\ P\isaliteral{5C3C5E697375623E}{}\isactrlisub j\ x\isaliteral{5C3C5E697375623E}{}\isactrlisub j\ {\isaliteral{5C3C52696768746172726F773E}{\isasymRightarrow}}\ P\ {\isaliteral{28}{\isacharparenleft}}C{\isaliteral{5C3C414C3E}{\isasymAL}}\isaliteral{5C3C5E7375623E}{}\isactrlsub {\isadigit{1}}\ x\isaliteral{5C3C5E697375623E}{}\isactrlisub {\isadigit{1}}\ {\isaliteral{5C3C646F74733E}{\isasymdots}}\ x\isaliteral{5C3C5E697375623E}{}\isactrlisub k{\isaliteral{29}{\isacharparenright}}}} 
  \]\smallskip

  \noindent 
  in which the \isa{x}$_{i..j}$ \isa{{\isaliteral{5C3C73756273657465713E}{\isasymsubseteq}}} \isa{x}$_{1..k}$ are the
  recursive arguments of this term constructor (similarly for the other
  term-constructors). 

  Recall the lambda-calculus with \isa{Let}-patterns shown in
  \eqref{letpat}. The cases lemmas and the induction principle shown in
  \eqref{cases} and \eqref{induct} boil down in that example to the following three inference
  rules:

  \begin{equation}\label{inductex}\mbox{
  \begin{tabular}{c}
  \multicolumn{1}{@ {\hspace{-5mm}}l}{cases lemmas:}\smallskip\\
  \infer{\isa{P\isaliteral{5C3C5E627375623E}{}\isactrlbsub trm\isaliteral{5C3C5E657375623E}{}\isactrlesub }}
  {\begin{array}{@ {}l@ {}}
   \isa{{\isaliteral{5C3C666F72616C6C3E}{\isasymforall}}x{\isaliteral{2E}{\isachardot}}\ y\ {\isaliteral{3D}{\isacharequal}}\ Var\isaliteral{5C3C5E7375703E}{}\isactrlsup {\isaliteral{5C3C616C7068613E}{\isasymalpha}}\ x\ {\isaliteral{5C3C52696768746172726F773E}{\isasymRightarrow}}\ P\isaliteral{5C3C5E627375623E}{}\isactrlbsub trm\isaliteral{5C3C5E657375623E}{}\isactrlesub }\\
   \isa{{\isaliteral{5C3C666F72616C6C3E}{\isasymforall}}x\isaliteral{5C3C5E697375623E}{}\isactrlisub {\isadigit{1}}\ x\isaliteral{5C3C5E697375623E}{}\isactrlisub {\isadigit{2}}{\isaliteral{2E}{\isachardot}}\ y\ {\isaliteral{3D}{\isacharequal}}\ App\isaliteral{5C3C5E7375703E}{}\isactrlsup {\isaliteral{5C3C616C7068613E}{\isasymalpha}}\ x\isaliteral{5C3C5E697375623E}{}\isactrlisub {\isadigit{1}}\ x\isaliteral{5C3C5E697375623E}{}\isactrlisub {\isadigit{2}}\ {\isaliteral{5C3C52696768746172726F773E}{\isasymRightarrow}}\ P\isaliteral{5C3C5E627375623E}{}\isactrlbsub trm\isaliteral{5C3C5E657375623E}{}\isactrlesub }\\
   \isa{{\isaliteral{5C3C666F72616C6C3E}{\isasymforall}}x\isaliteral{5C3C5E697375623E}{}\isactrlisub {\isadigit{1}}\ x\isaliteral{5C3C5E697375623E}{}\isactrlisub {\isadigit{2}}{\isaliteral{2E}{\isachardot}}\ y\ {\isaliteral{3D}{\isacharequal}}\ Lam\isaliteral{5C3C5E7375703E}{}\isactrlsup {\isaliteral{5C3C616C7068613E}{\isasymalpha}}\ x\isaliteral{5C3C5E697375623E}{}\isactrlisub {\isadigit{1}}\ x\isaliteral{5C3C5E697375623E}{}\isactrlisub {\isadigit{2}}\ {\isaliteral{5C3C52696768746172726F773E}{\isasymRightarrow}}\ P\isaliteral{5C3C5E627375623E}{}\isactrlbsub trm\isaliteral{5C3C5E657375623E}{}\isactrlesub }\\
   \isa{{\isaliteral{5C3C666F72616C6C3E}{\isasymforall}}x\isaliteral{5C3C5E697375623E}{}\isactrlisub {\isadigit{1}}\ x\isaliteral{5C3C5E697375623E}{}\isactrlisub {\isadigit{2}}\ x\isaliteral{5C3C5E697375623E}{}\isactrlisub {\isadigit{3}}{\isaliteral{2E}{\isachardot}}\ y\ {\isaliteral{3D}{\isacharequal}}\ Let{\isaliteral{5F}{\isacharunderscore}}pat\isaliteral{5C3C5E7375703E}{}\isactrlsup {\isaliteral{5C3C616C7068613E}{\isasymalpha}}\ x\isaliteral{5C3C5E697375623E}{}\isactrlisub {\isadigit{1}}\ x\isaliteral{5C3C5E697375623E}{}\isactrlisub {\isadigit{2}}\ x\isaliteral{5C3C5E697375623E}{}\isactrlisub {\isadigit{3}}\ {\isaliteral{5C3C52696768746172726F773E}{\isasymRightarrow}}\ P\isaliteral{5C3C5E627375623E}{}\isactrlbsub trm\isaliteral{5C3C5E657375623E}{}\isactrlesub }
   \end{array}}\hspace{10mm}

  \infer{\isa{P\isaliteral{5C3C5E627375623E}{}\isactrlbsub pat\isaliteral{5C3C5E657375623E}{}\isactrlesub }}
  {\begin{array}{@ {}l@ {}}
   \isa{{\isaliteral{5C3C666F72616C6C3E}{\isasymforall}}x{\isaliteral{2E}{\isachardot}}\ y\ {\isaliteral{3D}{\isacharequal}}\ PVar\isaliteral{5C3C5E7375703E}{}\isactrlsup {\isaliteral{5C3C616C7068613E}{\isasymalpha}}\ x\ {\isaliteral{5C3C52696768746172726F773E}{\isasymRightarrow}}\ P\isaliteral{5C3C5E627375623E}{}\isactrlbsub pat\isaliteral{5C3C5E657375623E}{}\isactrlesub }\\
   \isa{{\isaliteral{5C3C666F72616C6C3E}{\isasymforall}}x\isaliteral{5C3C5E697375623E}{}\isactrlisub {\isadigit{1}}\ x\isaliteral{5C3C5E697375623E}{}\isactrlisub {\isadigit{2}}{\isaliteral{2E}{\isachardot}}\ y\ {\isaliteral{3D}{\isacharequal}}\ PTup\isaliteral{5C3C5E7375703E}{}\isactrlsup {\isaliteral{5C3C616C7068613E}{\isasymalpha}}\ x\isaliteral{5C3C5E697375623E}{}\isactrlisub {\isadigit{1}}\ x\isaliteral{5C3C5E697375623E}{}\isactrlisub {\isadigit{2}}\ {\isaliteral{5C3C52696768746172726F773E}{\isasymRightarrow}}\ P\isaliteral{5C3C5E627375623E}{}\isactrlbsub pat\isaliteral{5C3C5E657375623E}{}\isactrlesub }
  \end{array}}\medskip\\

  \multicolumn{1}{@ {\hspace{-5mm}}l}{induction principle:}\smallskip\\
  
  \infer{\isa{P\isaliteral{5C3C5E627375623E}{}\isactrlbsub trm\isaliteral{5C3C5E657375623E}{}\isactrlesub \ y\isaliteral{5C3C5E697375623E}{}\isactrlisub {\isadigit{1}}\ {\isaliteral{5C3C616E643E}{\isasymand}}\ P\isaliteral{5C3C5E627375623E}{}\isactrlbsub pat\isaliteral{5C3C5E657375623E}{}\isactrlesub \ y\isaliteral{5C3C5E697375623E}{}\isactrlisub {\isadigit{2}}}}
  {\begin{array}{@ {}l@ {}}
   \isa{{\isaliteral{5C3C666F72616C6C3E}{\isasymforall}}x{\isaliteral{2E}{\isachardot}}\ P\isaliteral{5C3C5E627375623E}{}\isactrlbsub trm\isaliteral{5C3C5E657375623E}{}\isactrlesub \ {\isaliteral{28}{\isacharparenleft}}Var\isaliteral{5C3C5E7375703E}{}\isactrlsup {\isaliteral{5C3C616C7068613E}{\isasymalpha}}\ x{\isaliteral{29}{\isacharparenright}}}\\
   \isa{{\isaliteral{5C3C666F72616C6C3E}{\isasymforall}}x\isaliteral{5C3C5E697375623E}{}\isactrlisub {\isadigit{1}}\ x\isaliteral{5C3C5E697375623E}{}\isactrlisub {\isadigit{2}}{\isaliteral{2E}{\isachardot}}\ P\isaliteral{5C3C5E627375623E}{}\isactrlbsub trm\isaliteral{5C3C5E657375623E}{}\isactrlesub \ x\isaliteral{5C3C5E697375623E}{}\isactrlisub {\isadigit{1}}\ {\isaliteral{5C3C616E643E}{\isasymand}}\ P\isaliteral{5C3C5E627375623E}{}\isactrlbsub trm\isaliteral{5C3C5E657375623E}{}\isactrlesub \ x\isaliteral{5C3C5E697375623E}{}\isactrlisub {\isadigit{2}}\ {\isaliteral{5C3C52696768746172726F773E}{\isasymRightarrow}}\ P\isaliteral{5C3C5E627375623E}{}\isactrlbsub trm\isaliteral{5C3C5E657375623E}{}\isactrlesub \ {\isaliteral{28}{\isacharparenleft}}App\isaliteral{5C3C5E7375703E}{}\isactrlsup {\isaliteral{5C3C616C7068613E}{\isasymalpha}}\ x\isaliteral{5C3C5E697375623E}{}\isactrlisub {\isadigit{1}}\ x\isaliteral{5C3C5E697375623E}{}\isactrlisub {\isadigit{2}}{\isaliteral{29}{\isacharparenright}}}\\
   \isa{{\isaliteral{5C3C666F72616C6C3E}{\isasymforall}}x\isaliteral{5C3C5E697375623E}{}\isactrlisub {\isadigit{1}}\ x\isaliteral{5C3C5E697375623E}{}\isactrlisub {\isadigit{2}}{\isaliteral{2E}{\isachardot}}\ P\isaliteral{5C3C5E627375623E}{}\isactrlbsub trm\isaliteral{5C3C5E657375623E}{}\isactrlesub \ x\isaliteral{5C3C5E697375623E}{}\isactrlisub {\isadigit{2}}\ {\isaliteral{5C3C52696768746172726F773E}{\isasymRightarrow}}\ P\isaliteral{5C3C5E627375623E}{}\isactrlbsub trm\isaliteral{5C3C5E657375623E}{}\isactrlesub \ {\isaliteral{28}{\isacharparenleft}}Lam\isaliteral{5C3C5E7375703E}{}\isactrlsup {\isaliteral{5C3C616C7068613E}{\isasymalpha}}\ x\isaliteral{5C3C5E697375623E}{}\isactrlisub {\isadigit{1}}\ x\isaliteral{5C3C5E697375623E}{}\isactrlisub {\isadigit{2}}{\isaliteral{29}{\isacharparenright}}}\\
   \isa{{\isaliteral{5C3C666F72616C6C3E}{\isasymforall}}x\isaliteral{5C3C5E697375623E}{}\isactrlisub {\isadigit{1}}\ x\isaliteral{5C3C5E697375623E}{}\isactrlisub {\isadigit{2}}\ x\isaliteral{5C3C5E697375623E}{}\isactrlisub {\isadigit{3}}{\isaliteral{2E}{\isachardot}}\ P\isaliteral{5C3C5E627375623E}{}\isactrlbsub pat\isaliteral{5C3C5E657375623E}{}\isactrlesub \ x\isaliteral{5C3C5E697375623E}{}\isactrlisub {\isadigit{1}}\ {\isaliteral{5C3C616E643E}{\isasymand}}\ P\isaliteral{5C3C5E627375623E}{}\isactrlbsub trm\isaliteral{5C3C5E657375623E}{}\isactrlesub \ x\isaliteral{5C3C5E697375623E}{}\isactrlisub {\isadigit{2}}\ {\isaliteral{5C3C616E643E}{\isasymand}}\ P\isaliteral{5C3C5E627375623E}{}\isactrlbsub trm\isaliteral{5C3C5E657375623E}{}\isactrlesub \ x\isaliteral{5C3C5E697375623E}{}\isactrlisub {\isadigit{3}}\ {\isaliteral{5C3C52696768746172726F773E}{\isasymRightarrow}}\ P\isaliteral{5C3C5E627375623E}{}\isactrlbsub trm\isaliteral{5C3C5E657375623E}{}\isactrlesub \ {\isaliteral{28}{\isacharparenleft}}Let{\isaliteral{5F}{\isacharunderscore}}pat\isaliteral{5C3C5E7375703E}{}\isactrlsup {\isaliteral{5C3C616C7068613E}{\isasymalpha}}\ x\isaliteral{5C3C5E697375623E}{}\isactrlisub {\isadigit{1}}\ x\isaliteral{5C3C5E697375623E}{}\isactrlisub {\isadigit{2}}\ x\isaliteral{5C3C5E697375623E}{}\isactrlisub {\isadigit{3}}{\isaliteral{29}{\isacharparenright}}}\\
   \isa{{\isaliteral{5C3C666F72616C6C3E}{\isasymforall}}x{\isaliteral{2E}{\isachardot}}\ P\isaliteral{5C3C5E627375623E}{}\isactrlbsub pat\isaliteral{5C3C5E657375623E}{}\isactrlesub \ {\isaliteral{28}{\isacharparenleft}}PVar\isaliteral{5C3C5E7375703E}{}\isactrlsup {\isaliteral{5C3C616C7068613E}{\isasymalpha}}\ x{\isaliteral{29}{\isacharparenright}}}\\
   \isa{{\isaliteral{5C3C666F72616C6C3E}{\isasymforall}}x\isaliteral{5C3C5E697375623E}{}\isactrlisub {\isadigit{1}}\ x\isaliteral{5C3C5E697375623E}{}\isactrlisub {\isadigit{2}}{\isaliteral{2E}{\isachardot}}\ P\isaliteral{5C3C5E627375623E}{}\isactrlbsub pat\isaliteral{5C3C5E657375623E}{}\isactrlesub \ x\isaliteral{5C3C5E697375623E}{}\isactrlisub {\isadigit{1}}\ {\isaliteral{5C3C616E643E}{\isasymand}}\ P\isaliteral{5C3C5E627375623E}{}\isactrlbsub pat\isaliteral{5C3C5E657375623E}{}\isactrlesub \ x\isaliteral{5C3C5E697375623E}{}\isactrlisub {\isadigit{2}}\ {\isaliteral{5C3C52696768746172726F773E}{\isasymRightarrow}}\ P\isaliteral{5C3C5E627375623E}{}\isactrlbsub pat\isaliteral{5C3C5E657375623E}{}\isactrlesub \ {\isaliteral{28}{\isacharparenleft}}PTup\isaliteral{5C3C5E7375703E}{}\isactrlsup {\isaliteral{5C3C616C7068613E}{\isasymalpha}}\ x\isaliteral{5C3C5E697375623E}{}\isactrlisub {\isadigit{1}}\ x\isaliteral{5C3C5E697375623E}{}\isactrlisub {\isadigit{2}}{\isaliteral{29}{\isacharparenright}}}
  \end{array}}
  \end{tabular}}
  \end{equation}\smallskip

  By working now completely on the alpha-equated level, we
  can first show using \eqref{calphaeqvt} and Property~\ref{swapfreshfresh} that the support of each term
  constructor is included in the support of its arguments, 
  namely

  \[
  \isa{{\isaliteral{28}{\isacharparenleft}}supp\ x\isaliteral{5C3C5E697375623E}{}\isactrlisub {\isadigit{1}}\ {\isaliteral{5C3C756E696F6E3E}{\isasymunion}}\ {\isaliteral{5C3C646F74733E}{\isasymdots}}\ {\isaliteral{5C3C756E696F6E3E}{\isasymunion}}\ supp\ x\isaliteral{5C3C5E697375623E}{}\isactrlisub r{\isaliteral{29}{\isacharparenright}}\ supports\ {\isaliteral{28}{\isacharparenleft}}C\isaliteral{5C3C5E7375703E}{}\isactrlsup {\isaliteral{5C3C616C7068613E}{\isasymalpha}}\ x\isaliteral{5C3C5E697375623E}{}\isactrlisub {\isadigit{1}}\ {\isaliteral{5C3C646F74733E}{\isasymdots}}\ x\isaliteral{5C3C5E697375623E}{}\isactrlisub r{\isaliteral{29}{\isacharparenright}}}
  \]\smallskip

  \noindent
  This allows us to prove using the induction principle for  \isa{ty{\isaliteral{5C3C414C3E}{\isasymAL}}}$_{1..n}$ 
  that every element of type \isa{ty{\isaliteral{5C3C414C3E}{\isasymAL}}}$_{1..n}$ is finitely supported 
  (using Proposition~\ref{supportsprop}{\it (i)}). 
  Similarly, we can establish by induction that the free-atom functions and binding 
  functions are equivariant, namely
  
  \[\mbox{
  \begin{tabular}{rcl}
  \isa{{\isaliteral{5C3C70693E}{\isasympi}}\ {\isaliteral{5C3C62756C6C65743E}{\isasymbullet}}\ {\isaliteral{28}{\isacharparenleft}}fa{\isaliteral{5F}{\isacharunderscore}}ty{\isaliteral{5C3C414C3E}{\isasymAL}}\isaliteral{5C3C5E697375623E}{}\isactrlisub i\ \ x{\isaliteral{29}{\isacharparenright}}} & $=$ & \isa{fa{\isaliteral{5F}{\isacharunderscore}}ty{\isaliteral{5C3C414C3E}{\isasymAL}}\isaliteral{5C3C5E697375623E}{}\isactrlisub i\ {\isaliteral{28}{\isacharparenleft}}{\isaliteral{5C3C70693E}{\isasympi}}\ {\isaliteral{5C3C62756C6C65743E}{\isasymbullet}}\ x{\isaliteral{29}{\isacharparenright}}}\\
  \isa{{\isaliteral{5C3C70693E}{\isasympi}}\ {\isaliteral{5C3C62756C6C65743E}{\isasymbullet}}\ {\isaliteral{28}{\isacharparenleft}}fa{\isaliteral{5F}{\isacharunderscore}}bn{\isaliteral{5C3C414C3E}{\isasymAL}}\isaliteral{5C3C5E697375623E}{}\isactrlisub j\ \ x{\isaliteral{29}{\isacharparenright}}} & $=$ & \isa{fa{\isaliteral{5F}{\isacharunderscore}}bn{\isaliteral{5C3C414C3E}{\isasymAL}}\isaliteral{5C3C5E697375623E}{}\isactrlisub j\ {\isaliteral{28}{\isacharparenleft}}{\isaliteral{5C3C70693E}{\isasympi}}\ {\isaliteral{5C3C62756C6C65743E}{\isasymbullet}}\ x{\isaliteral{29}{\isacharparenright}}}\\
  \isa{{\isaliteral{5C3C70693E}{\isasympi}}\ {\isaliteral{5C3C62756C6C65743E}{\isasymbullet}}\ {\isaliteral{28}{\isacharparenleft}}bn{\isaliteral{5C3C414C3E}{\isasymAL}}\isaliteral{5C3C5E697375623E}{}\isactrlisub j\ \ x{\isaliteral{29}{\isacharparenright}}}    & $=$ & \isa{bn{\isaliteral{5C3C414C3E}{\isasymAL}}\isaliteral{5C3C5E697375623E}{}\isactrlisub j\ {\isaliteral{28}{\isacharparenleft}}{\isaliteral{5C3C70693E}{\isasympi}}\ {\isaliteral{5C3C62756C6C65743E}{\isasymbullet}}\ x{\isaliteral{29}{\isacharparenright}}}\\
  \end{tabular}}
  \]\smallskip

  \noindent
  Lastly, we can show that the support of elements in \isa{ty{\isaliteral{5C3C414C3E}{\isasymAL}}}$_{1..n}$ is the same as the free-atom functions \isa{fa{\isaliteral{5F}{\isacharunderscore}}ty{\isaliteral{5C3C414C3E}{\isasymAL}}}$_{1..n}$.  This fact is important in the nominal setting where
  the general theory is formulated in terms of support and freshness, but also
  provides evidence that our notions of free-atoms and alpha-equivalence
  `match up' correctly.

  \begin{thm}\label{suppfa} 
  For \isa{x}$_{1..n}$ with type \isa{ty{\isaliteral{5C3C414C3E}{\isasymAL}}}$_{1..n}$, we have
  \isa{supp\ x\isaliteral{5C3C5E697375623E}{}\isactrlisub i\ {\isaliteral{3D}{\isacharequal}}\ fa{\isaliteral{5F}{\isacharunderscore}}ty{\isaliteral{5C3C414C3E}{\isasymAL}}\isaliteral{5C3C5E697375623E}{}\isactrlisub i\ x\isaliteral{5C3C5E697375623E}{}\isactrlisub i}.
  \end{thm}

  \begin{proof}
  The proof is by induction on \isa{x}$_{1..n}$. In each case
  we unfold the definition of \isa{supp}, move the swapping inside the 
  term-constructors and then use the quasi-injectivity lemmas in order to complete the
  proof. For the abstraction cases we use then the facts derived in Theorem~\ref{suppabs},
  for which we have to know that every body of an abstraction is finitely supported.
  This, we have proved earlier.
  \end{proof}

  \noindent
  Consequently, we can replace the free-atom functions by \isa{supp} in  
  our quasi-injection lemmas. In the examples shown in \eqref{alphalift}, for instance,
  we obtain for \isa{Let\isaliteral{5C3C5E7375703E}{}\isactrlsup {\isaliteral{5C3C616C7068613E}{\isasymalpha}}} and \isa{Let{\isaliteral{5F}{\isacharunderscore}}rec\isaliteral{5C3C5E7375703E}{}\isactrlsup {\isaliteral{5C3C616C7068613E}{\isasymalpha}}} 

  \[\mbox{
  \begin{tabular}{@ {}c @ {}}
  \infer{\isa{Let\isaliteral{5C3C5E7375703E}{}\isactrlsup {\isaliteral{5C3C616C7068613E}{\isasymalpha}}\ as\ t\ {\isaliteral{3D}{\isacharequal}}\ Let\isaliteral{5C3C5E7375703E}{}\isactrlsup {\isaliteral{5C3C616C7068613E}{\isasymalpha}}\ as{\isaliteral{27}{\isacharprime}}\ t{\isaliteral{27}{\isacharprime}}}}
  {\isa{{\isaliteral{28}{\isacharparenleft}}bn\isaliteral{5C3C5E7375703E}{}\isactrlsup {\isaliteral{5C3C616C7068613E}{\isasymalpha}}\ as{\isaliteral{2C}{\isacharcomma}}\ t{\isaliteral{29}{\isacharparenright}}\ {\isaliteral{5C3C617070726F783E}{\isasymapprox}}\,\raisebox{-1pt}{\makebox[0mm][l]{$_{\textit{list}}$}}\isaliteral{5C3C5E627375703E}{}\isactrlbsup {\isaliteral{3D}{\isacharequal}}{\isaliteral{2C}{\isacharcomma}}\ supp\isaliteral{5C3C5E657375703E}{}\isactrlesup \ {\isaliteral{28}{\isacharparenleft}}bn\isaliteral{5C3C5E7375703E}{}\isactrlsup {\isaliteral{5C3C616C7068613E}{\isasymalpha}}\ as{\isaliteral{27}{\isacharprime}}{\isaliteral{2C}{\isacharcomma}}\ t{\isaliteral{27}{\isacharprime}}{\isaliteral{29}{\isacharparenright}}} & 
  \hspace{5mm}\isa{as\ {\isaliteral{5C3C617070726F783E}{\isasymapprox}}{\isaliteral{5C3C414C3E}{\isasymAL}}\isaliteral{5C3C5E627375623E}{}\isactrlbsub bn\isaliteral{5C3C5E657375623E}{}\isactrlesub \ as{\isaliteral{27}{\isacharprime}}}}\\
  \\
  \makebox[0mm]{\infer{\isa{Let{\isaliteral{5F}{\isacharunderscore}}rec\isaliteral{5C3C5E7375703E}{}\isactrlsup {\isaliteral{5C3C616C7068613E}{\isasymalpha}}\ as\ t\ {\isaliteral{3D}{\isacharequal}}\ Let{\isaliteral{5F}{\isacharunderscore}}rec\isaliteral{5C3C5E7375703E}{}\isactrlsup {\isaliteral{5C3C616C7068613E}{\isasymalpha}}\ as{\isaliteral{27}{\isacharprime}}\ t{\isaliteral{27}{\isacharprime}}}}
  {\isa{{\isaliteral{28}{\isacharparenleft}}bn\isaliteral{5C3C5E7375703E}{}\isactrlsup {\isaliteral{5C3C616C7068613E}{\isasymalpha}}\ as{\isaliteral{2C}{\isacharcomma}}\ {\isaliteral{28}{\isacharparenleft}}as{\isaliteral{2C}{\isacharcomma}}\ t{\isaliteral{29}{\isacharparenright}}{\isaliteral{29}{\isacharparenright}}\ {\isaliteral{5C3C617070726F783E}{\isasymapprox}}\,\raisebox{-1pt}{\makebox[0mm][l]{$_{\textit{list}}$}}\isaliteral{5C3C5E627375703E}{}\isactrlbsup {\isaliteral{28}{\isacharparenleft}}{\isaliteral{3D}{\isacharequal}}{\isaliteral{2C}{\isacharcomma}}\ {\isaliteral{3D}{\isacharequal}}{\isaliteral{29}{\isacharparenright}}{\isaliteral{2C}{\isacharcomma}}\ {\isaliteral{28}{\isacharparenleft}}supp{\isaliteral{2C}{\isacharcomma}}\ supp{\isaliteral{29}{\isacharparenright}}\isaliteral{5C3C5E657375703E}{}\isactrlesup \ {\isaliteral{28}{\isacharparenleft}}bn\isaliteral{5C3C5E7375703E}{}\isactrlsup {\isaliteral{5C3C616C7068613E}{\isasymalpha}}\ as{\isaliteral{27}{\isacharprime}}{\isaliteral{2C}{\isacharcomma}}\ {\isaliteral{28}{\isacharparenleft}}as{\isaliteral{2C}{\isacharcomma}}\ t{\isaliteral{5C3C5052494D453E}{\isasymPRIME}}\ {\isaliteral{29}{\isacharparenright}}{\isaliteral{29}{\isacharparenright}}}}}\\
  \end{tabular}}
  \]\smallskip

  \noindent
  Taking into account that the compound equivalence relation \isa{{\isaliteral{28}{\isacharparenleft}}{\isaliteral{3D}{\isacharequal}}{\isaliteral{2C}{\isacharcomma}}\ {\isaliteral{3D}{\isacharequal}}{\isaliteral{29}{\isacharparenright}}} and the compound free-atom function \isa{{\isaliteral{28}{\isacharparenleft}}supp{\isaliteral{2C}{\isacharcomma}}\ supp{\isaliteral{29}{\isacharparenright}}} are by
  definition equal to \isa{{\isaliteral{3D}{\isacharequal}}} and \isa{supp}, respectively, the
  above rules simplify further to

  \[\mbox{
  \begin{tabular}{@ {}c @ {}}
  \infer{\isa{Let\isaliteral{5C3C5E7375703E}{}\isactrlsup {\isaliteral{5C3C616C7068613E}{\isasymalpha}}\ as\ t\ {\isaliteral{3D}{\isacharequal}}\ Let\isaliteral{5C3C5E7375703E}{}\isactrlsup {\isaliteral{5C3C616C7068613E}{\isasymalpha}}\ as{\isaliteral{27}{\isacharprime}}\ t{\isaliteral{27}{\isacharprime}}}}
  {\isa{{\isaliteral{5B}{\isacharbrackleft}}bn\isaliteral{5C3C5E7375703E}{}\isactrlsup {\isaliteral{5C3C616C7068613E}{\isasymalpha}}\ as{\isaliteral{5D}{\isacharbrackright}}\isaliteral{5C3C5E627375623E}{}\isactrlbsub list\isaliteral{5C3C5E657375623E}{}\isactrlesub {\isaliteral{2E}{\isachardot}}t\ {\isaliteral{3D}{\isacharequal}}\ {\isaliteral{5B}{\isacharbrackleft}}bn\isaliteral{5C3C5E7375703E}{}\isactrlsup {\isaliteral{5C3C616C7068613E}{\isasymalpha}}\ as{\isaliteral{27}{\isacharprime}}{\isaliteral{5D}{\isacharbrackright}}\isaliteral{5C3C5E627375623E}{}\isactrlbsub list\isaliteral{5C3C5E657375623E}{}\isactrlesub {\isaliteral{2E}{\isachardot}}t{\isaliteral{27}{\isacharprime}}} & 
  \hspace{5mm}\isa{as\ {\isaliteral{5C3C617070726F783E}{\isasymapprox}}{\isaliteral{5C3C414C3E}{\isasymAL}}\isaliteral{5C3C5E627375623E}{}\isactrlbsub bn\isaliteral{5C3C5E657375623E}{}\isactrlesub \ as{\isaliteral{27}{\isacharprime}}}}\\
  \\
  \makebox[0mm]{\infer{\isa{Let{\isaliteral{5F}{\isacharunderscore}}rec\isaliteral{5C3C5E7375703E}{}\isactrlsup {\isaliteral{5C3C616C7068613E}{\isasymalpha}}\ as\ t\ {\isaliteral{3D}{\isacharequal}}\ Let{\isaliteral{5F}{\isacharunderscore}}rec\isaliteral{5C3C5E7375703E}{}\isactrlsup {\isaliteral{5C3C616C7068613E}{\isasymalpha}}\ as{\isaliteral{27}{\isacharprime}}\ t{\isaliteral{27}{\isacharprime}}}}
  {\isa{{\isaliteral{5B}{\isacharbrackleft}}bn\isaliteral{5C3C5E7375703E}{}\isactrlsup {\isaliteral{5C3C616C7068613E}{\isasymalpha}}\ as{\isaliteral{5D}{\isacharbrackright}}\isaliteral{5C3C5E627375623E}{}\isactrlbsub list\isaliteral{5C3C5E657375623E}{}\isactrlesub {\isaliteral{2E}{\isachardot}}{\isaliteral{28}{\isacharparenleft}}as{\isaliteral{2C}{\isacharcomma}}\ t{\isaliteral{29}{\isacharparenright}}\ {\isaliteral{3D}{\isacharequal}}\ {\isaliteral{5B}{\isacharbrackleft}}bn\isaliteral{5C3C5E7375703E}{}\isactrlsup {\isaliteral{5C3C616C7068613E}{\isasymalpha}}\ as{\isaliteral{27}{\isacharprime}}{\isaliteral{5D}{\isacharbrackright}}\isaliteral{5C3C5E627375623E}{}\isactrlbsub list\isaliteral{5C3C5E657375623E}{}\isactrlesub {\isaliteral{2E}{\isachardot}}{\isaliteral{28}{\isacharparenleft}}as{\isaliteral{2C}{\isacharcomma}}\ t{\isaliteral{5C3C5052494D453E}{\isasymPRIME}}\ {\isaliteral{29}{\isacharparenright}}}}}\\
  \end{tabular}}
  \]\smallskip

  \noindent
  which means we can characterise equality between term-constructors (on the
  alpha-equated level) in terms of equality between the abstractions defined
  in Section~\ref{sec:binders}. From this we can deduce the support for \isa{Let\isaliteral{5C3C5E7375703E}{}\isactrlsup {\isaliteral{5C3C616C7068613E}{\isasymalpha}}} and \isa{Let{\isaliteral{5F}{\isacharunderscore}}rec\isaliteral{5C3C5E7375703E}{}\isactrlsup {\isaliteral{5C3C616C7068613E}{\isasymalpha}}}, namely

  \[\mbox{
  \begin{tabular}{l@ {\hspace{2mm}}l@ {\hspace{2mm}}l}
  \isa{supp\ {\isaliteral{28}{\isacharparenleft}}Let\isaliteral{5C3C5E7375703E}{}\isactrlsup {\isaliteral{5C3C616C7068613E}{\isasymalpha}}\ as\ t{\isaliteral{29}{\isacharparenright}}} & \isa{{\isaliteral{3D}{\isacharequal}}} & \isa{{\isaliteral{28}{\isacharparenleft}}supp\ t\ {\isaliteral{2D}{\isacharminus}}\ set\ {\isaliteral{28}{\isacharparenleft}}bn\isaliteral{5C3C5E7375703E}{}\isactrlsup {\isaliteral{5C3C616C7068613E}{\isasymalpha}}\ as{\isaliteral{29}{\isacharparenright}}{\isaliteral{29}{\isacharparenright}}\ {\isaliteral{5C3C756E696F6E3E}{\isasymunion}}\ fa{\isaliteral{5C3C414C3E}{\isasymAL}}\isaliteral{5C3C5E627375623E}{}\isactrlbsub bn\isaliteral{5C3C5E657375623E}{}\isactrlesub \ as}\\
  \isa{supp\ {\isaliteral{28}{\isacharparenleft}}Let{\isaliteral{5F}{\isacharunderscore}}rec\isaliteral{5C3C5E7375703E}{}\isactrlsup {\isaliteral{5C3C616C7068613E}{\isasymalpha}}\ as\ t{\isaliteral{29}{\isacharparenright}}} & \isa{{\isaliteral{3D}{\isacharequal}}} & \isa{{\isaliteral{28}{\isacharparenleft}}supp\ t\ {\isaliteral{5C3C756E696F6E3E}{\isasymunion}}\ supp\ as{\isaliteral{29}{\isacharparenright}}\ {\isaliteral{2D}{\isacharminus}}\ set\ {\isaliteral{28}{\isacharparenleft}}bn\isaliteral{5C3C5E7375703E}{}\isactrlsup {\isaliteral{5C3C616C7068613E}{\isasymalpha}}\ as{\isaliteral{29}{\isacharparenright}}}\\
  \end{tabular}}
  \]\smallskip

  \noindent
  using the support of abstractions derived in Theorem~\ref{suppabs}.

  To sum up this section, we have established a reasoning infrastructure for the
  types \isa{ty{\isaliteral{5C3C414C3E}{\isasymAL}}}$_{1..n}$ by first lifting definitions from the
  `raw' level to the quotient level and then by proving facts about
  these lifted definitions. All necessary proofs are generated automatically
  by custom ML-code.%
\end{isamarkuptext}%
\isamarkuptrue%
\isamarkupsection{Strong Induction Principles%
}
\isamarkuptrue%
\begin{isamarkuptext}%
In the previous section we derived induction principles for alpha-equated
  terms (see \eqref{induct} for the general form and \eqref{inductex} for an
  example). This was done by lifting the corresponding inductions principles
  for `raw' terms.  We already employed these induction principles for
  deriving several facts about alpha-equated terms, including the property that
  the free-atom functions and the notion of support coincide. Still, we
  call these induction principles \emph{weak}, because for a term-constructor,
  say \mbox{\isa{C\isaliteral{5C3C5E7375703E}{}\isactrlsup {\isaliteral{5C3C616C7068613E}{\isasymalpha}}\ x\isaliteral{5C3C5E697375623E}{}\isactrlisub {\isadigit{1}}{\isaliteral{5C3C646F74733E}{\isasymdots}}x\isaliteral{5C3C5E697375623E}{}\isactrlisub r}}, the induction
  hypothesis requires us to establish (under some assumptions) a property
  \isa{P\ {\isaliteral{28}{\isacharparenleft}}C\isaliteral{5C3C5E7375703E}{}\isactrlsup {\isaliteral{5C3C616C7068613E}{\isasymalpha}}\ x\isaliteral{5C3C5E697375623E}{}\isactrlisub {\isadigit{1}}{\isaliteral{5C3C646F74733E}{\isasymdots}}x\isaliteral{5C3C5E697375623E}{}\isactrlisub r{\isaliteral{29}{\isacharparenright}}} for \emph{all} \isa{x}$_{1..r}$. The problem with this is that in the presence of binders we cannot make
  any assumptions about the atoms that are bound---for example assuming the variable convention. 
  One obvious way around this
  problem is to rename bound atoms. Unfortunately, this leads to very clunky proofs
  and makes formalisations grievous experiences (especially in the context of 
  multiple bound atoms).

  For the older versions of Nominal Isabelle we described in \cite{Urban08} a
  method for automatically strengthening weak induction principles. These
  stronger induction principles allow the user to make additional assumptions
  about bound atoms. The advantage of these assumptions is that they make in
  most cases any renaming of bound atoms unnecessary.  To explain how the
  strengthening works, we use as running example the lambda-calculus with
  \isa{Let}-patterns shown in \eqref{letpat}. Its weak induction principle
  is given in \eqref{inductex}.  The stronger induction principle is as
  follows:

  \begin{equation}\label{stronginduct}
  \mbox{
  \begin{tabular}{@ {}c@ {}}
  \infer{\isa{P\isaliteral{5C3C5E627375623E}{}\isactrlbsub trm\isaliteral{5C3C5E657375623E}{}\isactrlesub \ c\ y\isaliteral{5C3C5E697375623E}{}\isactrlisub {\isadigit{1}}\ {\isaliteral{5C3C616E643E}{\isasymand}}\ P\isaliteral{5C3C5E627375623E}{}\isactrlbsub pat\isaliteral{5C3C5E657375623E}{}\isactrlesub \ c\ y\isaliteral{5C3C5E697375623E}{}\isactrlisub {\isadigit{2}}}}
  {\begin{array}{l}
   \isa{{\isaliteral{5C3C666F72616C6C3E}{\isasymforall}}x\ c{\isaliteral{2E}{\isachardot}}\ P\isaliteral{5C3C5E627375623E}{}\isactrlbsub trm\isaliteral{5C3C5E657375623E}{}\isactrlesub \ c\ {\isaliteral{28}{\isacharparenleft}}Var\isaliteral{5C3C5E7375703E}{}\isactrlsup {\isaliteral{5C3C616C7068613E}{\isasymalpha}}\ x{\isaliteral{29}{\isacharparenright}}}\\
   \isa{{\isaliteral{5C3C666F72616C6C3E}{\isasymforall}}x\isaliteral{5C3C5E697375623E}{}\isactrlisub {\isadigit{1}}\ x\isaliteral{5C3C5E697375623E}{}\isactrlisub {\isadigit{2}}\ c{\isaliteral{2E}{\isachardot}}\ {\isaliteral{28}{\isacharparenleft}}{\isaliteral{5C3C666F72616C6C3E}{\isasymforall}}d{\isaliteral{2E}{\isachardot}}\ P\isaliteral{5C3C5E627375623E}{}\isactrlbsub trm\isaliteral{5C3C5E657375623E}{}\isactrlesub \ d\ x\isaliteral{5C3C5E697375623E}{}\isactrlisub {\isadigit{1}}{\isaliteral{29}{\isacharparenright}}\ {\isaliteral{5C3C616E643E}{\isasymand}}\ {\isaliteral{28}{\isacharparenleft}}{\isaliteral{5C3C666F72616C6C3E}{\isasymforall}}d{\isaliteral{2E}{\isachardot}}\ P\isaliteral{5C3C5E627375623E}{}\isactrlbsub trm\isaliteral{5C3C5E657375623E}{}\isactrlesub \ d\ x\isaliteral{5C3C5E697375623E}{}\isactrlisub {\isadigit{2}}{\isaliteral{29}{\isacharparenright}}\ {\isaliteral{5C3C52696768746172726F773E}{\isasymRightarrow}}\ P\isaliteral{5C3C5E627375623E}{}\isactrlbsub trm\isaliteral{5C3C5E657375623E}{}\isactrlesub \ c\ {\isaliteral{28}{\isacharparenleft}}App\isaliteral{5C3C5E7375703E}{}\isactrlsup {\isaliteral{5C3C616C7068613E}{\isasymalpha}}\ x\isaliteral{5C3C5E697375623E}{}\isactrlisub {\isadigit{1}}\ x\isaliteral{5C3C5E697375623E}{}\isactrlisub {\isadigit{2}}{\isaliteral{29}{\isacharparenright}}}\\
   \isa{{\isaliteral{5C3C666F72616C6C3E}{\isasymforall}}x\isaliteral{5C3C5E697375623E}{}\isactrlisub {\isadigit{1}}\ x\isaliteral{5C3C5E697375623E}{}\isactrlisub {\isadigit{2}}\ c{\isaliteral{2E}{\isachardot}}\ atom\ x\isaliteral{5C3C5E697375623E}{}\isactrlisub {\isadigit{1}}\ {\isaliteral{23}{\isacharhash}}\ c\ {\isaliteral{5C3C616E643E}{\isasymand}}\ {\isaliteral{28}{\isacharparenleft}}{\isaliteral{5C3C666F72616C6C3E}{\isasymforall}}d{\isaliteral{2E}{\isachardot}}\ P\isaliteral{5C3C5E627375623E}{}\isactrlbsub trm\isaliteral{5C3C5E657375623E}{}\isactrlesub \ d\ x\isaliteral{5C3C5E697375623E}{}\isactrlisub {\isadigit{2}}{\isaliteral{29}{\isacharparenright}}\ {\isaliteral{5C3C52696768746172726F773E}{\isasymRightarrow}}\ P\isaliteral{5C3C5E627375623E}{}\isactrlbsub trm\isaliteral{5C3C5E657375623E}{}\isactrlesub \ c\ {\isaliteral{28}{\isacharparenleft}}Lam\isaliteral{5C3C5E7375703E}{}\isactrlsup {\isaliteral{5C3C616C7068613E}{\isasymalpha}}\ x\isaliteral{5C3C5E697375623E}{}\isactrlisub {\isadigit{1}}\ x\isaliteral{5C3C5E697375623E}{}\isactrlisub {\isadigit{2}}{\isaliteral{29}{\isacharparenright}}}\\
   \isa{{\isaliteral{5C3C666F72616C6C3E}{\isasymforall}}x\isaliteral{5C3C5E697375623E}{}\isactrlisub {\isadigit{1}}\ x\isaliteral{5C3C5E697375623E}{}\isactrlisub {\isadigit{2}}\ x\isaliteral{5C3C5E697375623E}{}\isactrlisub {\isadigit{3}}\ c{\isaliteral{2E}{\isachardot}}\ {\isaliteral{28}{\isacharparenleft}}set\ {\isaliteral{28}{\isacharparenleft}}bn\isaliteral{5C3C5E7375703E}{}\isactrlsup {\isaliteral{5C3C616C7068613E}{\isasymalpha}}\ x\isaliteral{5C3C5E697375623E}{}\isactrlisub {\isadigit{1}}{\isaliteral{29}{\isacharparenright}}{\isaliteral{29}{\isacharparenright}}\ {\isaliteral{23}{\isacharhash}}\isaliteral{5C3C5E7375703E}{}\isactrlsup {\isaliteral{2A}{\isacharasterisk}}\ c\ {\isaliteral{5C3C616E643E}{\isasymand}}}\\ 
   \hspace{10mm}\isa{{\isaliteral{28}{\isacharparenleft}}{\isaliteral{5C3C666F72616C6C3E}{\isasymforall}}d{\isaliteral{2E}{\isachardot}}\ P\isaliteral{5C3C5E627375623E}{}\isactrlbsub pat\isaliteral{5C3C5E657375623E}{}\isactrlesub \ d\ x\isaliteral{5C3C5E697375623E}{}\isactrlisub {\isadigit{1}}{\isaliteral{29}{\isacharparenright}}\ {\isaliteral{5C3C616E643E}{\isasymand}}\ {\isaliteral{28}{\isacharparenleft}}{\isaliteral{5C3C666F72616C6C3E}{\isasymforall}}d{\isaliteral{2E}{\isachardot}}\ P\isaliteral{5C3C5E627375623E}{}\isactrlbsub trm\isaliteral{5C3C5E657375623E}{}\isactrlesub \ d\ x\isaliteral{5C3C5E697375623E}{}\isactrlisub {\isadigit{2}}{\isaliteral{29}{\isacharparenright}}\ {\isaliteral{5C3C616E643E}{\isasymand}}\ {\isaliteral{28}{\isacharparenleft}}{\isaliteral{5C3C666F72616C6C3E}{\isasymforall}}d{\isaliteral{2E}{\isachardot}}\ P\isaliteral{5C3C5E627375623E}{}\isactrlbsub trm\isaliteral{5C3C5E657375623E}{}\isactrlesub \ d\ x\isaliteral{5C3C5E697375623E}{}\isactrlisub {\isadigit{3}}{\isaliteral{29}{\isacharparenright}}\ {\isaliteral{5C3C52696768746172726F773E}{\isasymRightarrow}}\ P\isaliteral{5C3C5E627375623E}{}\isactrlbsub trm\isaliteral{5C3C5E657375623E}{}\isactrlesub \ c\ {\isaliteral{28}{\isacharparenleft}}Let{\isaliteral{5F}{\isacharunderscore}}pat\isaliteral{5C3C5E7375703E}{}\isactrlsup {\isaliteral{5C3C616C7068613E}{\isasymalpha}}\ x\isaliteral{5C3C5E697375623E}{}\isactrlisub {\isadigit{1}}\ x\isaliteral{5C3C5E697375623E}{}\isactrlisub {\isadigit{2}}\ x\isaliteral{5C3C5E697375623E}{}\isactrlisub {\isadigit{3}}{\isaliteral{29}{\isacharparenright}}}\\
   \isa{{\isaliteral{5C3C666F72616C6C3E}{\isasymforall}}x\ c{\isaliteral{2E}{\isachardot}}\ P\isaliteral{5C3C5E627375623E}{}\isactrlbsub pat\isaliteral{5C3C5E657375623E}{}\isactrlesub \ c\ {\isaliteral{28}{\isacharparenleft}}PVar\isaliteral{5C3C5E7375703E}{}\isactrlsup {\isaliteral{5C3C616C7068613E}{\isasymalpha}}\ x{\isaliteral{29}{\isacharparenright}}}\\
   \isa{{\isaliteral{5C3C666F72616C6C3E}{\isasymforall}}x\isaliteral{5C3C5E697375623E}{}\isactrlisub {\isadigit{1}}\ x\isaliteral{5C3C5E697375623E}{}\isactrlisub {\isadigit{2}}\ c{\isaliteral{2E}{\isachardot}}\ {\isaliteral{28}{\isacharparenleft}}{\isaliteral{5C3C666F72616C6C3E}{\isasymforall}}d{\isaliteral{2E}{\isachardot}}\ P\isaliteral{5C3C5E627375623E}{}\isactrlbsub pat\isaliteral{5C3C5E657375623E}{}\isactrlesub \ d\ x\isaliteral{5C3C5E697375623E}{}\isactrlisub {\isadigit{1}}{\isaliteral{29}{\isacharparenright}}\ {\isaliteral{5C3C616E643E}{\isasymand}}\ {\isaliteral{28}{\isacharparenleft}}{\isaliteral{5C3C666F72616C6C3E}{\isasymforall}}d{\isaliteral{2E}{\isachardot}}\ P\isaliteral{5C3C5E627375623E}{}\isactrlbsub pat\isaliteral{5C3C5E657375623E}{}\isactrlesub \ d\ x\isaliteral{5C3C5E697375623E}{}\isactrlisub {\isadigit{2}}{\isaliteral{29}{\isacharparenright}}\ {\isaliteral{5C3C52696768746172726F773E}{\isasymRightarrow}}\ P\isaliteral{5C3C5E627375623E}{}\isactrlbsub pat\isaliteral{5C3C5E657375623E}{}\isactrlesub \ c\ {\isaliteral{28}{\isacharparenleft}}PTup\isaliteral{5C3C5E7375703E}{}\isactrlsup {\isaliteral{5C3C616C7068613E}{\isasymalpha}}\ x\isaliteral{5C3C5E697375623E}{}\isactrlisub {\isadigit{1}}\ x\isaliteral{5C3C5E697375623E}{}\isactrlisub {\isadigit{2}}{\isaliteral{29}{\isacharparenright}}}
  \end{array}}
  \end{tabular}}
  \end{equation}\smallskip

  \noindent
  Notice that instead of establishing two properties of the form \isa{\ P\isaliteral{5C3C5E627375623E}{}\isactrlbsub trm\isaliteral{5C3C5E657375623E}{}\isactrlesub \ y\isaliteral{5C3C5E697375623E}{}\isactrlisub {\isadigit{1}}\ {\isaliteral{5C3C616E643E}{\isasymand}}\ P\isaliteral{5C3C5E627375623E}{}\isactrlbsub pat\isaliteral{5C3C5E657375623E}{}\isactrlesub \ y\isaliteral{5C3C5E697375623E}{}\isactrlisub {\isadigit{2}}}, as the
  weak one does, the stronger induction principle establishes the properties
  of the form \isa{P\isaliteral{5C3C5E627375623E}{}\isactrlbsub trm\isaliteral{5C3C5E657375623E}{}\isactrlesub \ c\ y\isaliteral{5C3C5E697375623E}{}\isactrlisub {\isadigit{1}}\ {\isaliteral{5C3C616E643E}{\isasymand}}\ P\isaliteral{5C3C5E627375623E}{}\isactrlbsub pat\isaliteral{5C3C5E657375623E}{}\isactrlesub \ c\ y\isaliteral{5C3C5E697375623E}{}\isactrlisub {\isadigit{2}}} in which the additional parameter \isa{c} is assumed to be of finite support. The purpose of \isa{c} is to
  `control' which freshness assumptions the binders should satisfy in the
  \isa{Lam\isaliteral{5C3C5E7375703E}{}\isactrlsup {\isaliteral{5C3C616C7068613E}{\isasymalpha}}} and \isa{Let{\isaliteral{5F}{\isacharunderscore}}pat\isaliteral{5C3C5E7375703E}{}\isactrlsup {\isaliteral{5C3C616C7068613E}{\isasymalpha}}} cases: for \isa{Lam\isaliteral{5C3C5E7375703E}{}\isactrlsup {\isaliteral{5C3C616C7068613E}{\isasymalpha}}} we can assume the bound atom \isa{x\isaliteral{5C3C5E697375623E}{}\isactrlisub {\isadigit{1}}} is fresh
  for \isa{c} (third line); for \isa{Let{\isaliteral{5F}{\isacharunderscore}}pat\isaliteral{5C3C5E7375703E}{}\isactrlsup {\isaliteral{5C3C616C7068613E}{\isasymalpha}}} we can assume
  all bound atoms from an assignment are fresh for \isa{c} (fourth
  line). In order to see how an instantiation for \isa{c} in the
  conclusion `controls' the premises, one has to take into account that
  Isabelle/HOL is a typed logic. That means if \isa{c} is instantiated
  with, for example, a pair, then this type-constraint will be propagated to
  the premises. The main point is that if \isa{c} is instantiated
  appropriately, then the user can mimic the usual convenient `pencil-and-paper'
  reasoning employing the variable convention about bound and free variables
  being distinct \cite{Urban08}.

  In what follows we will show that the weak induction principle in
  \eqref{inductex} implies the strong one \eqref{stronginduct}. This fact was established for
  single binders in \cite{Urban08} by some quite involved, nevertheless
  automated, induction proof. In this paper we simplify the proof by
  leveraging the automated proving tools from the function package of
  Isabelle/HOL \cite{Krauss09}. The reasoning principle behind these tools
  is well-founded induction. To use them in our setting, we have to discharge
  two proof obligations: one is that we have well-founded measures (one for
  each type \isa{ty}$^\alpha_{1..n}$) that decrease in every induction
  step and the other is that we have covered all cases in the induction
  principle. Once these two proof obligations are discharged, the reasoning
  infrastructure of the function package will automatically derive the
  stronger induction principle. This way of establishing the stronger induction
  principle is considerably simpler than the earlier work presented in \cite{Urban08}.

  As measures we can use the size functions \isa{size{\isaliteral{5F}{\isacharunderscore}}ty}$^\alpha_{1..n}$,
  which we lifted in the previous section and which are all well-founded. It
  is straightforward to establish that the sizes decrease in every
  induction step. What is left to show is that we covered all cases. 
  To do so, we have to derive stronger cases lemmas, which look in our
  running example as follows:

  \[\mbox{
  \begin{tabular}{@ {}c@ {\hspace{4mm}}c@ {}}
  \infer{\isa{P\isaliteral{5C3C5E627375623E}{}\isactrlbsub trm\isaliteral{5C3C5E657375623E}{}\isactrlesub }}
  {\begin{array}{@ {}l@ {}}
   \isa{{\isaliteral{5C3C666F72616C6C3E}{\isasymforall}}x{\isaliteral{2E}{\isachardot}}\ y\ {\isaliteral{3D}{\isacharequal}}\ Var\isaliteral{5C3C5E7375703E}{}\isactrlsup {\isaliteral{5C3C616C7068613E}{\isasymalpha}}\ x\ {\isaliteral{5C3C52696768746172726F773E}{\isasymRightarrow}}\ P\isaliteral{5C3C5E627375623E}{}\isactrlbsub trm\isaliteral{5C3C5E657375623E}{}\isactrlesub }\\
   \isa{{\isaliteral{5C3C666F72616C6C3E}{\isasymforall}}x\isaliteral{5C3C5E697375623E}{}\isactrlisub {\isadigit{1}}\ x\isaliteral{5C3C5E697375623E}{}\isactrlisub {\isadigit{2}}{\isaliteral{2E}{\isachardot}}\ y\ {\isaliteral{3D}{\isacharequal}}\ App\isaliteral{5C3C5E7375703E}{}\isactrlsup {\isaliteral{5C3C616C7068613E}{\isasymalpha}}\ x\isaliteral{5C3C5E697375623E}{}\isactrlisub {\isadigit{1}}\ x\isaliteral{5C3C5E697375623E}{}\isactrlisub {\isadigit{2}}\ {\isaliteral{5C3C52696768746172726F773E}{\isasymRightarrow}}\ P\isaliteral{5C3C5E627375623E}{}\isactrlbsub trm\isaliteral{5C3C5E657375623E}{}\isactrlesub }\\
   \isa{{\isaliteral{5C3C666F72616C6C3E}{\isasymforall}}x\isaliteral{5C3C5E697375623E}{}\isactrlisub {\isadigit{1}}\ x\isaliteral{5C3C5E697375623E}{}\isactrlisub {\isadigit{2}}{\isaliteral{2E}{\isachardot}}\ atom\ x\isaliteral{5C3C5E697375623E}{}\isactrlisub {\isadigit{1}}\ {\isaliteral{23}{\isacharhash}}\ c\ {\isaliteral{5C3C616E643E}{\isasymand}}\ y\ {\isaliteral{3D}{\isacharequal}}\ Lam\isaliteral{5C3C5E7375703E}{}\isactrlsup {\isaliteral{5C3C616C7068613E}{\isasymalpha}}\ x\isaliteral{5C3C5E697375623E}{}\isactrlisub {\isadigit{1}}\ x\isaliteral{5C3C5E697375623E}{}\isactrlisub {\isadigit{2}}\ {\isaliteral{5C3C52696768746172726F773E}{\isasymRightarrow}}\ P\isaliteral{5C3C5E627375623E}{}\isactrlbsub trm\isaliteral{5C3C5E657375623E}{}\isactrlesub }\\
   \isa{{\isaliteral{5C3C666F72616C6C3E}{\isasymforall}}x\isaliteral{5C3C5E697375623E}{}\isactrlisub {\isadigit{1}}\ x\isaliteral{5C3C5E697375623E}{}\isactrlisub {\isadigit{2}}\ x\isaliteral{5C3C5E697375623E}{}\isactrlisub {\isadigit{3}}{\isaliteral{2E}{\isachardot}}\ set\ {\isaliteral{28}{\isacharparenleft}}bn\isaliteral{5C3C5E7375703E}{}\isactrlsup {\isaliteral{5C3C616C7068613E}{\isasymalpha}}\ x\isaliteral{5C3C5E697375623E}{}\isactrlisub {\isadigit{1}}{\isaliteral{29}{\isacharparenright}}\ {\isaliteral{23}{\isacharhash}}\isaliteral{5C3C5E7375703E}{}\isactrlsup {\isaliteral{2A}{\isacharasterisk}}\ c\ {\isaliteral{5C3C616E643E}{\isasymand}}\ y\ {\isaliteral{3D}{\isacharequal}}\ Let{\isaliteral{5F}{\isacharunderscore}}pat\isaliteral{5C3C5E7375703E}{}\isactrlsup {\isaliteral{5C3C616C7068613E}{\isasymalpha}}\ x\isaliteral{5C3C5E697375623E}{}\isactrlisub {\isadigit{1}}\ x\isaliteral{5C3C5E697375623E}{}\isactrlisub {\isadigit{2}}\ x\isaliteral{5C3C5E697375623E}{}\isactrlisub {\isadigit{3}}\ {\isaliteral{5C3C52696768746172726F773E}{\isasymRightarrow}}\ P\isaliteral{5C3C5E627375623E}{}\isactrlbsub trm\isaliteral{5C3C5E657375623E}{}\isactrlesub }
   \end{array}} &

  \infer{\isa{P\isaliteral{5C3C5E627375623E}{}\isactrlbsub pat\isaliteral{5C3C5E657375623E}{}\isactrlesub }}
  {\begin{array}{@ {}l@ {}}
   \isa{{\isaliteral{5C3C666F72616C6C3E}{\isasymforall}}x{\isaliteral{2E}{\isachardot}}\ y\ {\isaliteral{3D}{\isacharequal}}\ PVar\isaliteral{5C3C5E7375703E}{}\isactrlsup {\isaliteral{5C3C616C7068613E}{\isasymalpha}}\ x\ {\isaliteral{5C3C52696768746172726F773E}{\isasymRightarrow}}\ P\isaliteral{5C3C5E627375623E}{}\isactrlbsub pat\isaliteral{5C3C5E657375623E}{}\isactrlesub }\\
   \isa{{\isaliteral{5C3C666F72616C6C3E}{\isasymforall}}x\isaliteral{5C3C5E697375623E}{}\isactrlisub {\isadigit{1}}\ x\isaliteral{5C3C5E697375623E}{}\isactrlisub {\isadigit{2}}{\isaliteral{2E}{\isachardot}}\ y\ {\isaliteral{3D}{\isacharequal}}\ PTup\isaliteral{5C3C5E7375703E}{}\isactrlsup {\isaliteral{5C3C616C7068613E}{\isasymalpha}}\ x\isaliteral{5C3C5E697375623E}{}\isactrlisub {\isadigit{1}}\ x\isaliteral{5C3C5E697375623E}{}\isactrlisub {\isadigit{2}}\ {\isaliteral{5C3C52696768746172726F773E}{\isasymRightarrow}}\ P\isaliteral{5C3C5E627375623E}{}\isactrlbsub pat\isaliteral{5C3C5E657375623E}{}\isactrlesub }
  \end{array}}
  \end{tabular}}
  \]\smallskip 

  \noindent
  They are stronger in the sense that they allow us to assume in the \isa{Lam\isaliteral{5C3C5E7375703E}{}\isactrlsup {\isaliteral{5C3C616C7068613E}{\isasymalpha}}} and \isa{Let{\isaliteral{5F}{\isacharunderscore}}pat\isaliteral{5C3C5E7375703E}{}\isactrlsup {\isaliteral{5C3C616C7068613E}{\isasymalpha}}} cases that the bound atoms
  avoid, or are fresh for, a context \isa{c} (which is assumed to be finitely supported).
  
  These stronger cases lemmas can be derived from the `weak' cases lemmas
  given in \eqref{inductex}. This is trivial in case of patterns (the one on
  the right-hand side) since the weak and strong cases lemma coincide (there
  is no binding in patterns).  Interesting are only the cases for \isa{Lam\isaliteral{5C3C5E7375703E}{}\isactrlsup {\isaliteral{5C3C616C7068613E}{\isasymalpha}}} and \isa{Let{\isaliteral{5F}{\isacharunderscore}}pat\isaliteral{5C3C5E7375703E}{}\isactrlsup {\isaliteral{5C3C616C7068613E}{\isasymalpha}}}, where we have some binders and
  therefore have an additional assumption about avoiding \isa{c}.  Let us
  first establish the case for \isa{Lam\isaliteral{5C3C5E7375703E}{}\isactrlsup {\isaliteral{5C3C616C7068613E}{\isasymalpha}}}. By the weak cases lemma
  \eqref{inductex} we can assume that

  \begin{equation}\label{assm}
  \isa{y\ {\isaliteral{3D}{\isacharequal}}\ Lam\isaliteral{5C3C5E7375703E}{}\isactrlsup {\isaliteral{5C3C616C7068613E}{\isasymalpha}}\ x\isaliteral{5C3C5E697375623E}{}\isactrlisub {\isadigit{1}}\ x\isaliteral{5C3C5E697375623E}{}\isactrlisub {\isadigit{2}}}
  \end{equation}\smallskip

  \noindent
  holds, and need to establish \isa{P\isaliteral{5C3C5E627375623E}{}\isactrlbsub trm\isaliteral{5C3C5E657375623E}{}\isactrlesub }. The stronger cases lemma has the 
  corresponding implication 

  \begin{equation}\label{imp}
  \isa{{\isaliteral{5C3C666F72616C6C3E}{\isasymforall}}x\isaliteral{5C3C5E697375623E}{}\isactrlisub {\isadigit{1}}\ x\isaliteral{5C3C5E697375623E}{}\isactrlisub {\isadigit{2}}{\isaliteral{2E}{\isachardot}}\ atom\ x\isaliteral{5C3C5E697375623E}{}\isactrlisub {\isadigit{1}}\ {\isaliteral{23}{\isacharhash}}\ c\ {\isaliteral{5C3C616E643E}{\isasymand}}\ y\ {\isaliteral{3D}{\isacharequal}}\ Lam\isaliteral{5C3C5E7375703E}{}\isactrlsup {\isaliteral{5C3C616C7068613E}{\isasymalpha}}\ x\isaliteral{5C3C5E697375623E}{}\isactrlisub {\isadigit{1}}\ x\isaliteral{5C3C5E697375623E}{}\isactrlisub {\isadigit{2}}\ {\isaliteral{5C3C52696768746172726F773E}{\isasymRightarrow}}\ P\isaliteral{5C3C5E627375623E}{}\isactrlbsub trm\isaliteral{5C3C5E657375623E}{}\isactrlesub }
  \end{equation}\smallskip

  \noindent
  which we must use in order to infer \isa{P\isaliteral{5C3C5E627375623E}{}\isactrlbsub trm\isaliteral{5C3C5E657375623E}{}\isactrlesub }. Clearly, we cannot
  use this implication directly, because we have no information whether or not \isa{x\isaliteral{5C3C5E697375623E}{}\isactrlisub {\isadigit{1}}} is fresh for \isa{c}.  However, we can use Properties
  \ref{supppermeq} and \ref{avoiding} to rename \isa{x\isaliteral{5C3C5E697375623E}{}\isactrlisub {\isadigit{1}}}. We know
  by Theorem~\ref{suppfa} that \isa{{\isaliteral{7B}{\isacharbraceleft}}atom\ x\isaliteral{5C3C5E697375623E}{}\isactrlisub {\isadigit{1}}{\isaliteral{7D}{\isacharbraceright}}\ {\isaliteral{23}{\isacharhash}}\isaliteral{5C3C5E7375703E}{}\isactrlsup {\isaliteral{2A}{\isacharasterisk}}\ Lam\isaliteral{5C3C5E7375703E}{}\isactrlsup {\isaliteral{5C3C616C7068613E}{\isasymalpha}}\ x\isaliteral{5C3C5E697375623E}{}\isactrlisub {\isadigit{1}}\ x\isaliteral{5C3C5E697375623E}{}\isactrlisub {\isadigit{2}}} (since its support is \isa{supp\ x\isaliteral{5C3C5E697375623E}{}\isactrlisub {\isadigit{2}}\ {\isaliteral{2D}{\isacharminus}}\ {\isaliteral{7B}{\isacharbraceleft}}atom\ x\isaliteral{5C3C5E697375623E}{}\isactrlisub {\isadigit{1}}{\isaliteral{7D}{\isacharbraceright}}}). Property \ref{avoiding} provides us then with a
  permutation \isa{{\isaliteral{5C3C70693E}{\isasympi}}}, such that \isa{{\isaliteral{7B}{\isacharbraceleft}}atom\ {\isaliteral{28}{\isacharparenleft}}{\isaliteral{5C3C70693E}{\isasympi}}\ {\isaliteral{5C3C62756C6C65743E}{\isasymbullet}}\ x\isaliteral{5C3C5E697375623E}{}\isactrlisub {\isadigit{1}}{\isaliteral{29}{\isacharparenright}}{\isaliteral{7D}{\isacharbraceright}}\ {\isaliteral{23}{\isacharhash}}\isaliteral{5C3C5E7375703E}{}\isactrlsup {\isaliteral{2A}{\isacharasterisk}}\ c} and \mbox{\isa{supp\ {\isaliteral{28}{\isacharparenleft}}Lam\isaliteral{5C3C5E7375703E}{}\isactrlsup {\isaliteral{5C3C616C7068613E}{\isasymalpha}}\ x\isaliteral{5C3C5E697375623E}{}\isactrlisub {\isadigit{1}}\ x\isaliteral{5C3C5E697375623E}{}\isactrlisub {\isadigit{2}}{\isaliteral{29}{\isacharparenright}}\ {\isaliteral{23}{\isacharhash}}\isaliteral{5C3C5E7375703E}{}\isactrlsup {\isaliteral{2A}{\isacharasterisk}}\ {\isaliteral{5C3C70693E}{\isasympi}}}} hold.
  By using Property \ref{supppermeq}, we can infer from the latter that 

  \[
  \isa{Lam\isaliteral{5C3C5E7375703E}{}\isactrlsup {\isaliteral{5C3C616C7068613E}{\isasymalpha}}\ {\isaliteral{28}{\isacharparenleft}}{\isaliteral{5C3C70693E}{\isasympi}}\ {\isaliteral{5C3C62756C6C65743E}{\isasymbullet}}\ x\isaliteral{5C3C5E697375623E}{}\isactrlisub {\isadigit{1}}{\isaliteral{29}{\isacharparenright}}\ {\isaliteral{28}{\isacharparenleft}}{\isaliteral{5C3C70693E}{\isasympi}}\ {\isaliteral{5C3C62756C6C65743E}{\isasymbullet}}\ x\isaliteral{5C3C5E697375623E}{}\isactrlisub {\isadigit{2}}{\isaliteral{29}{\isacharparenright}}\ {\isaliteral{3D}{\isacharequal}}\ Lam\isaliteral{5C3C5E7375703E}{}\isactrlsup {\isaliteral{5C3C616C7068613E}{\isasymalpha}}\ x\isaliteral{5C3C5E697375623E}{}\isactrlisub {\isadigit{1}}\ x\isaliteral{5C3C5E697375623E}{}\isactrlisub {\isadigit{2}}} 
  \]\smallskip

  \noindent
  holds. We can use this equation in the assumption \eqref{assm}, and hence
  use the implication \eqref{imp} with the renamed \isa{{\isaliteral{5C3C70693E}{\isasympi}}\ {\isaliteral{5C3C62756C6C65743E}{\isasymbullet}}\ x\isaliteral{5C3C5E697375623E}{}\isactrlisub {\isadigit{1}}}
  and \isa{{\isaliteral{5C3C70693E}{\isasympi}}\ {\isaliteral{5C3C62756C6C65743E}{\isasymbullet}}\ x\isaliteral{5C3C5E697375623E}{}\isactrlisub {\isadigit{2}}} for concluding this case.

  The \isa{Let{\isaliteral{5F}{\isacharunderscore}}pat\isaliteral{5C3C5E7375703E}{}\isactrlsup {\isaliteral{5C3C616C7068613E}{\isasymalpha}}}-case involving a deep binder is slightly more complicated.
  We have the assumption

  \begin{equation}\label{assmtwo}
  \isa{y\ {\isaliteral{3D}{\isacharequal}}\ Let{\isaliteral{5F}{\isacharunderscore}}pat\isaliteral{5C3C5E7375703E}{}\isactrlsup {\isaliteral{5C3C616C7068613E}{\isasymalpha}}\ x\isaliteral{5C3C5E697375623E}{}\isactrlisub {\isadigit{1}}\ x\isaliteral{5C3C5E697375623E}{}\isactrlisub {\isadigit{2}}\ x\isaliteral{5C3C5E697375623E}{}\isactrlisub {\isadigit{3}}}
  \end{equation}\smallskip

  \noindent
  and the implication from the stronger cases lemma

  \begin{equation}\label{impletpat}
  \isa{{\isaliteral{5C3C666F72616C6C3E}{\isasymforall}}x\isaliteral{5C3C5E697375623E}{}\isactrlisub {\isadigit{1}}\ x\isaliteral{5C3C5E697375623E}{}\isactrlisub {\isadigit{2}}\ x\isaliteral{5C3C5E697375623E}{}\isactrlisub {\isadigit{3}}{\isaliteral{2E}{\isachardot}}\ set\ {\isaliteral{28}{\isacharparenleft}}bn\isaliteral{5C3C5E7375703E}{}\isactrlsup {\isaliteral{5C3C616C7068613E}{\isasymalpha}}\ x\isaliteral{5C3C5E697375623E}{}\isactrlisub {\isadigit{1}}{\isaliteral{29}{\isacharparenright}}\ {\isaliteral{23}{\isacharhash}}\isaliteral{5C3C5E7375703E}{}\isactrlsup {\isaliteral{2A}{\isacharasterisk}}\ c\ {\isaliteral{5C3C616E643E}{\isasymand}}\ y\ {\isaliteral{3D}{\isacharequal}}\ Let{\isaliteral{5F}{\isacharunderscore}}pat\isaliteral{5C3C5E7375703E}{}\isactrlsup {\isaliteral{5C3C616C7068613E}{\isasymalpha}}\ x\isaliteral{5C3C5E697375623E}{}\isactrlisub {\isadigit{1}}\ x\isaliteral{5C3C5E697375623E}{}\isactrlisub {\isadigit{2}}\ x\isaliteral{5C3C5E697375623E}{}\isactrlisub {\isadigit{3}}\ {\isaliteral{5C3C52696768746172726F773E}{\isasymRightarrow}}\ P\isaliteral{5C3C5E627375623E}{}\isactrlbsub trm\isaliteral{5C3C5E657375623E}{}\isactrlesub }
  \end{equation}\smallskip

  \noindent
  The reason that this case is more complicated is that we cannot directly apply Property 
  \ref{avoiding} for obtaining a renaming permutation. Property \ref{avoiding} requires
  that the binders are fresh for the term in which we want to perform the renaming. But
  this is not true in terms such as (using an informal notation)

  \[
  \isa{Let\ {\isaliteral{28}{\isacharparenleft}}x{\isaliteral{2C}{\isacharcomma}}\ y{\isaliteral{29}{\isacharparenright}}\ {\isaliteral{3A}{\isacharcolon}}{\isaliteral{3D}{\isacharequal}}\ {\isaliteral{28}{\isacharparenleft}}x{\isaliteral{2C}{\isacharcomma}}\ y{\isaliteral{29}{\isacharparenright}}\ in\ {\isaliteral{28}{\isacharparenleft}}x{\isaliteral{2C}{\isacharcomma}}\ y{\isaliteral{29}{\isacharparenright}}}
  \]\smallskip

  \noindent
  where \isa{x} and \isa{y} are bound in the term, but are also free
  in the right-hand side of the assignment. We can, however, obtain such a renaming permutation, say
  \isa{{\isaliteral{5C3C70693E}{\isasympi}}}, for the abstraction \isa{{\isaliteral{5B}{\isacharbrackleft}}bn\isaliteral{5C3C5E7375703E}{}\isactrlsup {\isaliteral{5C3C616C7068613E}{\isasymalpha}}\ x\isaliteral{5C3C5E697375623E}{}\isactrlisub {\isadigit{1}}{\isaliteral{5D}{\isacharbrackright}}\isaliteral{5C3C5E627375623E}{}\isactrlbsub list\isaliteral{5C3C5E657375623E}{}\isactrlesub {\isaliteral{2E}{\isachardot}}x\isaliteral{5C3C5E697375623E}{}\isactrlisub {\isadigit{3}}}. As a result we have \mbox{\isa{set\ {\isaliteral{28}{\isacharparenleft}}bn\isaliteral{5C3C5E7375703E}{}\isactrlsup {\isaliteral{5C3C616C7068613E}{\isasymalpha}}\ {\isaliteral{28}{\isacharparenleft}}{\isaliteral{5C3C70693E}{\isasympi}}\ {\isaliteral{5C3C62756C6C65743E}{\isasymbullet}}\ x\isaliteral{5C3C5E697375623E}{}\isactrlisub {\isadigit{1}}{\isaliteral{29}{\isacharparenright}}{\isaliteral{29}{\isacharparenright}}\ {\isaliteral{23}{\isacharhash}}\isaliteral{5C3C5E7375703E}{}\isactrlsup {\isaliteral{2A}{\isacharasterisk}}\ c}} and \isa{{\isaliteral{5B}{\isacharbrackleft}}bn\isaliteral{5C3C5E7375703E}{}\isactrlsup {\isaliteral{5C3C616C7068613E}{\isasymalpha}}\ {\isaliteral{28}{\isacharparenleft}}{\isaliteral{5C3C70693E}{\isasympi}}\ {\isaliteral{5C3C62756C6C65743E}{\isasymbullet}}\ x\isaliteral{5C3C5E697375623E}{}\isactrlisub {\isadigit{1}}{\isaliteral{29}{\isacharparenright}}{\isaliteral{5D}{\isacharbrackright}}\isaliteral{5C3C5E627375623E}{}\isactrlbsub list\isaliteral{5C3C5E657375623E}{}\isactrlesub {\isaliteral{2E}{\isachardot}}{\isaliteral{28}{\isacharparenleft}}{\isaliteral{5C3C70693E}{\isasympi}}\ {\isaliteral{5C3C62756C6C65743E}{\isasymbullet}}\ x\isaliteral{5C3C5E697375623E}{}\isactrlisub {\isadigit{3}}{\isaliteral{29}{\isacharparenright}}\ {\isaliteral{3D}{\isacharequal}}\ {\isaliteral{5B}{\isacharbrackleft}}bn\isaliteral{5C3C5E7375703E}{}\isactrlsup {\isaliteral{5C3C616C7068613E}{\isasymalpha}}\ x\isaliteral{5C3C5E697375623E}{}\isactrlisub {\isadigit{1}}{\isaliteral{5D}{\isacharbrackright}}\isaliteral{5C3C5E627375623E}{}\isactrlbsub list\isaliteral{5C3C5E657375623E}{}\isactrlesub {\isaliteral{2E}{\isachardot}}x\isaliteral{5C3C5E697375623E}{}\isactrlisub {\isadigit{3}}} (remember \isa{set} and \isa{bn\isaliteral{5C3C5E7375703E}{}\isactrlsup {\isaliteral{5C3C616C7068613E}{\isasymalpha}}} are equivariant).  Now the quasi-injective property for \isa{Let{\isaliteral{5F}{\isacharunderscore}}pat\isaliteral{5C3C5E7375703E}{}\isactrlsup {\isaliteral{5C3C616C7068613E}{\isasymalpha}}} states that

  \[
  \infer{\isa{Let{\isaliteral{5F}{\isacharunderscore}}pat\isaliteral{5C3C5E7375703E}{}\isactrlsup {\isaliteral{5C3C616C7068613E}{\isasymalpha}}\ p\ t\isaliteral{5C3C5E697375623E}{}\isactrlisub {\isadigit{1}}\ t\isaliteral{5C3C5E697375623E}{}\isactrlisub {\isadigit{2}}\ {\isaliteral{3D}{\isacharequal}}\ Let{\isaliteral{5F}{\isacharunderscore}}pat\isaliteral{5C3C5E7375703E}{}\isactrlsup {\isaliteral{5C3C616C7068613E}{\isasymalpha}}\ p{\isaliteral{5C3C5052494D453E}{\isasymPRIME}}\ t{\isaliteral{5C3C5052494D453E}{\isasymPRIME}}\isaliteral{5C3C5E697375623E}{}\isactrlisub {\isadigit{1}}\ t{\isaliteral{5C3C5052494D453E}{\isasymPRIME}}\isaliteral{5C3C5E697375623E}{}\isactrlisub {\isadigit{2}}}}
  {\isa{{\isaliteral{5B}{\isacharbrackleft}}bn\isaliteral{5C3C5E7375703E}{}\isactrlsup {\isaliteral{5C3C616C7068613E}{\isasymalpha}}\ p{\isaliteral{5D}{\isacharbrackright}}\isaliteral{5C3C5E627375623E}{}\isactrlbsub list\isaliteral{5C3C5E657375623E}{}\isactrlesub {\isaliteral{2E}{\isachardot}}\ t\isaliteral{5C3C5E697375623E}{}\isactrlisub {\isadigit{2}}\ {\isaliteral{3D}{\isacharequal}}\ {\isaliteral{5B}{\isacharbrackleft}}bn\isaliteral{5C3C5E7375703E}{}\isactrlsup {\isaliteral{5C3C616C7068613E}{\isasymalpha}}\ p{\isaliteral{27}{\isacharprime}}{\isaliteral{5D}{\isacharbrackright}}\isaliteral{5C3C5E627375623E}{}\isactrlbsub list\isaliteral{5C3C5E657375623E}{}\isactrlesub {\isaliteral{2E}{\isachardot}}\ t{\isaliteral{5C3C5052494D453E}{\isasymPRIME}}\isaliteral{5C3C5E697375623E}{}\isactrlisub {\isadigit{2}}}\;\; & 
  \isa{p\ {\isaliteral{5C3C617070726F783E}{\isasymapprox}}{\isaliteral{5C3C414C3E}{\isasymAL}}\isaliteral{5C3C5E627375623E}{}\isactrlbsub bn\isaliteral{5C3C5E657375623E}{}\isactrlesub \ p{\isaliteral{5C3C5052494D453E}{\isasymPRIME}}}\;\; & \isa{t\isaliteral{5C3C5E697375623E}{}\isactrlisub {\isadigit{1}}\ {\isaliteral{3D}{\isacharequal}}\ t{\isaliteral{5C3C5052494D453E}{\isasymPRIME}}\isaliteral{5C3C5E697375623E}{}\isactrlisub {\isadigit{1}}}}
  \]\smallskip

  \noindent
  Since all atoms in a pattern are bound by \isa{Let{\isaliteral{5F}{\isacharunderscore}}pat\isaliteral{5C3C5E7375703E}{}\isactrlsup {\isaliteral{5C3C616C7068613E}{\isasymalpha}}}, we can infer
  that \isa{{\isaliteral{28}{\isacharparenleft}}{\isaliteral{5C3C70693E}{\isasympi}}\ {\isaliteral{5C3C62756C6C65743E}{\isasymbullet}}\ x\isaliteral{5C3C5E697375623E}{}\isactrlisub {\isadigit{1}}{\isaliteral{29}{\isacharparenright}}\ {\isaliteral{5C3C617070726F783E}{\isasymapprox}}{\isaliteral{5C3C414C3E}{\isasymAL}}\isaliteral{5C3C5E627375623E}{}\isactrlbsub bn\isaliteral{5C3C5E657375623E}{}\isactrlesub \ x\isaliteral{5C3C5E697375623E}{}\isactrlisub {\isadigit{1}}} holds for every \isa{{\isaliteral{5C3C70693E}{\isasympi}}}. Therefore we have that

  \[
  \isa{Let{\isaliteral{5F}{\isacharunderscore}}pat\isaliteral{5C3C5E7375703E}{}\isactrlsup {\isaliteral{5C3C616C7068613E}{\isasymalpha}}\ {\isaliteral{28}{\isacharparenleft}}{\isaliteral{5C3C70693E}{\isasympi}}\ {\isaliteral{5C3C62756C6C65743E}{\isasymbullet}}\ x\isaliteral{5C3C5E697375623E}{}\isactrlisub {\isadigit{1}}{\isaliteral{29}{\isacharparenright}}\ x\isaliteral{5C3C5E697375623E}{}\isactrlisub {\isadigit{2}}\ {\isaliteral{28}{\isacharparenleft}}{\isaliteral{5C3C70693E}{\isasympi}}\ {\isaliteral{5C3C62756C6C65743E}{\isasymbullet}}\ x\isaliteral{5C3C5E697375623E}{}\isactrlisub {\isadigit{3}}{\isaliteral{29}{\isacharparenright}}\ {\isaliteral{3D}{\isacharequal}}\ Let{\isaliteral{5F}{\isacharunderscore}}pat\isaliteral{5C3C5E7375703E}{}\isactrlsup {\isaliteral{5C3C616C7068613E}{\isasymalpha}}\ x\isaliteral{5C3C5E697375623E}{}\isactrlisub {\isadigit{1}}\ x\isaliteral{5C3C5E697375623E}{}\isactrlisub {\isadigit{2}}\ x\isaliteral{5C3C5E697375623E}{}\isactrlisub {\isadigit{3}}}  
  \]\smallskip
  
  \noindent
  Taking the left-hand side in the assumption shown in \eqref{assmtwo}, we can use
  the implication \eqref{impletpat} from the stronger cases lemma to infer \isa{P\isaliteral{5C3C5E627375623E}{}\isactrlbsub trm\isaliteral{5C3C5E657375623E}{}\isactrlesub }, as needed.

  The remaining difficulty is when a deep binder contains some atoms that are
  bound and some that are free. An example is \isa{Let\isaliteral{5C3C5E7375703E}{}\isactrlsup {\isaliteral{5C3C616C7068613E}{\isasymalpha}}} in
  \eqref{letrecs}.  In such cases \isa{{\isaliteral{28}{\isacharparenleft}}{\isaliteral{5C3C70693E}{\isasympi}}\ {\isaliteral{5C3C62756C6C65743E}{\isasymbullet}}\ x\isaliteral{5C3C5E697375623E}{}\isactrlisub {\isadigit{1}}{\isaliteral{29}{\isacharparenright}}\ {\isaliteral{5C3C617070726F783E}{\isasymapprox}}{\isaliteral{5C3C414C3E}{\isasymAL}}\isaliteral{5C3C5E627375623E}{}\isactrlbsub bn\isaliteral{5C3C5E657375623E}{}\isactrlesub \ x\isaliteral{5C3C5E697375623E}{}\isactrlisub {\isadigit{1}}} does not hold in general. The idea however is
  that \isa{{\isaliteral{5C3C70693E}{\isasympi}}} only renames atoms that become bound. In this way \isa{{\isaliteral{5C3C70693E}{\isasympi}}}
  does not affect \isa{{\isaliteral{5C3C617070726F783E}{\isasymapprox}}{\isaliteral{5C3C414C3E}{\isasymAL}}\isaliteral{5C3C5E627375623E}{}\isactrlbsub bn\isaliteral{5C3C5E657375623E}{}\isactrlesub } (which only tracks alpha-equivalence of terms that are not
  under the binder). However, the problem is that the
  permutation operation \isa{{\isaliteral{5C3C70693E}{\isasympi}}\ {\isaliteral{5C3C62756C6C65743E}{\isasymbullet}}\ x\isaliteral{5C3C5E697375623E}{}\isactrlisub {\isadigit{1}}} applies to all atoms in \isa{x\isaliteral{5C3C5E697375623E}{}\isactrlisub {\isadigit{1}}}. To avoid this
  we introduce an auxiliary permutation operations, written \isa{{\isaliteral{5F}{\isacharunderscore}}\ {\isaliteral{5C3C62756C6C65743E}{\isasymbullet}}\isaliteral{5C3C5E627375623E}{}\isactrlbsub bn\isaliteral{5C3C5E657375623E}{}\isactrlesub \ {\isaliteral{5F}{\isacharunderscore}}}, for deep binders that only permutes bound atoms (or
  more precisely the atoms specified by the \isa{bn}-functions) and leaves
  the other atoms unchanged. Like the functions \isa{fa{\isaliteral{5F}{\isacharunderscore}}bn}$_{1..m}$, we
  can define these permutation operations over raw terms analysing how the functions \isa{bn}$_{1..m}$ are defined. Assuming the user specified a clause

  \[  
  \isa{bn\ {\isaliteral{28}{\isacharparenleft}}C\ x\isaliteral{5C3C5E697375623E}{}\isactrlisub {\isadigit{1}}\ {\isaliteral{5C3C646F74733E}{\isasymdots}}\ x\isaliteral{5C3C5E697375623E}{}\isactrlisub r{\isaliteral{29}{\isacharparenright}}\ {\isaliteral{3D}{\isacharequal}}\ rhs}
  \]\smallskip

  \noindent
  we define \isa{{\isaliteral{5C3C70693E}{\isasympi}}\ {\isaliteral{5C3C62756C6C65743E}{\isasymbullet}}\isaliteral{5C3C5E627375623E}{}\isactrlbsub bn\isaliteral{5C3C5E657375623E}{}\isactrlesub \ {\isaliteral{28}{\isacharparenleft}}C\ x\isaliteral{5C3C5E697375623E}{}\isactrlisub {\isadigit{1}}\ {\isaliteral{5C3C646F74733E}{\isasymdots}}\ x\isaliteral{5C3C5E697375623E}{}\isactrlisub r{\isaliteral{29}{\isacharparenright}}\ {\isaliteral{5C3C65717569763E}{\isasymequiv}}\ C\ y\isaliteral{5C3C5E697375623E}{}\isactrlisub {\isadigit{1}}\ {\isaliteral{5C3C646F74733E}{\isasymdots}}\ y\isaliteral{5C3C5E697375623E}{}\isactrlisub r} with \isa{y\isaliteral{5C3C5E697375623E}{}\isactrlisub i} determined as follows:

  \[\mbox{
  \begin{tabular}{c@ {\hspace{2mm}}p{0.9\textwidth}}
  $\bullet$ & \isa{y\isaliteral{5C3C5E697375623E}{}\isactrlisub i\ {\isaliteral{5C3C65717569763E}{\isasymequiv}}\ x\isaliteral{5C3C5E697375623E}{}\isactrlisub i} provided \isa{x\isaliteral{5C3C5E697375623E}{}\isactrlisub i} does not occur in \isa{rhs}\\
  $\bullet$ & \isa{y\isaliteral{5C3C5E697375623E}{}\isactrlisub i\ {\isaliteral{5C3C65717569763E}{\isasymequiv}}\ {\isaliteral{5C3C70693E}{\isasympi}}\ {\isaliteral{5C3C62756C6C65743E}{\isasymbullet}}\isaliteral{5C3C5E627375623E}{}\isactrlbsub bn\isaliteral{5C3C5E657375623E}{}\isactrlesub \ x\isaliteral{5C3C5E697375623E}{}\isactrlisub i} provided \isa{bn\ x\isaliteral{5C3C5E697375623E}{}\isactrlisub i} is in \isa{rhs}\\
  $\bullet$ & \isa{y\isaliteral{5C3C5E697375623E}{}\isactrlisub i\ {\isaliteral{5C3C65717569763E}{\isasymequiv}}\ {\isaliteral{5C3C70693E}{\isasympi}}\ {\isaliteral{5C3C62756C6C65743E}{\isasymbullet}}\ x\isaliteral{5C3C5E697375623E}{}\isactrlisub i} otherwise
  \end{tabular}}
  \]\smallskip

  \noindent
  Using again the quotient package  we can lift the auxiliary permutation operations
  \isa{{\isaliteral{5F}{\isacharunderscore}}\ {\isaliteral{5C3C62756C6C65743E}{\isasymbullet}}\isaliteral{5C3C5E627375623E}{}\isactrlbsub bn\isaliteral{5C3C5E657375623E}{}\isactrlesub \ {\isaliteral{5F}{\isacharunderscore}}}
  to alpha-equated terms. Moreover we can prove the following two properties:

  \begin{lem}\label{permutebn} 
  Given a binding function \isa{bn\isaliteral{5C3C5E7375703E}{}\isactrlsup {\isaliteral{5C3C616C7068613E}{\isasymalpha}}} and auxiliary equivalence \isa{{\isaliteral{5C3C617070726F783E}{\isasymapprox}}{\isaliteral{5C3C414C3E}{\isasymAL}}\isaliteral{5C3C5E627375623E}{}\isactrlbsub bn\isaliteral{5C3C5E657375623E}{}\isactrlesub } 
  then for all \isa{{\isaliteral{5C3C70693E}{\isasympi}}}\smallskip\\
  {\it (i)} \isa{{\isaliteral{5C3C70693E}{\isasympi}}\ {\isaliteral{5C3C62756C6C65743E}{\isasymbullet}}\ {\isaliteral{28}{\isacharparenleft}}bn\isaliteral{5C3C5E7375703E}{}\isactrlsup {\isaliteral{5C3C616C7068613E}{\isasymalpha}}\ x{\isaliteral{29}{\isacharparenright}}\ {\isaliteral{3D}{\isacharequal}}\ bn\isaliteral{5C3C5E7375703E}{}\isactrlsup {\isaliteral{5C3C616C7068613E}{\isasymalpha}}\ {\isaliteral{28}{\isacharparenleft}}{\isaliteral{5C3C70693E}{\isasympi}}\ {\isaliteral{5C3C62756C6C65743E}{\isasymbullet}}{\isaliteral{5C3C414C3E}{\isasymAL}}\isaliteral{5C3C5E627375623E}{}\isactrlbsub bn\isaliteral{5C3C5E657375623E}{}\isactrlesub \ x{\isaliteral{29}{\isacharparenright}}} and\\ 
  {\it (ii)} \isa{{\isaliteral{28}{\isacharparenleft}}{\isaliteral{5C3C70693E}{\isasympi}}\ \ {\isaliteral{5C3C62756C6C65743E}{\isasymbullet}}{\isaliteral{5C3C414C3E}{\isasymAL}}\isaliteral{5C3C5E627375623E}{}\isactrlbsub bn\isaliteral{5C3C5E657375623E}{}\isactrlesub \ x{\isaliteral{29}{\isacharparenright}}\ {\isaliteral{5C3C617070726F783E}{\isasymapprox}}{\isaliteral{5C3C414C3E}{\isasymAL}}\isaliteral{5C3C5E627375623E}{}\isactrlbsub bn\isaliteral{5C3C5E657375623E}{}\isactrlesub \ x}.
  \end{lem}

  \begin{proof} 
  By induction on \isa{x}. The properties follow by unfolding of the
  definitions.
  \end{proof}

  \noindent
  The first property states that a permutation applied to a binding function
  is equivalent to first permuting the binders and then calculating the bound
  atoms. The second states that \isa{{\isaliteral{5F}{\isacharunderscore}}\ {\isaliteral{5C3C62756C6C65743E}{\isasymbullet}}{\isaliteral{5C3C414C3E}{\isasymAL}}\isaliteral{5C3C5E627375623E}{}\isactrlbsub bn\isaliteral{5C3C5E657375623E}{}\isactrlesub \ {\isaliteral{5F}{\isacharunderscore}}} preserves
  \isa{{\isaliteral{5C3C617070726F783E}{\isasymapprox}}{\isaliteral{5C3C414C3E}{\isasymAL}}\isaliteral{5C3C5E627375623E}{}\isactrlbsub bn\isaliteral{5C3C5E657375623E}{}\isactrlesub }.  The main point of the auxiliary
  permutation functions is that they allow us to rename just the bound atoms in a
  term, without changing anything else.
  
  Having the auxiliary permutation function in place, we can now solve all remaining cases. 
  For the \isa{Let\isaliteral{5C3C5E7375703E}{}\isactrlsup {\isaliteral{5C3C616C7068613E}{\isasymalpha}}} term-constructor, for example, we can by Property \ref{avoiding} 
  obtain a \isa{{\isaliteral{5C3C70693E}{\isasympi}}} such that 

  \[
  \isa{{\isaliteral{28}{\isacharparenleft}}{\isaliteral{5C3C70693E}{\isasympi}}\ {\isaliteral{5C3C62756C6C65743E}{\isasymbullet}}\ {\isaliteral{28}{\isacharparenleft}}set\ {\isaliteral{28}{\isacharparenleft}}bn\isaliteral{5C3C5E7375703E}{}\isactrlsup {\isaliteral{5C3C616C7068613E}{\isasymalpha}}\ x\isaliteral{5C3C5E697375623E}{}\isactrlisub {\isadigit{1}}{\isaliteral{29}{\isacharparenright}}{\isaliteral{29}{\isacharparenright}}\ {\isaliteral{23}{\isacharhash}}\isaliteral{5C3C5E7375703E}{}\isactrlsup {\isaliteral{2A}{\isacharasterisk}}\ c} \hspace{10mm}
  \isa{{\isaliteral{5C3C70693E}{\isasympi}}\ {\isaliteral{5C3C62756C6C65743E}{\isasymbullet}}\ {\isaliteral{5B}{\isacharbrackleft}}bn\isaliteral{5C3C5E7375703E}{}\isactrlsup {\isaliteral{5C3C616C7068613E}{\isasymalpha}}\ x\isaliteral{5C3C5E697375623E}{}\isactrlisub {\isadigit{1}}{\isaliteral{5D}{\isacharbrackright}}\isaliteral{5C3C5E627375623E}{}\isactrlbsub list\isaliteral{5C3C5E657375623E}{}\isactrlesub {\isaliteral{2E}{\isachardot}}\ x\isaliteral{5C3C5E697375623E}{}\isactrlisub {\isadigit{2}}\ {\isaliteral{3D}{\isacharequal}}\ {\isaliteral{5B}{\isacharbrackleft}}bn\isaliteral{5C3C5E7375703E}{}\isactrlsup {\isaliteral{5C3C616C7068613E}{\isasymalpha}}\ x\isaliteral{5C3C5E697375623E}{}\isactrlisub {\isadigit{1}}{\isaliteral{5D}{\isacharbrackright}}\isaliteral{5C3C5E627375623E}{}\isactrlbsub list\isaliteral{5C3C5E657375623E}{}\isactrlesub {\isaliteral{2E}{\isachardot}}\ x\isaliteral{5C3C5E697375623E}{}\isactrlisub {\isadigit{2}}} 
  \]\smallskip

  \noindent
  hold. Using the first part of Lemma \ref{permutebn}, we can simplify this
  to \isa{set\ {\isaliteral{28}{\isacharparenleft}}bn\isaliteral{5C3C5E7375703E}{}\isactrlsup {\isaliteral{5C3C616C7068613E}{\isasymalpha}}\ {\isaliteral{28}{\isacharparenleft}}{\isaliteral{5C3C70693E}{\isasympi}}\ {\isaliteral{5C3C62756C6C65743E}{\isasymbullet}}{\isaliteral{5C3C414C3E}{\isasymAL}}\isaliteral{5C3C5E627375623E}{}\isactrlbsub bn\isaliteral{5C3C5E657375623E}{}\isactrlesub \ x\isaliteral{5C3C5E697375623E}{}\isactrlisub {\isadigit{1}}{\isaliteral{29}{\isacharparenright}}{\isaliteral{29}{\isacharparenright}}\ {\isaliteral{23}{\isacharhash}}\isaliteral{5C3C5E7375703E}{}\isactrlsup {\isaliteral{2A}{\isacharasterisk}}\ c} and 
  \mbox{\isa{{\isaliteral{5B}{\isacharbrackleft}}bn\isaliteral{5C3C5E7375703E}{}\isactrlsup {\isaliteral{5C3C616C7068613E}{\isasymalpha}}\ {\isaliteral{28}{\isacharparenleft}}{\isaliteral{5C3C70693E}{\isasympi}}\ {\isaliteral{5C3C62756C6C65743E}{\isasymbullet}}{\isaliteral{5C3C414C3E}{\isasymAL}}\isaliteral{5C3C5E627375623E}{}\isactrlbsub bn\isaliteral{5C3C5E657375623E}{}\isactrlesub \ x\isaliteral{5C3C5E697375623E}{}\isactrlisub {\isadigit{1}}{\isaliteral{29}{\isacharparenright}}{\isaliteral{5D}{\isacharbrackright}}\isaliteral{5C3C5E627375623E}{}\isactrlbsub list\isaliteral{5C3C5E657375623E}{}\isactrlesub {\isaliteral{2E}{\isachardot}}\ {\isaliteral{28}{\isacharparenleft}}{\isaliteral{5C3C70693E}{\isasympi}}\ {\isaliteral{5C3C62756C6C65743E}{\isasymbullet}}\ x\isaliteral{5C3C5E697375623E}{}\isactrlisub {\isadigit{2}}{\isaliteral{29}{\isacharparenright}}\ {\isaliteral{3D}{\isacharequal}}\ {\isaliteral{5B}{\isacharbrackleft}}bn\isaliteral{5C3C5E7375703E}{}\isactrlsup {\isaliteral{5C3C616C7068613E}{\isasymalpha}}\ x\isaliteral{5C3C5E697375623E}{}\isactrlisub {\isadigit{1}}{\isaliteral{5D}{\isacharbrackright}}\isaliteral{5C3C5E627375623E}{}\isactrlbsub list\isaliteral{5C3C5E657375623E}{}\isactrlesub {\isaliteral{2E}{\isachardot}}\ x\isaliteral{5C3C5E697375623E}{}\isactrlisub {\isadigit{2}}}}. Since
  \isa{{\isaliteral{28}{\isacharparenleft}}{\isaliteral{5C3C70693E}{\isasympi}}\ \ {\isaliteral{5C3C62756C6C65743E}{\isasymbullet}}{\isaliteral{5C3C414C3E}{\isasymAL}}\isaliteral{5C3C5E627375623E}{}\isactrlbsub bn\isaliteral{5C3C5E657375623E}{}\isactrlesub \ x\isaliteral{5C3C5E697375623E}{}\isactrlisub {\isadigit{1}}{\isaliteral{29}{\isacharparenright}}\ {\isaliteral{5C3C617070726F783E}{\isasymapprox}}{\isaliteral{5C3C414C3E}{\isasymAL}}\isaliteral{5C3C5E627375623E}{}\isactrlbsub bn\isaliteral{5C3C5E657375623E}{}\isactrlesub \ x\isaliteral{5C3C5E697375623E}{}\isactrlisub {\isadigit{1}}} holds by the second part,
  we can infer that

  \[
  \isa{Let\isaliteral{5C3C5E7375703E}{}\isactrlsup {\isaliteral{5C3C616C7068613E}{\isasymalpha}}\ {\isaliteral{28}{\isacharparenleft}}{\isaliteral{5C3C70693E}{\isasympi}}\ {\isaliteral{5C3C62756C6C65743E}{\isasymbullet}}{\isaliteral{5C3C414C3E}{\isasymAL}}\isaliteral{5C3C5E627375623E}{}\isactrlbsub bn\isaliteral{5C3C5E657375623E}{}\isactrlesub \ x\isaliteral{5C3C5E697375623E}{}\isactrlisub {\isadigit{1}}{\isaliteral{29}{\isacharparenright}}\ {\isaliteral{28}{\isacharparenleft}}{\isaliteral{5C3C70693E}{\isasympi}}\ {\isaliteral{5C3C62756C6C65743E}{\isasymbullet}}\ x\isaliteral{5C3C5E697375623E}{}\isactrlisub {\isadigit{2}}{\isaliteral{29}{\isacharparenright}}\ {\isaliteral{3D}{\isacharequal}}\ Let\isaliteral{5C3C5E7375703E}{}\isactrlsup {\isaliteral{5C3C616C7068613E}{\isasymalpha}}\ x\isaliteral{5C3C5E697375623E}{}\isactrlisub {\isadigit{1}}\ x\isaliteral{5C3C5E697375623E}{}\isactrlisub {\isadigit{2}}}  
  \]\smallskip

  \noindent
  holds. This allows us to use the implication from the strong cases
  lemma, and we are done.

  Consequently,  we can discharge all proof-obligations about having `covered all
  cases'. This completes the proof establishing that the weak induction principles imply 
  the strong induction principles. These strong induction principles have already proved 
  being very useful in practice, particularly for proving properties about 
  capture-avoiding substitution \cite{Urban08}.%
\end{isamarkuptext}%
\isamarkuptrue%
\isamarkupsection{Related Work\label{related}%
}
\isamarkuptrue%
\begin{isamarkuptext}%
To our knowledge the earliest usage of general binders in a theorem prover
  is described by Nara\-schew\-ski and Nipkow \cite{NaraschewskiNipkow99} with a
  formalisation of the algorithm W. This formalisation implements binding in
  type-schemes using a de-Bruijn indices representation. Since type-schemes in
  W contain only a single place where variables are bound, different indices
  do not refer to different binders (as in the usual de-Bruijn
  representation), but to different bound variables. A similar idea has been
  recently explored for general binders by Chargu\'eraud \cite{chargueraud09}
  in the locally nameless approach to
  binding.  There, de-Bruijn indices consist of two
  numbers, one referring to the place where a variable is bound, and the other
  to which variable is bound. The reasoning infrastructure for both
  representations of bindings comes for free in theorem provers like
  Isabelle/HOL and Coq, since the corresponding term-calculi can be implemented
  as `normal' datatypes.  However, in both approaches it seems difficult to
  achieve our fine-grained control over the `semantics' of bindings
  (i.e.~whether the order of binders should matter, or vacuous binders should
  be taken into account). To do so, one would require additional predicates
  that filter out unwanted terms. Our guess is that such predicates result in
  rather intricate formal reasoning. We are not aware of any formalisation of 
  a non-trivial language that uses Chargu\'eraud's idea.

  Another technique for representing binding is higher-order abstract syntax
  (HOAS), which for example is implemented in the Twelf system \cite{pfenningsystem}. 
  This representation technique supports very elegantly many aspects of
  \emph{single} binding, and impressive work by Lee et al~\cite{LeeCraryHarper07} 
  has been done that uses HOAS for mechanising the metatheory of SML. We
  are, however, not aware how multiple binders of SML are represented in this
  work. Judging from the submitted Twelf-solution for the POPLmark challenge,
  HOAS cannot easily deal with binding constructs where the number of bound
  variables is not fixed. For example, in the second part of this challenge,
  \isa{Let}s involve patterns that bind multiple variables at once. In
  such situations, HOAS seems to have to resort to the
  iterated-single-binders-approach with all the unwanted consequences when
  reasoning about the resulting terms.

  Two formalisations involving general binders have been 
  performed in older
  versions of Nominal Isabelle (one about Psi-calculi and one about algorithm W 
  \cite{BengtsonParow09,UrbanNipkow09}).  Both
  use the approach based on iterated single binders. Our experience with
  the latter formalisation has been disappointing. The major pain arose from
  the need to `unbind' bound variables and the resulting formal reasoning turned out to
  be rather unpleasant. In contrast, the unbinding can be 
  done in one step with our
  general binders described in this paper.

  The most closely related work to the one presented here is the Ott-tool by
  Sewell et al \cite{ott-jfp} and the C$\alpha$ml language by Pottier
  \cite{Pottier06}. Ott is a nifty front-end for creating \LaTeX{} documents
  from specifications of term-calculi involving general binders. For a subset
  of the specifications Ott can also generate theorem prover code using a `raw'
  representation of terms, and in Coq also a locally nameless
  representation. The developers of this tool have also put forward (on paper)
  a definition for alpha-equivalence and free variables for terms that can be
  specified in Ott.  This definition is rather different from ours, not using
  any nominal techniques.  To our knowledge there is no concrete mathematical
  result concerning this notion of alpha-equivalence and free variables. We
  have proved that our definitions lead to alpha-equated terms, whose support
  is as expected (that means bound atoms are removed from the support). We
  also showed that our specifications lift from `raw' terms to 
  alpha-equivalence classes. For this we have established (automatically) that every
  term-constructor and function defined for `raw' terms 
  is respectful w.r.t.~alpha-equivalence.

  Although we were heavily inspired by the syntax of Ott, its definition of
  alpha-equi\-valence is unsuitable for our extension of Nominal
  Isabelle. First, it is far too complicated to be a basis for automated
  proofs implemented on the ML-level of Isabelle/HOL. Second, it covers cases
  of binders depending on other binders, which just do not make sense for our
  alpha-equated terms (the corresponding \isa{fa}-functions would not lift). 
  Third, it allows empty types that have no meaning in a
  HOL-based theorem prover. We also had to generalise slightly Ott's binding
  clauses. In Ott one specifies binding clauses with a single body; we allow
  more than one. We have to do this, because this makes a difference for our
  notion of alpha-equivalence in case of \isacommand{binds (set)} and
  \isacommand{binds (set+)}. Consider the examples
  
  \[\mbox{
  \begin{tabular}{@ {}l@ {\hspace{2mm}}l@ {}}
  \isa{Foo\isaliteral{5C3C5E697375623E}{}\isactrlisub {\isadigit{1}}\ xs{\isaliteral{3A}{\isacharcolon}}{\isaliteral{3A}{\isacharcolon}}name\ fset\ t{\isaliteral{3A}{\isacharcolon}}{\isaliteral{3A}{\isacharcolon}}trm\ s{\isaliteral{3A}{\isacharcolon}}{\isaliteral{3A}{\isacharcolon}}trm} &  
      \isacommand{binds (set)} \isa{xs} \isacommand{in} \isa{t\ s}\\
  \isa{Foo\isaliteral{5C3C5E697375623E}{}\isactrlisub {\isadigit{2}}\ xs{\isaliteral{3A}{\isacharcolon}}{\isaliteral{3A}{\isacharcolon}}name\ fset\ t{\isaliteral{3A}{\isacharcolon}}{\isaliteral{3A}{\isacharcolon}}trm\ s{\isaliteral{3A}{\isacharcolon}}{\isaliteral{3A}{\isacharcolon}}trm} &  
      \isacommand{binds (set)} \isa{xs} \isacommand{in} \isa{t}, 
      \isacommand{binds (set)} \isa{xs} \isacommand{in} \isa{s}\\
  \end{tabular}}
  \]\smallskip
  
  \noindent
  In the first term-constructor we have a single body that happens to be
  `spread' over two arguments; in the second term-constructor we have two
  independent bodies in which the same variables are bound. As a result we
  have\footnote{Assuming \isa{a\ {\isaliteral{5C3C6E6F7465713E}{\isasymnoteq}}\ b}, there is no permutation that can
  make \isa{{\isaliteral{28}{\isacharparenleft}}a{\isaliteral{2C}{\isacharcomma}}\ b{\isaliteral{29}{\isacharparenright}}} equal with both \isa{{\isaliteral{28}{\isacharparenleft}}a{\isaliteral{2C}{\isacharcomma}}\ b{\isaliteral{29}{\isacharparenright}}} and \isa{{\isaliteral{28}{\isacharparenleft}}b{\isaliteral{2C}{\isacharcomma}}\ a{\isaliteral{29}{\isacharparenright}}}, but
  there are two permutations so that we can make \isa{{\isaliteral{28}{\isacharparenleft}}a{\isaliteral{2C}{\isacharcomma}}\ b{\isaliteral{29}{\isacharparenright}}} and \isa{{\isaliteral{28}{\isacharparenleft}}a{\isaliteral{2C}{\isacharcomma}}\ b{\isaliteral{29}{\isacharparenright}}} equal with one permutation, and \isa{{\isaliteral{28}{\isacharparenleft}}a{\isaliteral{2C}{\isacharcomma}}\ b{\isaliteral{29}{\isacharparenright}}} and \isa{{\isaliteral{28}{\isacharparenleft}}b{\isaliteral{2C}{\isacharcomma}}\ a{\isaliteral{29}{\isacharparenright}}} with the other.}

  \[\mbox{
  \begin{tabular}{r@ {\hspace{1.5mm}}c@ {\hspace{1.5mm}}l}
  \isa{Foo\isaliteral{5C3C5E697375623E}{}\isactrlisub {\isadigit{1}}\ {\isaliteral{7B}{\isacharbraceleft}}a{\isaliteral{2C}{\isacharcomma}}\ b{\isaliteral{7D}{\isacharbraceright}}\ {\isaliteral{28}{\isacharparenleft}}a{\isaliteral{2C}{\isacharcomma}}\ b{\isaliteral{29}{\isacharparenright}}\ {\isaliteral{28}{\isacharparenleft}}a{\isaliteral{2C}{\isacharcomma}}\ b{\isaliteral{29}{\isacharparenright}}} & $\not=$ & 
  \isa{Foo\isaliteral{5C3C5E697375623E}{}\isactrlisub {\isadigit{1}}\ {\isaliteral{7B}{\isacharbraceleft}}a{\isaliteral{2C}{\isacharcomma}}\ b{\isaliteral{7D}{\isacharbraceright}}\ {\isaliteral{28}{\isacharparenleft}}a{\isaliteral{2C}{\isacharcomma}}\ b{\isaliteral{29}{\isacharparenright}}\ {\isaliteral{28}{\isacharparenleft}}b{\isaliteral{2C}{\isacharcomma}}\ a{\isaliteral{29}{\isacharparenright}}}
  \end{tabular}}
  \]\smallskip
 
  \noindent
  but 

  \[\mbox{
  \begin{tabular}{r@ {\hspace{1.5mm}}c@ {\hspace{1.5mm}}l}
  \isa{Foo\isaliteral{5C3C5E697375623E}{}\isactrlisub {\isadigit{2}}\ {\isaliteral{7B}{\isacharbraceleft}}a{\isaliteral{2C}{\isacharcomma}}\ b{\isaliteral{7D}{\isacharbraceright}}\ {\isaliteral{28}{\isacharparenleft}}a{\isaliteral{2C}{\isacharcomma}}\ b{\isaliteral{29}{\isacharparenright}}\ {\isaliteral{28}{\isacharparenleft}}a{\isaliteral{2C}{\isacharcomma}}\ b{\isaliteral{29}{\isacharparenright}}} & $=$ & 
  \isa{Foo\isaliteral{5C3C5E697375623E}{}\isactrlisub {\isadigit{2}}\ {\isaliteral{7B}{\isacharbraceleft}}a{\isaliteral{2C}{\isacharcomma}}\ b{\isaliteral{7D}{\isacharbraceright}}\ {\isaliteral{28}{\isacharparenleft}}a{\isaliteral{2C}{\isacharcomma}}\ b{\isaliteral{29}{\isacharparenright}}\ {\isaliteral{28}{\isacharparenleft}}b{\isaliteral{2C}{\isacharcomma}}\ a{\isaliteral{29}{\isacharparenright}}}\\
  \end{tabular}}
  \]\smallskip
  
  \noindent
  and therefore need the extra generality to be able to distinguish
  between both specifications.  Because of how we set up our
  definitions, we also had to impose some restrictions (like a single
  binding function for a deep binder) that are not present in Ott. Our
  expectation is that we can still cover many interesting term-calculi
  from programming language research, for example the Core-Haskell
  language from the Introduction. With the work presented in this
  paper we can define it formally as shown in
  Figure~\ref{nominalcorehas} and then Nominal Isabelle derives
  automatically a corresponding reasoning infrastructure. However we
  have found out that telescopes seem to not easily be representable
  in our framework.  The reason is that we need to be able to lift our
  \isa{bn}-functions to alpha-equated lambda-terms and therefore
  need to restrict what these \isa{bn}-functions can return.
  Telescopes can be represented in the framework described in
  \cite{WeirichYorgeySheard11} using an extension of the usual
  locally-nameless representation. 

  \begin{figure}[p]
  \begin{boxedminipage}{\linewidth}
  \small
  \begin{tabular}{l}
  \isacommand{atom\_decl}~\isa{var\ cvar\ tvar}\\[1mm]
  \isacommand{nominal\_datatype}~\isa{tkind\ {\isaliteral{3D}{\isacharequal}}}~\isa{KStar}~$|$~\isa{KFun\ tkind\ tkind}\\ 
  \isacommand{and}~\isa{ckind\ {\isaliteral{3D}{\isacharequal}}}~\isa{CKSim\ ty\ ty}\\
  \isacommand{and}~\isa{ty\ {\isaliteral{3D}{\isacharequal}}}~\isa{TVar\ tvar}~$|$~\isa{T\ string}~$|$~\isa{TApp\ ty\ ty}\\
  $|$~\isa{TFun\ string\ ty{\isaliteral{5F}{\isacharunderscore}}list}~%
  $|$~\isa{TAll\ tv{\isaliteral{3A}{\isacharcolon}}{\isaliteral{3A}{\isacharcolon}}tvar\ tkind\ ty{\isaliteral{3A}{\isacharcolon}}{\isaliteral{3A}{\isacharcolon}}ty}\hspace{3mm}\isacommand{binds}~\isa{tv}~\isacommand{in}~\isa{ty}\\
  $|$~\isa{TArr\ ckind\ ty}\\
  \isacommand{and}~\isa{ty{\isaliteral{5F}{\isacharunderscore}}lst\ {\isaliteral{3D}{\isacharequal}}}~\isa{TNil}~$|$~\isa{TCons\ ty\ ty{\isaliteral{5F}{\isacharunderscore}}lst}\\
  \isacommand{and}~\isa{cty\ {\isaliteral{3D}{\isacharequal}}}~\isa{CVar\ cvar}~%
  $|$~\isa{C\ string}~$|$~\isa{CApp\ cty\ cty}~$|$~\isa{CFun\ string\ co{\isaliteral{5F}{\isacharunderscore}}lst}\\
  $|$~\isa{CAll\ cv{\isaliteral{3A}{\isacharcolon}}{\isaliteral{3A}{\isacharcolon}}cvar\ ckind\ cty{\isaliteral{3A}{\isacharcolon}}{\isaliteral{3A}{\isacharcolon}}cty}\hspace{3mm}\isacommand{binds}~\isa{cv}~\isacommand{in}~\isa{cty}\\
  $|$~\isa{CArr\ ckind\ cty}~$|$~\isa{CRefl\ ty}~$|$~\isa{CSym\ cty}~$|$~\isa{CCirc\ cty\ cty}\\
  $|$~\isa{CAt\ cty\ ty}~$|$~\isa{CLeft\ cty}~$|$~\isa{CRight\ cty}~$|$~\isa{CSim\ cty\ cty}\\
  $|$~\isa{CRightc\ cty}~$|$~\isa{CLeftc\ cty}~$|$~\isa{Coerce\ cty\ cty}\\
  \isacommand{and}~\isa{co{\isaliteral{5F}{\isacharunderscore}}lst\ {\isaliteral{3D}{\isacharequal}}}~\isa{CNil}~$|$~\isa{CCons\ cty\ co{\isaliteral{5F}{\isacharunderscore}}lst}\\
  \isacommand{and}~\isa{trm\ {\isaliteral{3D}{\isacharequal}}}~\isa{Var\ var}~$|$~\isa{K\ string}\\
  $|$~\isa{LAM{\isaliteral{5F}{\isacharunderscore}}ty\ tv{\isaliteral{3A}{\isacharcolon}}{\isaliteral{3A}{\isacharcolon}}tvar\ tkind\ t{\isaliteral{3A}{\isacharcolon}}{\isaliteral{3A}{\isacharcolon}}trm}\hspace{3mm}\isacommand{binds}~\isa{tv}~\isacommand{in}~\isa{t}\\
  $|$~\isa{LAM{\isaliteral{5F}{\isacharunderscore}}cty\ cv{\isaliteral{3A}{\isacharcolon}}{\isaliteral{3A}{\isacharcolon}}cvar\ ckind\ t{\isaliteral{3A}{\isacharcolon}}{\isaliteral{3A}{\isacharcolon}}trm}\hspace{3mm}\isacommand{binds}~\isa{cv}~\isacommand{in}~\isa{t}\\
  $|$~\isa{App{\isaliteral{5F}{\isacharunderscore}}ty\ trm\ ty}~$|$~\isa{App{\isaliteral{5F}{\isacharunderscore}}cty\ trm\ cty}~$|$~\isa{App\ trm\ trm}\\
  $|$~\isa{Lam\ v{\isaliteral{3A}{\isacharcolon}}{\isaliteral{3A}{\isacharcolon}}var\ ty\ t{\isaliteral{3A}{\isacharcolon}}{\isaliteral{3A}{\isacharcolon}}trm}\hspace{3mm}\isacommand{binds}~\isa{v}~\isacommand{in}~\isa{t}\\
  $|$~\isa{Let\ x{\isaliteral{3A}{\isacharcolon}}{\isaliteral{3A}{\isacharcolon}}var\ ty\ trm\ t{\isaliteral{3A}{\isacharcolon}}{\isaliteral{3A}{\isacharcolon}}trm}\hspace{3mm}\isacommand{binds}~\isa{x}~\isacommand{in}~\isa{t}\\
  $|$~\isa{Case\ trm\ assoc{\isaliteral{5F}{\isacharunderscore}}lst}~$|$~\isa{Cast\ trm\ co}\\
  \isacommand{and}~\isa{assoc{\isaliteral{5F}{\isacharunderscore}}lst\ {\isaliteral{3D}{\isacharequal}}}~\isa{ANil}~%
  $|$~\isa{ACons\ p{\isaliteral{3A}{\isacharcolon}}{\isaliteral{3A}{\isacharcolon}}pat\ t{\isaliteral{3A}{\isacharcolon}}{\isaliteral{3A}{\isacharcolon}}trm\ assoc{\isaliteral{5F}{\isacharunderscore}}lst}\hspace{3mm}\isacommand{binds}~\isa{bv\ p}~\isacommand{in}~\isa{t}\\
  \isacommand{and}~\isa{pat\ {\isaliteral{3D}{\isacharequal}}}~\isa{Kpat\ string\ tvtk{\isaliteral{5F}{\isacharunderscore}}lst\ tvck{\isaliteral{5F}{\isacharunderscore}}lst\ vt{\isaliteral{5F}{\isacharunderscore}}lst}\\
  \isacommand{and}~\isa{vt{\isaliteral{5F}{\isacharunderscore}}lst\ {\isaliteral{3D}{\isacharequal}}}~\isa{VTNil}~$|$~\isa{VTCons\ var\ ty\ vt{\isaliteral{5F}{\isacharunderscore}}lst}\\
  \isacommand{and}~\isa{tvtk{\isaliteral{5F}{\isacharunderscore}}lst\ {\isaliteral{3D}{\isacharequal}}}~\isa{TVTKNil}~$|$~\isa{TVTKCons\ tvar\ tkind\ tvtk{\isaliteral{5F}{\isacharunderscore}}lst}\\
  \isacommand{and}~\isa{tvck{\isaliteral{5F}{\isacharunderscore}}lst\ {\isaliteral{3D}{\isacharequal}}}~\isa{TVCKNil}~$|$ \isa{TVCKCons\ cvar\ ckind\ tvck{\isaliteral{5F}{\isacharunderscore}}lst}\\
  \isacommand{binder}\\
  \;\isa{bv\ {\isaliteral{3A}{\isacharcolon}}{\isaliteral{3A}{\isacharcolon}}\ pat\ {\isaliteral{5C3C52696768746172726F773E}{\isasymRightarrow}}\ atom\ list}~\isacommand{and}\\
  \;\isa{bv\isaliteral{5C3C5E697375623E}{}\isactrlisub {\isadigit{1}}\ {\isaliteral{3A}{\isacharcolon}}{\isaliteral{3A}{\isacharcolon}}\ vt{\isaliteral{5F}{\isacharunderscore}}lst\ {\isaliteral{5C3C52696768746172726F773E}{\isasymRightarrow}}\ atom\ list}~\isacommand{and}\\
  \;\isa{bv\isaliteral{5C3C5E697375623E}{}\isactrlisub {\isadigit{2}}\ {\isaliteral{3A}{\isacharcolon}}{\isaliteral{3A}{\isacharcolon}}\ tvtk{\isaliteral{5F}{\isacharunderscore}}lst\ {\isaliteral{5C3C52696768746172726F773E}{\isasymRightarrow}}\ atom\ list}~\isacommand{and}\\
  \;\isa{bv\isaliteral{5C3C5E697375623E}{}\isactrlisub {\isadigit{3}}\ {\isaliteral{3A}{\isacharcolon}}{\isaliteral{3A}{\isacharcolon}}\ tvck{\isaliteral{5F}{\isacharunderscore}}lst\ {\isaliteral{5C3C52696768746172726F773E}{\isasymRightarrow}}\ atom\ list}\\
  \isacommand{where}\\
  \phantom{$|$}~\isa{bv\ {\isaliteral{28}{\isacharparenleft}}K\ s\ tvts\ tvcs\ vs{\isaliteral{29}{\isacharparenright}}\ {\isaliteral{3D}{\isacharequal}}\ {\isaliteral{28}{\isacharparenleft}}bv\isaliteral{5C3C5E697375623E}{}\isactrlisub {\isadigit{3}}\ tvts{\isaliteral{29}{\isacharparenright}}\ {\isaliteral{40}{\isacharat}}\ {\isaliteral{28}{\isacharparenleft}}bv\isaliteral{5C3C5E697375623E}{}\isactrlisub {\isadigit{2}}\ tvcs{\isaliteral{29}{\isacharparenright}}\ {\isaliteral{40}{\isacharat}}\ {\isaliteral{28}{\isacharparenleft}}bv\isaliteral{5C3C5E697375623E}{}\isactrlisub {\isadigit{1}}\ vs{\isaliteral{29}{\isacharparenright}}}\\
  $|$~\isa{bv\isaliteral{5C3C5E697375623E}{}\isactrlisub {\isadigit{1}}\ VTNil\ {\isaliteral{3D}{\isacharequal}}\ {\isaliteral{5B}{\isacharbrackleft}}{\isaliteral{5D}{\isacharbrackright}}}\\
  $|$~\isa{bv\isaliteral{5C3C5E697375623E}{}\isactrlisub {\isadigit{1}}\ {\isaliteral{28}{\isacharparenleft}}VTCons\ x\ ty\ tl{\isaliteral{29}{\isacharparenright}}\ {\isaliteral{3D}{\isacharequal}}\ {\isaliteral{28}{\isacharparenleft}}atom\ x{\isaliteral{29}{\isacharparenright}}{\isaliteral{3A}{\isacharcolon}}{\isaliteral{3A}{\isacharcolon}}{\isaliteral{28}{\isacharparenleft}}bv\isaliteral{5C3C5E697375623E}{}\isactrlisub {\isadigit{1}}\ tl{\isaliteral{29}{\isacharparenright}}}\\
  $|$~\isa{bv\isaliteral{5C3C5E697375623E}{}\isactrlisub {\isadigit{2}}\ TVTKNil\ {\isaliteral{3D}{\isacharequal}}\ {\isaliteral{5B}{\isacharbrackleft}}{\isaliteral{5D}{\isacharbrackright}}}\\
  $|$~\isa{bv\isaliteral{5C3C5E697375623E}{}\isactrlisub {\isadigit{2}}\ {\isaliteral{28}{\isacharparenleft}}TVTKCons\ a\ ty\ tl{\isaliteral{29}{\isacharparenright}}\ {\isaliteral{3D}{\isacharequal}}\ {\isaliteral{28}{\isacharparenleft}}atom\ a{\isaliteral{29}{\isacharparenright}}{\isaliteral{3A}{\isacharcolon}}{\isaliteral{3A}{\isacharcolon}}{\isaliteral{28}{\isacharparenleft}}bv\isaliteral{5C3C5E697375623E}{}\isactrlisub {\isadigit{2}}\ tl{\isaliteral{29}{\isacharparenright}}}\\
  $|$~\isa{bv\isaliteral{5C3C5E697375623E}{}\isactrlisub {\isadigit{3}}\ TVCKNil\ {\isaliteral{3D}{\isacharequal}}\ {\isaliteral{5B}{\isacharbrackleft}}{\isaliteral{5D}{\isacharbrackright}}}\\
  $|$~\isa{bv\isaliteral{5C3C5E697375623E}{}\isactrlisub {\isadigit{3}}\ {\isaliteral{28}{\isacharparenleft}}TVCKCons\ c\ cty\ tl{\isaliteral{29}{\isacharparenright}}\ {\isaliteral{3D}{\isacharequal}}\ {\isaliteral{28}{\isacharparenleft}}atom\ c{\isaliteral{29}{\isacharparenright}}{\isaliteral{3A}{\isacharcolon}}{\isaliteral{3A}{\isacharcolon}}{\isaliteral{28}{\isacharparenleft}}bv\isaliteral{5C3C5E697375623E}{}\isactrlisub {\isadigit{3}}\ tl{\isaliteral{29}{\isacharparenright}}}\\
  \end{tabular}
  \end{boxedminipage}
  \caption{A definition for Core-Haskell in Nominal Isabelle. For the moment we
  do not support nested types; therefore we explicitly have to unfold the 
  lists \isa{co{\isaliteral{5F}{\isacharunderscore}}lst}, \isa{assoc{\isaliteral{5F}{\isacharunderscore}}lst} and so on. Apart from that limitation, the 
  definition follows closely the original shown in Figure~\ref{corehas}. The
  point of our work is that having made such a definition in Nominal Isabelle,
  one obtains automatically a reasoning infrastructure for Core-Haskell.
  \label{nominalcorehas}}
  \end{figure}
  \afterpage{\clearpage}

  Pottier presents a programming language, called C$\alpha$ml, for
  representing terms with general binders inside OCaml \cite{Pottier06}. This
  language is implemented as a front-end that can be translated to OCaml with
  the help of a library. He presents a type-system in which the scope of
  general binders can be specified using special markers, written \isa{inner} and \isa{outer}. It seems our and his specifications can be
  inter-translated as long as ours use the binding mode \isacommand{binds}
  only.  However, we have not proved this. Pottier gives a definition for
  alpha-equivalence, which also uses a permutation operation (like ours).
  Still, this definition is rather different from ours and he only proves that
  it defines an equivalence relation. A complete reasoning infrastructure is
  well beyond the purposes of his language. Similar work for Haskell with
  similar results was reported by Cheney \cite{Cheney05a} and more recently 
  by Weirich et al \cite{WeirichYorgeySheard11}.

  In a slightly different domain (programming with dependent types),
  Altenkirch et al \cite{Altenkirch10} present a calculus with a notion of
  alpha-equivalence related to our binding mode \isacommand{binds (set+)}.
  Their definition is similar to the one by Pottier, except that it has a more
  operational flavour and calculates a partial (renaming) map. In this way,
  the definition can deal with vacuous binders. However, to our best
  knowledge, no concrete mathematical result concerning this definition of
  alpha-equivalence has been proved.%
\end{isamarkuptext}%
\isamarkuptrue%
\isamarkupsection{Conclusion%
}
\isamarkuptrue%
\begin{isamarkuptext}%
We have presented an extension of Nominal Isabelle for dealing with general
  binders, that is where term-constructors have multiple bound atoms. For this
  extension we introduced new definitions of alpha-equivalence and automated
  all necessary proofs in Isabelle/HOL.  To specify general binders we used
  the syntax from Ott, but extended it in some places and restricted
  it in others so that the definitions make sense in the context of alpha-equated
  terms. We also introduced two binding modes (set and set+) that do not exist
  in Ott. We have tried out the extension with calculi such as Core-Haskell,
  type-schemes and approximately a dozen of other typical examples from
  programming language research~\cite{SewellBestiary}. The code will
  eventually become part of the Isabelle distribution.\footnote{It 
  can be downloaded already from \href{http://isabelle.in.tum.de/nominal/download}
  {http://isabelle.in.tum.de/nominal/download}.}

  We have left out a discussion about how functions can be defined over
  alpha-equated terms involving general binders. In earlier versions of
  Nominal Isabelle this turned out to be a thorny issue.  We hope to do better
  this time by using the function package \cite{Krauss09} that has recently
  been implemented in Isabelle/HOL and also by restricting function
  definitions to equivariant functions (for them we can provide more
  automation).

  There are some restrictions we had
  to impose in this paper that can be lifted using 
  a recent reimplementation \cite{Traytel12} of the datatype package for Isabelle/HOL, which
  however is not yet part of the stable distribution.
  This reimplementation allows nested
  datatype definitions and would allow one to specify, for instance, the function kinds
  in Core-Haskell as \isa{TFun\ string\ {\isaliteral{28}{\isacharparenleft}}ty\ list{\isaliteral{29}{\isacharparenright}}} instead of the unfolded
  version \isa{TFun\ string\ ty{\isaliteral{5F}{\isacharunderscore}}list} (see Figure~\ref{nominalcorehas}). We can 
  also use it to represent the \isa{Let}-terms from the Introduction where
  the order of \isa{let}-assignments does not matter. This means we can represent \isa{Let}s
  such that the following two terms are equal

  \[
  \isa{Let\ x\isaliteral{5C3C5E697375623E}{}\isactrlisub {\isadigit{1}}\ {\isaliteral{3D}{\isacharequal}}\ t\isaliteral{5C3C5E697375623E}{}\isactrlisub {\isadigit{1}}\ and\ x\isaliteral{5C3C5E697375623E}{}\isactrlisub {\isadigit{2}}\ {\isaliteral{3D}{\isacharequal}}\ t\isaliteral{5C3C5E697375623E}{}\isactrlisub {\isadigit{2}}\ in\ s} \;\;=\;\;
  \isa{Let\ x\isaliteral{5C3C5E697375623E}{}\isactrlisub {\isadigit{2}}\ {\isaliteral{3D}{\isacharequal}}\ t\isaliteral{5C3C5E697375623E}{}\isactrlisub {\isadigit{2}}\ and\ x\isaliteral{5C3C5E697375623E}{}\isactrlisub {\isadigit{1}}\ {\isaliteral{3D}{\isacharequal}}\ t\isaliteral{5C3C5E697375623E}{}\isactrlisub {\isadigit{1}}\ in\ s} 
  \]\smallskip

  \noindent
  For this we have to represent the \isa{Let}-assignments as finite sets
  of pair and a binding function that picks out the left components to be bound in \isa{s}.

  One line of future investigation is whether we can go beyond the 
  simple-minded form of binding functions that we adopted from Ott. At the moment, binding
  functions can only return the empty set, a singleton atom set or unions
  of atom sets (similarly for lists). It remains to be seen whether 
  properties like
  
  \[
  \mbox{\isa{fa{\isaliteral{5F}{\isacharunderscore}}ty\ x\ \ {\isaliteral{3D}{\isacharequal}}\ \ bn\ x\ {\isaliteral{5C3C756E696F6E3E}{\isasymunion}}\ fa{\isaliteral{5F}{\isacharunderscore}}bn\ x}}
  \]\smallskip
  
  \noindent
  allow us to support more interesting binding functions. 
  
  We have also not yet played with other binding modes. For example we can
  imagine that there is need for a binding mode where instead of usual lists,
  we abstract lists of distinct elements (the corresponding type \isa{dlist} already exists in the library of Isabelle/HOL). We expect the
  presented work can be extended to accommodate such binding modes.\medskip
  
  \noindent
  {\bf Acknowledgements:} We are very grateful to Andrew Pitts for many
  discussions about Nominal Isabelle. We thank Peter Sewell for making the
  informal notes \cite{SewellBestiary} available to us and also for patiently
  explaining some of the finer points of the Ott-tool.  Stephanie Weirich
  suggested to separate the subgrammars of kinds and types in our Core-Haskell
  example. Ramana Kumar and Andrei Popescu helped us with comments for
  an earlier version of this paper.%
\end{isamarkuptext}%
\isamarkuptrue%
\isadelimtheory
\endisadelimtheory
\isatagtheory
\endisatagtheory
{\isafoldtheory}%
\isadelimtheory
\endisadelimtheory
\end{isabellebody}%


\bibliographystyle{plain}
\bibliography{root}


\end{document}